\documentclass[11pt]{article}
\usepackage{arxiv}
\usepackage[utf8]{inputenc}
\usepackage[T1]{fontenc}
\usepackage{graphicx}
\usepackage{amsmath,amsfonts,amsthm, amssymb, dsfont} 
\usepackage{pdflscape}
\usepackage{afterpage}
\usepackage{setspace}
\usepackage{color, colortbl}
\usepackage[natbibapa]{apacite}
\usepackage[shortlabels]{enumitem}
\usepackage{float,caption}
\usepackage{subcaption}
\usepackage{booktabs}
\usepackage{threeparttable}
\usepackage[hidelinks]{hyperref}
\usepackage{tikz}
\RequirePackage{doi}
\usetikzlibrary{patterns}

\floatstyle{boxed}
\newfloat{procedure}{thp}{lop}
\floatname{procedure}{Procedure}

\definecolor{Gray}{gray}{0.9}

\newcommand{\E}{\mathbb{E}}
\newcommand{\Var}{\mathrm{Var}}
\newcommand{\Cov}{\mathrm{Cov}}

\newcommand{\muError}{\hat{\mu}_{S_{-i}}(1,X_t) - \mu_0(1,X_t)}
\newcommand{\muErrorTilde}{\mu(1,X_t) - \mu_0(1,X_t)}
\newcommand{\invpError}{\frac{1}{\hat{e}_{S_{-i}}(X_t)} -\frac{1}{e_0(X_t)}}

\newcommand{\norm}[1]{\left\lVert#1\right\rVert}

\newtheorem{assumption}{Assumption}

\newtheorem{theorem}{Theorem}
\newtheorem{lemma}{Lemma}

\title{Semiparametric inference for impulse response functions using double/debiased machine learning}

\author{Daniele~Ballinari\thanks{The views, opinions, findings, and conclusions or recommendations expressed in this paper are strictly those of the authors. They do not necessarily reflect the views of the Swiss National Bank (SNB). The SNB takes no responsibility for any errors or omissions in, or for the correctness of, the information contained in this paper.} \thanks{We thank Nora Bearth, Jonathan Chassot and Victor Chernozhukov for helpful comments and suggestions. We are also grateful to Guido Kuersteiner for providing the data used in the empirical application.}\\
	Swiss National Bank\\
	\texttt{daniele.ballinari@snb.ch} \\
	\And
	Alexander~Wehrli\footnotemark[1]\footnotemark[2]\\
	Swiss National Bank\\
	\texttt{alexander.wehrli@snb.ch} \\
}
\date{\today\\[0.1cm]First version: November 18, 2024}

\renewcommand{\shorttitle}{}
\renewcommand{\headeright}{D. Ballinari \& A. Wehrli}
\renewcommand{\undertitle}{}

\begin{document}

\maketitle

\begin{abstract}
We introduce a double/debiased machine learning estimator for the impulse response function in settings where a time series of interest is subjected to multiple discrete treatments, assigned over time, which can have a causal effect on future outcomes. The proposed estimator can rely on fully nonparametric relations between treatment and outcome variables, opening up the possibility to use flexible machine learning approaches to estimate impulse response functions. To this end, we extend the theory of double machine learning from an \emph{i.i.d.} to a time series setting and show that the proposed estimator is consistent and asymptotically normally distributed at the parametric rate, allowing for semiparametric inference for dynamic effects in a time series setting. The properties of the estimator are validated numerically in finite samples by applying it to learn the impulse response function in the presence of serial dependence in both the confounder and observation innovation processes. We also illustrate the methodology empirically by applying it to the estimation of the effects of macroeconomic shocks.
\end{abstract}

\jelcodes{C14 \and C22 \and C51 \and C53 \and C55}

\section{Introduction}
The estimation of the response of a time series to an external impulse is a common task in many scientific disciplines. For example, in economics, one might be interested in the reaction of the economy to a change in a central bank's monetary policy \citep{Angrist2018}. In the analysis of their trading costs, financial professionals are interested in the causal effect that their trades have on an asset's price \citep{Bouchaud2018}. In medicine, when administering a drug to a patient over time, one is interested in measuring its effect on the health of the patient \citep{Bica2020}. Readers are referred to the surveys in \cite{Runge2023} and \cite{Moraffah2021} for more examples.

The quantity of interest in these applications is commonly referred to as the \emph{impulse response function} (IRF). Ideally, the IRF measures the \emph{causal} effect that an action (or ``treatment'') has on the time series of interest. Recently, IRFs and ideas stemming from the causal inference framework have been related \citep{Jorda2023}. In particular, \cite{Rambachan2021} provided assumptions under which IRFs coincide with classical \emph{average treatment effects} (ATE) analyzed in the potential outcomes framework of causal inference \citep{Rubin1974, Robins1986}. Given this relation between ATE and IRFs, it seams natural to adopt estimation procedures from the causal inference literature for the problem of IRF estimation. Traditionally, IRFs have primarily been estimated by modelling the entire dynamic system under consideration, e.g. using vector autoregressive processes \citep{Sims1980}. The seminal work of \cite{Jorda2005} later showed how to directly estimate univariate conditional expectations using local projections \citep{LPNBER2024}. This approach compares the conditional expectation of an outcome variable, once conditional on a shock (treatment) and once conditional on no shock. As such, the local projection framework is directly related to the regression adjustment approach used for ATE estimation. The approach of \cite{Jorda2005} allows for some flexibility in the estimation of the impulse response function, as it can easily incorporate polynomial and interaction terms of the regressors, state dependence \citep{Goncalves2024}, or instrumental variables \citep{Stock2018}. More recently, \cite{Adamek2024} extended the local projection approach to high-dimensional settings using penalized local projections. \cite{JordaAusterity2016} and \cite{Angrist2018} are examples of applications that use propensity score weighting for the estimation of IRFs, another common estimation approach coming from the causal inference literature, additionally accommodating asymmetric and nonlinear responses.

While estimators for IRFs have become more flexible in recent years, they still require the definition of a functional form relating treatment and outcome variables. Here, we introduce an estimator for the IRF that can rely on fully nonparametric relations between treatment and outcome variables, opening up the possibility to use flexible machine learning approaches to estimate IRFs. We consider a setting where a single time series is subjected to a discrete treatment at multiple points in time and one is interested in the average (causal) effect that these treatments have at different prediction horizons. Inspired by the approach of \cite{Chernozhukov2018} for the \emph{i.i.d.} setting, the proposed estimator leverages the efficient influence function for the IRF in combination with cross-fitting, which makes the IRF estimation insensitive -- or, formally, (Neyman-)orthogonal -- to the biased estimation of the conditional expectation and treatment probability functions by machine learners. Moreover, the proposed estimator avoids over-fitting by using a cross-fitting procedure. These two ingredients, orthogonality and cross-fitting, eliminate regularization and over-fitting bias to which machine learning algorithms are prone to, an approach coined \emph{double/debiased machine learning} \citep{Chernozhukov2017AER}. Our theoretical results show that the proposed IRF estimator is consistent and asymptotically normally distributed at the parametric rate, building the basis for semiparametric inference for dynamic effects in time series settings. The problem studied in this paper relates to classical semiparametric estimation techniques for dependent data, e.g. using kernel \citep{Robinson1983Kernel,LiRacine2006}, series \citep{CHEN2015447,LeeRobinson2016} or general sieve estimators \citep{Chen1998,Chen2015}. Apart from such classical approaches, our contribution relates to a nascent but growing body of literature on the application of machine learning for semi- and nonparametric causal inference in time series problems. \cite{Lewis2021} e.g. provide a method of estimating and conducting inference for dynamic treatment effects, based on a sequential regression peeling process, focusing on a fixed-length time series panel setup. \cite{Bica2020} propose an estimator for time-varying treatment effects in the presence of hidden confounders, building on a recurrent neural network. Their setting also considers a fixed-length time series observed for multiple individuals. Using panel data, \cite{Paranhos2025} employs a generalized random forest to obtain locally linear impulse response functions. \cite{Hauzenberger2025} estimate impulse response functions using Bayesian neutral networks. \cite{Grecov2021} consider a multivariate time series setting where some units become and remain treated at a specific point in time. The counterfactual outcomes are obtained from a global forecasting model based on a recurrent neural network. Others have employed flexible machine learning approaches for causal discovery in time series; see, among others, \cite{Yin2022,Bussmann2021,Nauta2019}.

The rest of the paper is organized as follows. Section \ref{sec:identification} sets up the problem. Section \ref{sec:estimation} presents the double machine learning (DML) estimator and our main theoretical results on its asymptotic properties. Proofs are relegated to the \hyperref[app:proofs]{Appendix} for legibility. Section \ref{sec:practical_implementation} offers recommendations for the practical implementation of the estimator. Section \ref{sec:simulation} validates the developed theory in a simulation study and Section \ref{sec:lp} provides a comparison to local projections. Section \ref{sec:empirical} applies the proposed estimator to policy decisions in a macrodynamic setting and Section \ref{sec:conclusion} concludes.

\section{Notation and identification}\label{sec:identification}
\subsection{Notation}
Let $S^{(h)} = \{Z_t^{(h)} : t\in \mathcal{T}\}$ be stochastic processes generated from a distribution $\mathcal{P}$ with $Z_t^{(h)} = (Y_{t+h}, X_t, D_t)$, where $Y_{t+h}$ is a scalar real-valued random variable, $D_t$ a binary treatment variable, and $X_t \in \mathcal{X} \subseteq \mathbb{R}^n$ a random vector which may contain also lagged variables, including lagged values of $Y_t$ and $D_t$. If not specified otherwise, $\mathcal{T}$ is a collection of ordered time indices (also referred to as the index set) with cardinality $|\mathcal{T}|=T$. The quantity of interest is the impulse response function at horizon $h\in \mathbb{N}_0$ for a binary impulse variable $D_t$ on the outcome variable $Y_{t+h}$, defined as \citep{Rambachan2021}
\begin{equation}\label{eq:irf}
    \theta_0^{(h)}=\E\left[\E[Y_{t+h}|D_t = 1, X_t] - \E[Y_{t+h}|D_t = 0, X_t]\right].
\end{equation}
We focus on the case where $\theta_0^{(h)}$ and the conditional first moments of the outcome and treatment random variables are time-invariant.
\begin{assumption}\label{ass:time_invariance}
    For all $s,t \in \mathcal{T}$ and $h \in \mathbb{N}_0$, the following holds.
    \begin{enumerate}
        \item \label{ass:constant_ate} The impulse response function is time-invariant, i.e. $\theta_0^{(h)} = \E[\E[Y_{t+h}|D_t = 1, X_t] - \E[Y_{t+h}|D_t = 0, X_t]] = \E[\E[Y_{s+h}|D_s = 1, X_s] - \E[Y_{s+h}| D_s = 0, X_s]]$.
        \item \label{ass:constant_nuisance} The conditional first moments of $Y_{t+h}$ and $D_t$ are time-invariant, i.e. $\mu_0(d, x, h) = \E[Y_{t+h}|D_t=d, X_t=x] = \E[Y_{s+h}|D_s=d, X_s=x]$ and $e_0(x) =\Pr(D_t=1|X_t=x) = \Pr(D_s=1|X_s=x)$.
    \end{enumerate}
\end{assumption}
While we present results for a binary treatment $D_t$, this can be generalized to multivariate discrete treatments \citep{Angrist2018}, which we will also do in our empirical application. In fact, discrete treatments can be interpreted as pairwise binary treatment comparisons, and as such, the theoretical results presented in the sequel extend directly to multivariate treatments. Moreover, even in settings where the treatment variable is continuous, one can obtain estimation results by discretizing the treatment of interest \citep[see, e.g.][]{Knaus2021}. We refer to $\Gamma_0 = (\mu_0(d, x, h), e_0(x))$ as \emph{nuisance functions}.

Much of traditional estimation of IRFs relies on regression adjustment, i.e., the estimation of $\theta_0^{(h)}$ as the average difference between $\mu_0(1,X_t,h)$ and $\mu_0(0,X_t,h)$ \citep{Cochran1968,Robins1986,Pearl2009,Jorda2005}.
However, regression adjustment estimators typically tend to be rather sensitive to small amounts of misspecification in the conditional expectation models. Alternatively, approaches using inverse probability weighting \citep{Rosenbaum1983, Tsiatis2006,Angrist2018} have been devised, but are also sensitive to misspecification of the propensity score models.
For the estimator presented in Section \ref{sec:estimation}, as in e.g. \cite{Chernozhukov2018}, we instead rely on the efficient influence function \citep{Robins1995, Hahn1998} to estimate the IRF, namely $g\left(Z_t^{(h)}, h; \Gamma_0\right) - \theta_0^{(h)}$ where
\begin{equation}\label{eq:dr_score}
\begin{aligned}
g\left(Z_t^{(h)}, h; \Gamma_0\right) &= \mu_0(1,X_t,h) - \mu_0(0,X_t,h) + \frac{D_t}{e_0(X_t)}(Y_{t+h} - \mu_0(1,X_t,h))\\
&\quad - \frac{1-D_t}{1-e_0(X_t)}(Y_{t+h} - \mu_0(0,X_t,h))
\end{aligned}
\end{equation}
and it can be shown that
\begin{equation*}
    \theta_0^{(h)} = \E\left[g\left(Z_t^{(h)}, h; \Gamma_0\right)\right].
\end{equation*}
An influence function measures how a small perturbation of the data affects an estimator. The \emph{efficient} influence function is the particular influence function that (among all regular estimators) achieves the lowest possible asymptotic variance allowed by the semiparametric model. Readers are referred to \cite{Hines2022}, \cite{Kennedy2023Review}, \cite{Fisher2021} and \cite{Tsiatis2006} for a review of influence functions and semiparametric theory. In contrast to regression adjustment or inverse probability weighting, an estimator relying on the above influence function is \emph{Neyman orthogonal} \citep{Neyman1959,Neyman1979,Chernozhukov2017AER}. Technically, the efficient influence function of the IRF is Neyman orthogonal since the (Gateaux) derivative of its expected value with respect to either nuisance function equals zero \citep[for a detailed discussion, see][]{Chernozhukov2018}. This property ensures that small deviations from the true nuisance functions have no first-order effect on the estimation of $\theta_0^{(h)}$. Loosely speaking, if the estimated nuisance functions are ``close enough'' to their true values, estimation errors only have a vanishing impact on the IRF estimator. While the theory presented in Section~\ref{sec:estimation} leverages this Neyman orthogonality property, it is worth mentioning that an estimator based on Equation \eqref{eq:dr_score} is also \emph{doubly robust}, in the sense that it remains consistent if only one of the nuisance functions is correctly specified. In the sequel, we use standard notations $O(\cdot)$ and $o(\cdot)$ to indicate rates of convergence for sequences. In particular, if $\{x_t\}_1^\infty$ is any real sequence, $\{a_t\}_1^\infty$ a sequence of positive real numbers, and there exists a finite constant $B$ such that $|x_t|/a_t\leq B$ for all $t$, we write $x_t=O(a_t)$. If $x_t/a_t$ converges to zero, we write $o(a_t)$. We use $\norm{\cdot}_q$ to denote the $L_q$-norm; e.g. we write $\norm{f}_q = \norm{f(Z)}_q = \left(\int \vert f(z) \vert^q \mathrm{d}\mathcal{P}(z)\right)^{1/q}$.

\subsection{Identification}
While the focus of this paper in on the estimation of the quantity introduced in Equation~\eqref{eq:irf}, we here briefly present assumptions under which $\theta_0^{(h)}$ conveys the interpretation of the average \emph{causal} effect that the binary impulse variable $D_t$ has on the outcome variable $Y_{t+h}$. Identification is formulated within the potential outcomes framework of causal inference \citep{Rubin1974, Robins1986}. Let $Y_{t+h}(d)$ be the potential outcome, i.e., the random variable one would observe at time $t+h$ if the treatment at time $t$ would have been $D_t=d$. The following assumption ensures identification of the causal effect.
\begin{assumption}[\cite{Angrist2018,Rambachan2021}]\label{ass:identification}
    For all $t \in \mathcal{T}$ and  $h \in \mathbb{N}_0$ the following holds.
    \begin{enumerate}
        \item\label{ass:identification_cia} The potential outcomes are conditionally independent of the treatment, i.e. \\$Y_{t+h}(1), Y_{t+h}(0) \perp D_t \vert X_t$.
        \item\label{ass:identification_sutva} The observed outcome is $Y_{t+h} = D_t Y_{t+h}(1) + (1-D_t) Y_{t+h}(0)$.
        \item\label{ass:identification_overlap} For all $x\in \mathcal{X}$ it holds that $\eta < e_0(x) < 1-\eta$, for some $0 < \eta < 1$.
    \end{enumerate}
\end{assumption}
Assumption \ref{ass:identification}.\ref{ass:identification_cia} requires conditional independence between the treatment at time $t$ and the potential outcomes. Notably, it is not necessary for $D_t$ to be conditionally independent from future treatment assignments. However, in the case where $D_t \not\perp D_{t+1}, \dots, D_{t+h} \vert X_t$, the identified effect corresponds to the effect of a treatment including potential future treatments caused by $D_t$ \citep{Jorda2023}. Assumption \ref{ass:identification}.\ref{ass:identification_overlap} imposes that at each point in time, the treatment assignment is not deterministic. In other words, there are no situations in which either $D_t=1$ or $D_t=0$ with (conditional) probability of one. In essence, these assumptions require the treatment variable of interest (or a sufficiently informative proxy) to be observed, and a set of control variables to be available such that the treatment assignment is (conditionally) as good as random. In macroeconomics, this requirement corresponds to identification strategies that rely on constructed shocks, such as narrative monetary policy shocks \citep[e.g.,][]{Romer2004, Ramey2016}, or ``direct causal inference'' approaches based on externally constructed measures of structural shocks \citep{Nakamura2018}. For approaches that use continuous shock measures, the shock can also be incorporated into our framework by discretizing them. This is however only required to define the treatment variable within our setup, not by the identification strategy itself.
For a more rigorous discussion of the identification of treatment effects with time-dependent data -- and in particular for the connection between Assumption \ref{ass:identification} and classical macroeconomic shocks -- we refer to \cite{Rambachan2021}. The following theorem finally establishes identification of the average treatment effect.
\begin{theorem}\label{thm:identification}
    Under Assumptions \ref{ass:time_invariance} and \ref{ass:identification} the ATE is identified as
    \begin{equation*}
        \theta_0^{(h)} = \E\left[Y_{t+h}(1)-Y_{t+h}(0)\right].
    \end{equation*}
\end{theorem}
As pointed out by \cite{Chernozhukov2018}, the DML estimator will yield unbiased results also in a setting where the causal identification assumptions presented in this section fail to hold. In this case however, the estimated IRF has to be interpreted as a prediction difference rather than a causal effect.

\section{Estimation}\label{sec:estimation}

This section outlines the estimator for $\theta_0^{(h)}$ and its asymptotic properties when the nuisance functions are estimated with flexible, nonparametric machine learning algorithms. The estimator is developed in three steps. First, we provide results for the (hypothetical) case where the nuisance functions are known. In a second step, nuisance functions are estimated, but multiple independent stochastic processes generated from the same distribution are available. Lastly, we provide results for the case where nuisance functions have to be estimated and only a single stochastic process is available.

\subsection{An oracle estimator}\label{subsec:oracle}
In case the nuisance functions are known, we can estimate the effect of interest by simply averaging the stochastic processes $\mathcal{G}^{(h)} = \left\{ g\left(Z_t^{(h)}, h; \Gamma_0\right) : t \in \mathcal{T}\right\}$ over the index set. We refer to this estimator as the \emph{oracle estimator} for the IRF. The following assumption and theorem provide conditions under which the oracle estimator is asymptotically normally distributed.
\begin{assumption}\label{ass:oracle}
    For some $\beta>2$ and all $h \in \mathbb{N}_0$, the following conditions hold.
    \begin{enumerate}
        \item\label{ass:oracle_stationarity} The stochastic processes $\mathcal{G}^{(h)}$ are weakly stationary.
        \item The variance satisfies that $0 < V^{(h)}_0 = \lim_{T\to\infty} \Var\left[\frac{1}{\sqrt{T}} \sum_{t \in \mathcal{T}}  g\left(Z_t^{(h)}, h; \Gamma_0\right)\right]$.
        \item\label{ass:oracle_moments} $\mathcal{G}^{(h)}$ is uniformly $L_\beta$-bounded, i.e. $\sup_{t\in\mathcal{T}} \E\left[\left\vert g\left(Z_t^{(h)}, h; \Gamma_0\right)\right\vert^\beta\right] < \infty$.
        \item\label{ass:oracle_dependence} $\mathcal{G}^{(h)}$ is $\alpha$-mixing, with coefficients $\alpha(s)$, $s \in \mathbb{N}$, satisfying $\sum_{s=1}^\infty \alpha(s)^{(\beta-2)/\beta} < \infty$.
    \end{enumerate}
\end{assumption}
\begin{theorem}\label{thm:oracle}
    Let the oracle IRF estimator be given by
    \begin{equation*}
        \tilde\theta^{(h)} = \frac{1}{T} \sum_{t\in\mathcal{T}} g\left(Z_t^{(h)}, h; \Gamma_0\right).
    \end{equation*}
    Under Assumptions \ref{ass:time_invariance} and \ref{ass:oracle}, we have that
    \begin{equation*}
        \sqrt{T} (\tilde\theta^{(h)} - \theta_0^{(h)}) \overset{\mathrm{d}}{\to} \mathcal{N}\left(0, V^{(h)}_0 \right),
    \end{equation*}
    with $V^{(h)}_0 = \sum_{s=-\infty}^{+\infty} \Cov\left[g\left(Z_t^{(h)}, h; \Gamma_0\right), g\left(Z_{t-s}^{(h)}, h; \Gamma_0\right)\right]$.
\end{theorem}

Given that $g(z, h; \Gamma)$ is a measurable function as long as the nuisance functions are measurable,  Assumption~\ref{ass:oracle}.\ref{ass:oracle_dependence} is satisfied if the stochastic processes $S^{(h)}$ are $\alpha$-mixing for some $\beta > 2$ \citep{Davidson2021}. This is however not necessary, and $S^{(h)}$ can exhibit less favorable dependence structures as long as $\mathcal{G}^{(h)}$ adheres to Assumptions \ref{ass:time_invariance} and \ref{ass:oracle}. Moreover, while the above assumptions are standard in the application of functional central limit theory, variations are possible that still lead to the desired asymptotics. For example, weak stationarity in Assumption \ref{ass:oracle}.\ref{ass:oracle_stationarity} can be relaxed to a constant mean, permitted that additionally $V^{(h)}_0 < \infty$ (see e.g. the discussion in \cite{Philipps1987}). Weak stationarity is however required in our setting in order to obtain a tractable estimator for $V^{(h)}_0$.

It is important to remark that the two assumptions \ref{ass:oracle}.\ref{ass:oracle_moments} and \ref{ass:oracle}.\ref{ass:oracle_dependence} represent an inherent trade-off. The more absolute moments of $\mathcal{G}^{(h)}$ are required to exist, the more dependence is acceptable in the stochastic processes to still reach asymptotic normality. For the sum in Assumption \ref{ass:oracle}.\ref{ass:oracle_dependence} to converge, we need $\alpha(s) = O(s^{-\phi})$ for some $\phi > \phi_0 = \beta/(\beta-2)$, i.e. the process needs to be $\alpha$-mixing of size $-\phi_0$. As $\beta\to\infty$ so that all moments are finite, the required mixing size $\phi_0\to 1$. Because the mixing coefficients also determine bounds for the (absolute) autocovariance function of the process (see \cite{Davidson2021}, Corollary 15.3), this directly implies that with all moments existing, $\Cov\left[g\left(Z_t^{(h)}, h; \Gamma_0\right), g\left(Z_{t-s}^{(h)}, h; \Gamma_0\right)\right] = O(s^{-1})$ for the assumptions to be satisfied.

\subsection{The double machine learning estimator with multiple independent stochastic processes}\label{subsec:estimator}
We now provide asymptotic results for the case where the nuisance functions are estimated and $K\geq 2$ independent stochastic processes generated from the same distribution $\mathcal{P}$ are available. Denote the individual stochastic processes as $S_i^{(h)} = \{Z^{(h)}_t : t \in \mathcal{T}_i\}$, where, without loss of generality, we assume $\vert \mathcal{T}_i \vert = T/K$ for all $i=1, ..., K$. The estimation procedure is outlined in Procedure \ref{proc:dml}.
\begin{procedure}
    \caption{DML estimator for the IRF with cross-fitting on multiple independent stochastic processes}\label{proc:dml}
    For each forecast horizon $h$, follow the subsequent procedure.
    \begin{enumerate}
    \item For each $i=1, ..., K$
    \begin{enumerate}
      \item Fit appropriate machine learners $\hat\Gamma_{S_{-i}^{(h)}} = (\hat{\mu}_{S_{-i}^{(h)}}(d, x,h)$, $\hat{e}_{S_{-i}^{(h)}}(x))$ on the sample $S_{-i}^{(h)} = \bigcup_{j = 1, j\neq i}^K S_j^{(h)}$.
      \item Compute the average of $g\left(z, h; \hat\Gamma_{S_{-i}^{(h)}}\right)$ on $S_i^{(h)}$ as
      \begin{equation*}
        \hat{\theta}^{(h)}_{S_i^{(h)}} = \frac{1}{|\mathcal{T}_i|} \sum_{t\in\mathcal{T}_i} g\left(Z_t^{(h)}, h; \hat\Gamma_{S_{-i}^{(h)}}\right).
      \end{equation*}
    \end{enumerate}
    \item Compute the IRF estimator at horizon $h$ as
    \begin{equation*}
        \hat{\theta}^{(h)} = \sum_{i=1}^{K} \frac{|\mathcal{T}_i|}{T} \hat{\theta}^{(h)}_{S_i^{(h)}}.
    \end{equation*}
    \end{enumerate}
\end{procedure}
Asymptotics for the DML estimator $\hat{\theta}^{(h)}$ from Procedure \ref{proc:dml} are obtained by introducing assumptions under which $\hat{\theta}^{(h)}$ has the same asymptotic distribution as the oracle estimator $\tilde\theta^{(h)}$ in Theorem~\ref{thm:oracle}. To this end, we impose standard assumptions on the convergence rates of the learners used in step 1. of Procedure~\ref{proc:dml}. In particular, we assume that the machine learners are consistent and that the product of the estimation errors decays fast enough.
\begin{assumption}\label{ass:risk_decay}
    Let the realization set be $\Xi_{T}^{(h)}$, which is a shrinking neighborhood of the true nuisance functions $\Gamma_0 = (\mu_0(d,x,h), e_0(x))$. Let $\{\Delta_T\}_{T\geq1}$ and $\{\delta_T\}_{T\geq1}$ be sequences of positive constants converging to zero.

    Define the statistical rates $r_{\mu,T} = \sup_{t\in\mathcal{T}}\sup_{\mu \in \Xi_{T}^{(h)}} \norm{\mu(D_t,X_t, h) - \mu_0(D_t,X_t, h)}_2$ and $r_{e,T} = \sup_{t\in\mathcal{T}}\sup_{e \in \Xi_{T}^{(h)}} \norm{e(X_t) - e_0(X_t)}_2$. Let $C$ be a fixed strictly positive constant. For all $i=1, ..., K$ and $h\in\mathbb{N}_0$, the following conditions hold.
    \begin{enumerate}
        \item The nuisance function estimators $\hat\Gamma_{S_{-i}^{(h)}}$ belong to $\Xi_{T}^{(h)}$ with probability at least $1-\Delta_T$.
        \item \label{ass:risk_decay_bounded} For $q>2$, we have $\sup_{t\in\mathcal{T}}\sup_{\mu \in \Xi_{T}^{(h)}} \norm{\mu(D_t,X_t, h) - \mu_0(D_t,X_t, h)}_q \leq C < \infty$ and $\sup_{t\in\mathcal{T}}\sup_{e \in \Xi_{T}^{(h)}} \norm{e(X_t) - e_0(X_t)}_q \leq C < \infty$.
        \item\label{ass:risk_decay_rate} $r_{\mu,T} \leq \delta_T$, $r_{e,T} \leq \delta_T$ and $r_{\mu,T} \cdot r_{e,T} \leq T^{-1/2} \delta_T$.
        \item\label{ass:risk_decay_prob}$\sup_{t\in\mathcal{T}}\sup_{e \in \Xi_{T}^{(h)}}\norm{e(X_t) -1/2}_\infty \leq 1/2 - \eta $~~for $ 0<\eta<1$.
        \item\label{ass:risk_decay_variance} $\sup_{t\in\mathcal{T}}\E\left[\left(Y_{t+h} - \mu_0(d,X_t,h)\right)^2 \vert X_t, D_t=d\right] \leq \epsilon^2_d < \infty$
    \end{enumerate}
\end{assumption}
The statistical rates in Assumption~\ref{ass:risk_decay} are defined in terms of uniform $L_2$-norms. Under the additional assumptions of strict stationarity of $S^{(h)}$ and measurability of the nuisance functions, these uniform norms would reduce to simple $L_2$-norms. Assumption~\ref{ass:risk_decay}.\ref{ass:risk_decay_bounded} requires that, with probability approaching one, the estimation errors are uniformly $L_q$-bounded for $q>2$. Assumption~\ref{ass:risk_decay}.\ref{ass:risk_decay_rate} requires the estimation errors to converge uniformly in $L_2$-norm to zero and their products to converge at least at the rate $\sqrt{T}$ with probability approaching one. Our assumptions on convergence rates are set up to accommodate the application of a broad range of machine learning estimators for the nuisance functions. There is a rich literature deriving convergence rates of machine learners under more strict conditions than used here. \cite{Wong2020}, for example, provide Lasso convergence rates for stochastic processes under the assumption of exact sparsity. For $\alpha$-mixing Gaussian processes, they find the $L_2$ convergence rate to be of order $O\left(A(T) \sqrt{\log{\mathrm{dim}(\mathcal{X})}/T} \right)$, where $A(T)=\sum_{s=0}^T\alpha(s)$. Assumption~\ref{ass:risk_decay}.\ref{ass:risk_decay_rate} is satisfied for $A(T) = o(T^{1/4})$, imposing a restriction on how fast the dependence in the data has to decay. We provide more references in Section \ref{subsubsec:nuisance_estimators}. Next, Assumption \ref{ass:risk_decay}.\ref{ass:risk_decay_prob} implies that the estimated propensity scores remain uniformly bounded away from zero and one with probability approaching one. Finally, Assumption \ref{ass:risk_decay}.\ref{ass:risk_decay_variance} requires that the conditional variance of the outcome variable is a bounded random variable.

We furthermore impose the following sequential conditional exogeneity condition.
\begin{assumption}\label{ass:X_richness}
    For all $t\in\mathcal{T}$ and $h\in\mathbb{N}_0$ we have that $\E[Y_{t+h} \vert X_t, D_t = d, \{Z^{(0)}_u : u \in \mathcal{T}, u < t\}] = \E[Y_{t+h} \vert X_t, D_t = d]$ and $\E[D_t \vert X_t, \{Z^{(0)}_u : u \in \mathcal{T}, u < t\}] = \E[D_t \vert X_t]$.
\end{assumption}
Assumption \ref{ass:X_richness} implies that the residuals $D_t - e_0(X_t)$ and $Y_{t+h} - D_t\mu_0(1,X_t, h) - (1-D_t)\mu_0(0,X_t, h)$ are mean independent of past information on $(Y_t, X_t, D_t)$. In line with standard assumptions in the literature \citep[e.g.][]{Semenova2023,Olea2024}, this in practice requires a rich enough set of control variables, which can also contain past values of $Y_t$ and $D_t$.

The following theorem finally establishes that the DML estimator for the IRF is asymptotically unbiased and normally distributed. In particular, the estimator retains the parametric $\sqrt{T}$ convergence rate. The proof is relegated to the \hyperref[app:proofs]{Appendix}. The main idea of the proof is to show that the IRF estimator using estimated nuisance functions converges to the oracle IRF estimator $\tilde\theta^{(h)}$, which itself is asymptotically normally distributed as shown in Theorem \ref{thm:oracle}, at rate $\sqrt{T}$.
\begin{theorem}\label{thm:ate}
    Let $S_i^{(h)} = \left\{Z_t^{(h)} : t \in \mathcal{T}_i\right\}$ for $i=1, ..., K\geq 2$ and $h\in\mathbb{N}_0$ be independent stochastic processes generated from the same distribution $\mathcal{P}$ such that $|\mathcal{T}_i|=T/K$ for all $i$. Define the estimator as
    \begin{align*}
        \hat{\theta}^{(h)} = \sum_{i=1}^{K}\frac{|\mathcal{T}_i|}{T} \hat{\theta}^{(h)}_{S_i^{(h)}} \quad \text{with}\\
        \hat{\theta}^{(h)}_{S_i^{(h)}} = \frac{1}{|\mathcal{T}_i|} \sum_{t \in \mathcal{T}_i} g\left(Z_t^{(h)}, h; \hat\Gamma_{S_{-i}^{(h)}}\right),
    \end{align*}
    where the nuisance functions $\hat\Gamma_{S_{-i}^{(h)}}=(\hat\mu_{S_{-i}^{(h)}}, \hat{e}_{S_{-i}^{(h)}})$ are estimated on $S_{-i}^{(h)} = \bigcup_{j = 1, j\neq i}^K S_j^{(h)}$. Then under Assumptions \ref{ass:time_invariance} and \ref{ass:oracle} - \ref{ass:X_richness} it holds that
    \begin{equation*}
        \sqrt{T} (\hat\theta^{(h)} - \theta_0^{(h)}) \overset{\mathrm{d}}{\to} \mathcal{N}\left(0, V^{(h)}_0 \right),
    \end{equation*}
    with $V^{(h)}_0$ as in Theorem \ref{thm:oracle}.
\end{theorem}

\subsection{The double machine learning estimator with one stochastic process}\label{subsec:dr_single}
In practice, multiple independent stochastic processes generated from the same distribution are often not available. Instead, a sample from a single stochastic process $S^{(h)}$ is observed. In the same spirit as approaches used for cross-validating models with dependent data \citep{Bergmeir2012, Racine2000} and the cross-fitting approach proposed by \cite{Semenova2023} for panel data, we split the single stochastic process into sub-sequences, removing a block of $k_T$ coordinates where the process is split. The resulting sub-sequences replace the independent stochastic processes $S_i^{(h)}$ from Procedure \ref{proc:dml}. To this end, let $\{\mathcal{T}_i : i=1, \dots, K\}$ be a partition of the index set $\mathcal{T}$ such that the order of the time indices within each $\mathcal{T}_i$ and across all subsets follows the original order in $\mathcal{T}$. Without loss of generality, we continue to assume that $\vert \mathcal{T}_i \vert = T / K$ for all $i$. The estimation procedure is described in Procedure \ref{proc:dr_single}.
\begin{procedure}[H]
    \caption{DML estimator for the IRF with cross-fitting on a single stochastic processes}\label{proc:dr_single}
    For each forecast horizon $h$, follow the subsequent procedure.
    \begin{enumerate}
      \item For each $i=1, ..., K$
      \begin{enumerate}
        \item Define $\mathcal{T}_{-i} = \{t : t \in \mathcal{T} \wedge \left(t < \inf(\mathcal{T}_i) - k_T \vee t > \sup(\mathcal{T}_i) + k_T\right)\}$
        \item Fit appropriate machine learners $\hat\Gamma_{S_{-i}^{(h)}} = (\hat\mu_{S_{-i}^{(h)}}(d,x,h), \hat{e}_{S_{-i}^{(h)}}(x))$ on the sample $S_{-i}^{(h)} = \{Z^{(h)}_t : t \in \mathcal{T}_{-i}\}$.
        \item Compute the average of $g\left(z, h; \hat\Gamma_{S_{-i}^{(h)}}\right)$ on $S_i^{(h)}$ as
        \begin{equation*}
            \hat{\theta}^{(h)}_{S_i^{(h)}} = \frac{1}{|\mathcal{T}_i|} \sum_{t\in\mathcal{T}_i} g\left(Z_t^{(h)}, h; \hat\Gamma_{S_{-i}^{(h)}}\right).
        \end{equation*}
      \end{enumerate}
    \item Compute the IRF estimator at horizon $h$ as
    \begin{equation*}
        \hat{\theta}^{(h)} = \sum_{i=1}^{K} \frac{|\mathcal{T}_i|}{T} \hat{\theta}^{(h)}_{S_i^{(h)}}.
    \end{equation*}

    \end{enumerate}
\end{procedure}
An illustration of the cross-fitting approach for $K=4$ is given in Figure~\ref{fig:cross_fitting}. The time series is divided into four sub-sequences. The nuisance functions are estimated on the union of sub-sequences $S_1$, $S_2$ and $S_3$, where $k_T$ coordinates are removed at the boundaries of $S_3$. Sub-sequence $S_3$ is used to compute $\hat\theta^{(h)}_{S_3}$. This procedure is repeated such that each sub-sequence is used once to conduct inference.
\begin{figure}[htbp]
    \centering
    \caption{Illustration of the cross-fitting procedure}
    \begin{tikzpicture}[scale=1.5]
      \draw[fill=gray!30] (0,0) rectangle (1.5,1);
      \draw[fill=gray!30] (1.5,0) rectangle (3.0,1);
      \draw[fill=blue!30] (3.0,0) rectangle (4.5,1);
      \draw[fill=gray!30] (4.5,0) rectangle (6.0,1);
      \fill[pattern=north east lines] (2.5,0) rectangle (3.0,1);
      \fill[pattern=north east lines] (4.5,0) rectangle (5,1);

      \node at (0.75,0.5) {$S_1^{(h)}$};
      \node at (2.25,0.5) {$S_2^{(h)}$};
      \node at (2.75,-0.25) {$k_T$};
      \node at (3.75,0.5) {$S_3^{(h)}$};
      \node at (4.75,-0.25) {$k_T$};
      \node at (5.25,0.5) {$S_4^{(h)}$};

      \draw[<->,thick] (2.5,-0.1) -- (3.0,-0.1);
      \draw[<->,thick] (4.5,-0.1) -- (5.0,-0.1);
      \draw[->,thick] (0, 1.2) -- (6.1, 1.2);
      \draw (0,1.15) -- (0, 1.25);
      \node[above] at (0,1.25) {$Z^{(h)}_1$};
      \draw (6,1.15) -- (6, 1.25);
      \node[above] at (6,1.25) {$Z^{(h)}_T$};
      \node[above] at (3,1.25) {\dots};
    \end{tikzpicture}
    \label{fig:cross_fitting}
    \begin{minipage}{\textwidth}
        \footnotesize\textsc{Note}: The figure illustrates the cross-fitting procedure for $K=4$. The nuisance functions are estimated using appropriate machine learners on the union of sub-sequences $S_1^{(h)}$, $S_2^{(h)}$ and $S_4^{(h)}$ after dropping $k_T$ observation at the boundaries to $S_3^{(h)}$. Sub-sequence $S_3^{(h)}$ is used to compute $\hat{\theta}^{(h)}_{S_3^{(h)}}$. This procedure is repeated such that each sub-sequence is used once to conduct inference.
      \end{minipage}
\end{figure}

The asymptotic results for the estimator described in Procedure \ref{proc:dr_single} require some additional assumptions.
\begin{assumption}\label{ass:sample_dependence}
    For $d\in\{0,1\}$, $h\in\mathbb{N}_0$ and some scalar constant $p\geq1$ the following conditions hold.
    \begin{enumerate}
        \item\label{ass:sample_dependence_k} $k_T = O(T)$.
        \item\label{ass:sample_dependence_continuous} The nuisance functions $\mu_0(d, x,h)$, $e_0(x)$ and the functions $\mu(d,x,h), e(x)\in\Xi_{T}^{(h)}$ are measurable.
        \item\label{ass:sample_dependence_bounded_y_error} For for $r>p$ and $1/r = 1/r' + 1/r''$, we have $\sup_{t\in\mathcal{T}}\, \lVert\sup_{\mu\in\Xi_{T}^{(h)}} (\mu(d,X_t,h) - \mu_0(d,X_t,h))\lVert_{2r'} < \infty$ and $\sup_{t\in\mathcal{T}} \norm{e_0(X_t)-D_t}_{2r''} < \infty$.
        \item\label{ass:sample_dependence_bounded_d_error} For $q>p$ and $1/q = 1/q' + 1/q''$, we have $\sup_{t\in\mathcal{T}} \, \lVert\sup_{e\in\Xi_{T}^{(h)}} (e(X_t) - e_0(X_t))\lVert_{2q'} < \infty$ and $\sup_{t\in\mathcal{T}} \norm{Y_t - \mu_0(D_t,X_t, h)}_{2q''} < \infty$.
        \item\label{ass:sample_dependence_mixing} The stochastic processes $S^{(h)} = \{Z_t^{(h)}:t\in\mathcal{T}\}$ are $\alpha$-mixing with coefficients $\alpha(s)$, satisfying for $T \to \infty$ that $\alpha\left(k_T\right)^{\psi} = o(T^{-1})$, where $\psi = 1/p - 1/\min(r,q)$.
    \end{enumerate}
\end{assumption}
Assumption \ref{ass:sample_dependence} imposes restrictions on the dependence structure in the stochastic processes $S^{(h)}$. As we can no longer rely on the independence between $S_i^{(h)}$ and $S_{-i}^{(h)}$, we require the dependence in the stochastic process to decay fast enough. Intuitively, after removing $k_T$ coordinates at the boundaries of $S_i^{(h)}$ it should (asymptotically) become independent of $S_{-i}^{(h)}$. The choice of $k_T$ hereby represents a trade-off. The larger $k_T$, the smaller is the effective estimation sample size. Indeed, Assumption \ref{ass:sample_dependence}.\ref{ass:sample_dependence_k} requires $k_T$ to not increase more rapidly than the sample size. However, $k_T$ needs to be large enough to satisfy Assumption \ref{ass:sample_dependence}.\ref{ass:sample_dependence_mixing}. The stronger the dependence in the stochastic processes $S^{(h)}$, the larger $k_T$ needs to be, in turn reducing the effective estimation sample size. For the sake of exposition, assume that $k_T = O(T^\vartheta)$ for $0<\vartheta\leq1$ and let us look at different exemplary assumptions on the dependence structure of $S^{(h)}$.
\begin{enumerate}[(i)]
\item[(i)] $\alpha$-mixing process: If the stochastic processes $S^{(h)}$ are $\alpha$-mixing of size $-\phi_0$, i.e. $\alpha(s) = O(s^{-\phi})$ for some $\phi>\phi_0$, then Assumption \ref{ass:sample_dependence}.\ref{ass:sample_dependence_mixing} is satisfied for $\vartheta>(\phi\psi)^{-1}$ and $0<\psi<1$. Put differently, the slower the decay in the dependence (smaller $\phi$), the larger $\vartheta$ and thus $k_T$ has to be.
\item[(ii)] Persistent process: For stochastic processes $S^{(h)}$ with very slowly decaying dependence of the form $\alpha(s) = O(s^{2d-1})$ with $0<d<1/2$, we have additionally that $0<\phi<1$ and as a consequence Assumption \ref{ass:sample_dependence}.\ref{ass:sample_dependence_mixing} is indeed never satisfied.
\item[(iii)] Weakly dependent processes: If the stochastic processes $S^{(h)}$ have mixing coefficients $\alpha(s)=O(\rho^s)$ for $0<\rho<1$, then Assumption \ref{ass:sample_dependence}.\ref{ass:sample_dependence_mixing} is already satisfied with $\vartheta > 0$.
\item[(iv)] Independent process: In this case, i.e. when it is assumed that the stochastic processes $S^{(h)}$ are a collection of independent random variables, then $\alpha(s)=0$ for all $s>0$ and Assumption \ref{ass:sample_dependence}.\ref{ass:sample_dependence_mixing} is already satisfied for $k_T=0$.
\end{enumerate}
These examples intuit two things that are required for our theory to hold in practice. First, if only one stochastic process $S^{(h)}$ is available for estimation, enough coordinates must be removed when splitting, so that no influence of past coordinates exists in coordinates of a subsequent split. Second, the stochastic processes $S^{(h)}$ itself have to exhibit fast enough decaying temporal dependence. If this is not the case, suitable transformations of the original process need to be found to achieve this. Lastly, the boundedness conditions in the above assumptions also represent a trade-off. The more moments of the residuals and estimation errors are finite (larger $r$ and $q$), the less stringent are the conditions on the dependence decay, and vice versa. In general, with only one stochastic process, the conditions on the estimation errors are more restrictive than those imposed in Assumption \ref{ass:risk_decay}.\ref{ass:risk_decay_bounded}. The proposed cross-fitting approach is closely related to the methodology introduced by \cite{Semenova2023} for panel data, wherein the data is partitioned into folds along the temporal dimension. The primary distinction lies in the selection of observations omitted between the estimation and inference samples. While \cite{Semenova2023} advocate for a $K$-fold partitioning of the sample, removing an entire fold between estimation and inference samples, i.e. $k_T=\lfloor T/K \rfloor$ (see Section \ref{sec:cross-fitting} for further discussion), the present study adopts a more flexible strategy. Specifically, as outlined above, the choice of $k_T$ can be informed by the underlying dependence structure of the data. The following theorem finally establishes asymptotic properties for the case when nuisance functions are estimated based on a single stochastic process.
\begin{theorem}\label{thm:ate_single}
    Given the stochastic processes $S^{(h)} = \{Z_t^{(h)}:t\in\mathcal{T}\}$ for $h\in\mathbb{N}_0$, define $K\geq 2$ sub-sequences $S_i^{(h)} = \left\{Z_t^{(h)} : t \in \mathcal{T}_i\right\}$ such that $\{\mathcal{T}_i : i=1, \dots, K\}$ is a partition of the index set $\mathcal{T}$ where the order of the time indices within each $\mathcal{T}_i$ and across all sub-sequences with length $|\mathcal{T}_i| = T/K$ follows the original order in $\mathcal{T}$. Define the estimator as
    \begin{align*}
        \hat{\theta}^{(h)} = \sum_{i=1}^{K}\frac{|\mathcal{T}_i|}{T} \hat{\theta}^{(h)}_{S_i^{(h)}} \quad\text{with}\\
        \hat{\theta}^{(h)}_{S_i^{(h)}} = \frac{1}{|\mathcal{T}_i|} \sum_{t \in \mathcal{T}_i} g\left(Z_t^{(h)}, h; \hat\Gamma_{S_{-i}^{(h)}}\right),
    \end{align*}
    where the nuisance functions $\hat\Gamma_{S_{-i}^{(h)}} = (\hat\mu_{S_{-i}^{(h)}}, \hat{e}_{S_{-i}^{(h)}})$ are estimated using the sub-sequences $S_{-i}^{(h)} = \{Z_t^{(h)} : t \in \mathcal{T}_{-i}\}$ for $\mathcal{T}_{-i} = \{t : t \in \mathcal{T} \wedge \left(t < \inf(\mathcal{T}_i) - k_T \vee t > \sup(\mathcal{T}_i) + k_T\right)\}$ respectively. Then under Assumptions \ref{ass:time_invariance} and \ref{ass:oracle} - \ref{ass:sample_dependence} it holds that
    \begin{equation*}
        \sqrt{T} (\hat\theta^{(h)} - \theta_0^{(h)}) \overset{\mathrm{d}}{\to} \mathcal{N}\left(0, V^{(h)}_0 \right),
    \end{equation*}
    with $V^{(h)}_0$ as in Theorem \ref{thm:oracle}.
\end{theorem}

\subsection{Variance estimation and inference}
The variances $V^{(h)}_0$ can be estimated using standard long-run variance estimators for time series, such as the one proposed by \cite{Newey1987}.
\begin{assumption}\label{ass:variance_estimator}
    The following conditions hold.
    \begin{enumerate}
        \item \label{ass:variance_estimator-i} There are some fixed finite constants $C$ and $r> 4$ such that \\$\sup_{t\in\mathcal{T}} \E\left[\left\vert g\left(Z_t^{(h)}, h; \Gamma_0\right)\right\vert^{r}\right] < C$.
        \item \label{ass:variance_estimator-ii} There exists a measurable function $m(z)$ such that $\sup_{\Gamma \in \Xi_T} |g(z, h; \Gamma)| < m(z)$, where for some finite constant $D$, $\sup_{t\in\mathcal{T}}\E[m(Z_t)^2] < D$.
        \item \label{ass:variance_estimator-iv} For $q>2$ and some fixed strictly positive and finite constant $C$ we have $\sup_{t\in\mathcal{T}}\norm{Y_t}_q \leq C$.
        \item For some scalar $0< b_m < b_r \leq 1/2$ it holds that:
        \begin{enumerate}
            \item \label{ass:variance_estimator-iii} The bandwidth $m_T$ is a function of the sample size such that $\lim_{T\to\infty} m_T = \infty$ and for $T\to\infty$ it holds that $T^{-b_m} m_T = o(1)$.
            \item \label{ass:variance_estimator-v} $r_{\mu,T}\leq \delta_TT^{-b_r}$ and $r_{e,T} \leq \delta_TT^{-b_r}$.
        \end{enumerate}
    \end{enumerate}
\end{assumption}
Assumptions \ref{ass:variance_estimator}.\ref{ass:variance_estimator-iii} and \ref{ass:variance_estimator}.\ref{ass:variance_estimator-v} represent an inherent trade-off. The slower the convergence rate of the machine learner, the slower the bandwidth is allowed to grow. If the learner converges at the parametric rate, Assumption \ref{ass:variance_estimator}.\ref{ass:variance_estimator-iii} reduces to the usual assumption $m_T=o(T^{1/2})$ \citep[see, e.g.][]{Kool1988}. Note that by Assumption \ref{ass:risk_decay}.\ref{ass:risk_decay_rate} we have that $b_r \geq 1/4$.

The following theorem establishes consistency of a variance estimator resulting from averaging \cite{Newey1987} type estimators on each $S_i^{(h)}$.
\begin{theorem}\label{thm:variance}
    Given the stochastic processes $S^{(h)} = \{Z_t^{(h)}:t\in\mathcal{T}\}$, define the sub-sequences $S_i^{(h)}$ for $i = 1, ..., K\geq 2$ and $h\in\mathbb{N}_0$ as in Theorem \ref{thm:ate_single}. Furthermore, define $v^{(h)}_t = g\left(Z^{(h)}_t, h; \Gamma_0\right) - \theta^{(h)}_0$ and the corresponding estimated quantities as $\hat{v}^{(h)}_{S_i^{(h)},t} = g\left(Z^{(h)}_t, h; \hat\Gamma_{S_{-i}}\right) - \hat\theta^{(h)}$. $\hat\theta^{(h)}$ and $S_{-i}^{(h)}$ are defined as in Theorem \ref{thm:ate_single} and the nuisance functions $\hat\Gamma_{S_{-i}^{(h)}} = (\hat\mu_{S_{-i}^{(h)}}, \hat{e}_{S_{-i}^{(h)}})$ are estimated on $S_{-i}^{(h)}$. Let $w(s,m_T) = 1-s/(m_T+1)$, where $m_T$ is a bandwidth parameter and define the additional index sets $\mathcal{T}_{i,s} = \{t \in \mathcal{T}_i \, : \, t-s \geq \inf(\mathcal{T}_i) \}$. Moreover, define the following \cite{Newey1987} type variance estimators as
    \begin{align*}
      \widehat{V}^{(h)}_{S_i^{(h)}} &= \frac{1}{\vert\mathcal{T}_i\vert}\left( \sum_{t\in\mathcal{T}_i} (\hat{v}^{(h)}_{S_{-i}^{(h)},t})^2 + 2 \sum_{s=1}^{m_T} w(s,m_T) \sum_{t\in\mathcal{T}_{i,s}} \hat{v}^{(h)}_{S_{-i}^{(h)},t} \hat{v}^{(h)}_{S_{-i}^{(h)},t-s}  \right).
    \end{align*}
    The variance estimator is finally defined as
    \begin{equation*}
        \widehat{V}^{(h)} = \sum_{i=1}^{K}\frac{|\mathcal{T}_i|}{T} \widehat{V}^{(h)}_{S_i^{(h)}}
    \end{equation*}
    Then for $V^{(h)}_0$ as in Theorem \ref{thm:oracle} and under Assumptions \ref{ass:time_invariance} and \ref{ass:oracle} - \ref{ass:variance_estimator} as $T\to\infty$ it holds that
    \begin{equation*}
      \left\vert\widehat{V}^{(h)} - V^{(h)}_0\right\vert \overset{p}{\to} 0
    \end{equation*}
    with measure $\mathcal{P}$.
\end{theorem}
Note that while Theorem \ref{thm:variance} is formulated in terms of a specific weight function $w(s,m)$, as in \cite{Newey1987}, the variance estimator is consistent for any weight function additionally satisfying for each $s$ that $\lim_{m(T)\to\infty} w(s,m_T)=1$ and $|w(s,m_T)| < \infty$. Using the estimators in Theorems \ref{thm:ate_single} and \ref{thm:variance}, inference can be conducted by constructing level-$\alpha$ confidence bounds for $\theta_0^{(h)}$ as
\begin{equation}\label{eq:ci}
  \theta^{(h)}_0 \in \left(\hat\theta^{(h)} \pm \frac{1}{\sqrt{T}} \Phi^{-1} \left(1-\frac{\alpha}{2}\right) \sqrt{\widehat{V}^{(h)}} \right),
\end{equation}
where $\Phi^{-1}(\cdot)$ represents the inverse cumulative distribution function of the standard normal distribution.

\section{Considerations on the practical implementation of the estimator}\label{sec:practical_implementation}
Here, we gather practical recommendations for the time series DML estimator.

\subsection{Cross-fitting and small samples}\label{sec:cross-fitting}
For small samples, we recommend setting $K$ to a rather large value (e.g. $K=10$ or $K=20$). This increases the number of observations available to estimate the nuisance functions. Regarding the choice of $k_T$, it is important to note that this is directly affected by the definition of the outcome variable. In cases where the effect on the outcome variable $h$ periods after the treatment is of interest, $k_T$ has to be chosen such that samples are not overlapping, i.e. $k_T \geq h$. Similarly, $k_T$ has to take into account possible lagged values in $X_t$. As a guideline, we recommend that practitioners set $K$ to either 10 or 20, and, following \cite{Semenova2023}, use $k_T = \lfloor T / K \rfloor$ as an initial choice. This satisfies Assumptions \ref{ass:sample_dependence}.\ref{ass:sample_dependence_k} and \ref{ass:sample_dependence}.\ref{ass:sample_dependence_mixing} across a wide range of dependence structures. Sensitivity analysis for variations in $k_T$ can then be done to ensure robustness of results.

\subsection{Estimators for the nuisance functions}\label{subsubsec:nuisance_estimators}
As in the \emph{i.i.d.} setting \citep{Chernozhukov2018}, our theory requires the estimators for the nuisance functions to be consistent with fast enough convergence rates. Following Theorems \ref{thm:ate} and \ref{thm:ate_single}, this needs to extend to estimation on time-dependent observations. Consistency and convergence rates on $\alpha$-mixing sequences are derived for Lasso in \cite{Wong2020}, for random forests in \cite{Goehry2020} and \cite{Davis2020}, for boosting algorithms in \cite{Lozano2014}, for support vector machines in \cite{Steinwart2009}, for kernel and nearest-neighbour regressions in \cite{Irle1997} and for spline and wavelet series regression estimators in \cite{CHEN2015447}. Consistency of deep feed-forward neural networks with ReLU activation functions on exponentially $\alpha$-mixing processes was recently shown in \cite{Ma2022}.

The selection of an appropriate machine learning algorithm ultimately has to consider the specific problem and the characteristics of the data. For example, the random forest algorithm has been shown to perform effectively in macroeconomic contexts, even with relatively small sample sizes \citep{Medeiros2021, Coulombe2022, Coulombe2024, Beck2025}. Our numerical experiments, as well as the empirical application, also find random forests to be effective on representative sample sizes and data generating processes. In contexts with typically larger sample sizes, recurrent neural networks have e.g. demonstrated success in modelling high-frequency market data \citep{Zhang2019,Lucchese2024}. For an overview of applying machine learning algorithms to time series, we also refer to the recent survey by \cite{Masini2023}.

Our numerical experiments suggest that, as in the \emph{i.i.d.} case \citep{Bach2024}, properly tuning hyperparameters of the chosen estimators plays an important role in the application of double machine learning also for time series data. In summary, we recommend to estimate and tune multiple different estimators and select the best in terms of the relevant loss function for the problem at hand (e.g. predictive mean squared error).

\subsection{Modelling multiple forecast horizons}
Our theory is agnostic to how the nuisance functions for different forecast horizons $h$ are modelled. In analogy to classical impulse response function estimation using local projections, each forecast horizon (and nuisance function) can be estimated separately, and potentially with different learning algorithms. Depending on the application, it is however also possible to estimate $\mu_0(d,x,h)$ with one model for all $h$, e.g. using sequence-to-sequence approaches \citep{Zelda2019}, provided they exhibit appropriate convergence rates. Applying approaches from multitask learning \citep{Caruana1997}, it is finally also possible to estimate $\mu_0(d,x,h)$ and $e_0(x)$ in one model if both problems can be learned using a shared representation. This approach has e.g. been explored by \cite{Shi2019} in the context of \emph{i.i.d.} data.

\subsection{Inference in finite samples}\label{subsubsec:inference_finite}
Theorem \ref{thm:variance} requires the choice of a kernel bandwidth $m_T$ that fulfils Assumption \ref{ass:variance_estimator}. A valid choice would be, for example,  $m_T=\gamma T^{1/3}$ where $\gamma$ is determined by the procedure proposed by \cite{NeweyWest1994}. The variance estimator, when combined with standard normal critical values, is asymptotically valid under general forms of heteroskedasticity and autocorrelation. However, extensive research has shown that this approach can perform poorly in finite samples \citep{Kiefer2005, Sun2014}. In response, the literature has proposed fixed-bandwidth asymptotics, where the ratio of bandwidth and sample size $m_T/T$ is held fixed as the sample size grows, rather than shrinking to zero (as in traditional small-bandwidth asymptotics). This framework yields non-standard limiting distributions and requires the use of fixed-bandwidth critical values \citep[see][]{Kiefer2005}, which better approximate the finite-sample behavior of test statistics. Our numerical experiments confirm this finding also for the estimator proposed in this manuscript. For practical applications, we thus advise to use fixed-bandwidth critical values when performing inference.

\subsection{Extreme propensity scores}

Assumption \ref{ass:identification} and \ref{ass:risk_decay} require the true propensity scores $e_0(x)$ to be bounded away from zero and one. In applications, however, certain treatments may have essentially zero probability for particular regions of the covariate space. For example policy interventions that are infeasible under extreme macroeconomic conditions. When such limited support is a concern, we recommend employing a propensity-score-based trimming procedure, following \cite{Crump2009}. This approach systematically excludes observations with extreme propensity scores to improve overlap between treated and control units and to ensure that the resulting estimand pertains to a subpopulation with adequate support. A related but different issue that may arise in finite samples are numerical instabilities of the reciprocal of the propensity score. In other words, while it may hold that at the population level the propensity score is bounded away from zero and one, in finite samples, the estimated propensity scores can still be close to zero or one. A simple approach to address this instability is to winsorize the estimated propensity scores to a small number, e.g. 0.01 and 0.99 \citep{Brian2011}. Alternatively, one can calibrate the estimated propensity score \citep{Ballinari2025, Klaassen2025}.

\section{Simulation experiments}\label{sec:simulation}
To validate our theoretical results from the previous sections in finite samples, we conduct simulation experiments. We compare the DML estimator from Procedures \ref{proc:dml} and \ref{proc:dr_single} to a regression adjustment estimator (RA) that estimates the nuisance functions $\mu_0(1,x, h)$ and $\mu_0(0,x, h)$ separately on the full sample and takes their difference (a T-learner in the terminology of \cite{Kuenzel2019}). To disentangle the effect of cross-fitting and usage of a Neyman orthogonal estimator, we also compute the IRF using the doubly robust influence function \eqref{eq:dr_score} without relying on cross-fitting (DR). In addition, we estimate the impulse response functions using standard local projections (LP) \citep{Jorda2005}. Results are presented throughout using random forests \citep{Breiman2001} as machine learning estimators for the nuisance functions. In the Online Appendix, Table~\ref{tab:results_boosting}, we also include results using a gradient boosting algorithm estimator \citep{Chen2016}, supporting the validity of our asymptotic theory for alternative machine learning estimators.\footnote{
The gradient boosting algorithm is found to require slightly larger sample sizes than random forests to reach a consistent estimate, highlighting that estimators can exhibit varying sample size requirements with respect to our asymptotic theory. For practical applications we thus suggest, cf. Section \ref{subsubsec:nuisance_estimators}, to select the estimator yielding the best predictive performance on the available sample size.} More details on the hyperparameter tuning schemes for both considered estimators are also given in the Online Appendix~\ref{app:tuning}. In all simulations, we generate data according to the following data generating process (DGP), which is a modification of the setup in \cite{NieWager2020}. For some noise level $\sigma_\epsilon$, a propensity score $e_0(X_t)$, a baseline effect $b(X_t)$ and a (conditional) treatment effect function $\tau(X_t)$, the outcome process is defined as
\begin{equation*}
    Y_t = b(X_t) + \left(D_t -0.5\right)\tau(X_t) + \gamma Y_{t-1} + \epsilon_t,
\end{equation*}
where the innovations $\epsilon_t$ are generated from a GARCH(1,1) process $\epsilon_t = \sigma_t \zeta_t$, with $\zeta_t \sim \mathcal{N}(0, 1)$ and $\sigma_t^2 = \omega + \beta_1 \zeta_{t-1}^2 + \beta_2 \sigma_{t-1}^2$. Following \cite{Jorda2005} we set $\beta_1 = 0.3$ and $\beta_2 = 0.5$. $\omega$ is set to ensure that the $\E[\epsilon_t^2] = \sigma^2_\epsilon$, that is  $\omega = (1-\beta_1 -\beta_2) \sigma^2_\epsilon$. We set $D_t \vert X_t \sim \text{Ber}\left(e_0(X_t)\right)$ with
\begin{align*}
  e_0(X_t) &= \left(1+ e^{-X_{1,t}} + e^{-X_{2,t}}\right)^{-1} \\
  b(X_t) &= 0.5\left((X_{1,t} + X_{2,t} + X_{3,t})^+ + (X_{4,t} + X_{5,t})^+\right) \\
  \tau(X_t) &= (X_{1,t} + X_{2,t} + X_{3,t})^+ - (X_{4,t} + X_{5,t})^+
\end{align*}
where $(x)^+ = \max(0, x)$. The confounder process is modelled as a $n$-dimensional, zero mean VARMA($p$,$q$) process
\begin{equation*}
  X_t = \sum_{i=1}^{p}A_iX_{t-i} + \sum_{j=1}^{q} M_j u_{t-j} + u_t,
\end{equation*}
where $u_t$ is a zero mean white noise random variable with nonsingular covariance matrix, parameterized as $\Sigma_u = \sigma^2_uI_n$ using some scalar $\sigma_u$ and the $n$-dimensional identity matrix $I_n$. In the spirit of \cite{Adamek2024}, the coefficient matrices are defined as $A_i = \alpha_A^{i-1} \Gamma^A$ and $M_j = \alpha_M^{j-1} \Gamma^M$, where $\alpha_A, \alpha_M$ are some scalars, $\Gamma^A$ (and $\Gamma^M$ correspondingly) is a tapered Toeplitz matrix with $\Gamma^A_{i,j} = \rho_A^{|i-j|+1}$ and $\Gamma^A_{i,j} = 0$ for $|i-j|\geq n/2$. We finally scale the process $X_t$ so that the confounders have unit variance. The baseline parametrization for our simulations is $\gamma=0.6$, $\sigma_\epsilon=\sigma_u=1$, $n=12$, $\alpha_A = \alpha_B = 0.3$, $p=2$, $q=1$, $\rho_A=0.35$ and $\rho_M=0.7$. The simulation procedure is described in Procedure~\ref{proc:sim}.
\begin{procedure}
\caption{Setup of the simulation study}\label{proc:sim}
\begin{enumerate}
  \item Draw a realization from the DGP with $T$ observations.
  \item For each evaluated forecast horizon $h = 0, 1, ..., H$:
  \begin{enumerate}[a)]
    \item Construct $\{S^{(h)}_i : i=1,\dots,K\}$ from the realization.
    \item Find optimal hyperparameters for the estimators for $\mu_{0}(0, X, h)$, $\mu_{0}(1, X, h)$, and $e_0(X)$ by cross-validation using $\{S^{(h)}_i : i=1,\dots,K\}$ as folds and removing $k_T$ observations at the boundary between estimation and inference sample.
    \item Train each of the four learners (DML, RA, DR, LP); and for each learner
    \begin{enumerate}
        \item compute the IRF estimator $\hat\theta^{(h)}$ according to Procedure~\ref{proc:dml} or Procedure~\ref{proc:dr_single}
        \item compute the variance estimator $\widehat{V}^{(h)}$ from Theorem \ref{thm:variance}. Following the arguments outlined in Section \ref{subsubsec:inference_finite}, we use the approach in \cite{NeweyWest1994} to determine the bandwidth $m_T$.
    \end{enumerate}
  \end{enumerate}
  \item Repeat steps 1. and 2. $N$ times.
\end{enumerate}
\end{procedure}
We perform the numerical experiments for two settings. A first one, where in step 2.a) the realizations for $S_i^{(h)}$ are in fact drawn separately by simulating $K$ independent realizations from the DGP in step 1., each with $T/K$ observations. In a second setting, the sub-samples are constructed from the one single realization drawn in step 1. In this setting, we remove $k_T=T/K$ coordinates at the boundary of the estimation and inference samples. Following our practical recommendation, we set $K=10$. Results for the baseline parametrization of the DGP for the DML, RA, DR and LP estimators for the setting with one stochastic process and for sample sizes $T\in\{125, 250, 500, 1'000, 8'000\}$ are shown in Table~\ref{tab:result_baseline}. Results for the setting with independent stochastic processes and for some parameter variations (different number of confounders, higher noise in the outcome process, empirically calibrated parameters) are deferred to Tables \ref{tab:result_baseline_two_proc}-\ref{tab:empircal_dgp} in the Online Appendix.

Overall, our simulations support the validity of our theory. Compared to the RA and DR estimators, the DML estimator exhibits the smallest average bias in all considered settings, converging at the expected $\sqrt{T}$-rate. Importantly, this holds true for both Procedure~\ref{proc:dml} relying on independent realizations and Procedure~\ref{proc:dr_single} using a single realization. Being linear estimators, local projections do not estimate the true nonlinear average treatment effect, but a weighted average of marginal effects \citep{kolesár2024dynamiccausaleffectsnonlinear}. This is highlighted by the observation that the bias of the LP estimator does not decrease and its coverage deteriorates with increasing sample sizes. On these points also refer to Section \ref{sec:lp}. Finally, the DML estimator produces valid confidence intervals, while the regression adjustment, doubly robust and local projection estimators fail to allow valid inference. As expected, when $K$ truly independent realizations are available (cf. Table \ref{tab:result_baseline_two_proc} in the Online Appendix), the bias of the DML estimator is lower than in the one-realization setting.

In Table \ref{tab:result_baseline}, for the DML estimator, we report both the coverage using asymptotic and fixed-bandwidth critical values (cf. Section \ref{subsubsec:inference_finite}). In small samples, coverage is better when using fixed-bandwidth critical values. As sample sizes increase, the improvement over asymptotic critical values becomes negligible.\footnote{Simulation results are qualitatively unchanged when the bandwidth is determined alternatively using the rule of thumb by \cite{Wooldridge2016} or \cite{Lazarus2018}. Results are available from the authors upon request.} In most settings, the coverage of the DML estimator is marginally too low, which likely reflects a finite-sample bias in estimating the variance of $\hat\theta^{(h)}$. Unreported results indeed show that, on average, our variance estimator is slightly smaller than the empirical variance of the IRF estimates across realizations. This downward bias is expected because, while the variance estimator is derived under the assumption of known nuisance functions, it is in practice constructed using estimates thereof. This introduces sampling variability that is not fully accounted for in finite samples.

\begin{table}[htbp]
    \caption{Simulation results for a baseline nonlinear DGP with $n=12$, $\sigma_\epsilon=1.0$ and random forest nuisance function estimates}
    \label{tab:result_baseline}
    \resizebox{\textwidth}{!}{\begin{threeparttable}\begin{tabular}{rrrrrrrrrrrrrrrrrr}
    \toprule
    \multicolumn{18}{c}{$h=0$, $\theta_0^{(h)}=0.3321$} \\
     & \multicolumn{5}{c}{DML} & \multicolumn{4}{c}{RA} & \multicolumn{4}{c}{DR} & \multicolumn{4}{c}{LP} \\
    T & Bias & std$(\hat{\theta}_h)$ & RMSE & $C_{b}$(95\%) & $C_{a}$(95\%) & Bias & std$(\hat{\theta}_h)$ & RMSE & $C_{b}$(95\%) & Bias & std$(\hat{\theta}_h)$ & RMSE & $C_{b}$(95\%) & Bias & std$(\hat{\theta}_h)$ & RMSE & $C_{b}$(95\%) \\
    \cmidrule(lr){1-1}\cmidrule(lr){2-6}\cmidrule(lr){7-10}\cmidrule(lr){11-14}\cmidrule(lr){15-18}
    125 & 0.053 & 0.797 & 0.799 & 0.945 & 0.944 & 0.360 & 0.455 & 0.580 & 0.511 & 0.263 & 0.413 & 0.490 & 0.707 & 0.090 & 0.357 & 0.368 & 0.871 \\
    250 & 0.060 & 0.398 & 0.402 & 0.949 & 0.941 & 0.285 & 0.300 & 0.413 & 0.514 & 0.208 & 0.278 & 0.347 & 0.732 & 0.128 & 0.251 & 0.282 & 0.847 \\
    500 & 0.030 & 0.220 & 0.222 & 0.954 & 0.947 & 0.219 & 0.196 & 0.294 & 0.464 & 0.142 & 0.183 & 0.232 & 0.776 & 0.135 & 0.173 & 0.220 & 0.829 \\
    1'000 & 0.029 & 0.144 & 0.147 & 0.931 & 0.928 & 0.182 & 0.139 & 0.229 & 0.422 & 0.103 & 0.131 & 0.167 & 0.784 & 0.138 & 0.126 & 0.187 & 0.730 \\
    8'000 & 0.006 & 0.042 & 0.042 & 0.966 & 0.965 & 0.176 & 0.048 & 0.182 & 0.016 & 0.040 & 0.042 & 0.058 & 0.853 & 0.131 & 0.046 & 0.139 & 0.141 \\
    \midrule
    \multicolumn{18}{c}{$h=1$, $\theta_0^{(h)}=0.1992$} \\
     & \multicolumn{5}{c}{DML} & \multicolumn{4}{c}{RA} & \multicolumn{4}{c}{DR} & \multicolumn{4}{c}{LP} \\
    T & Bias & std$(\hat{\theta}_h)$ & RMSE & $C_{b}$(95\%) & $C_{a}$(95\%) & Bias & std$(\hat{\theta}_h)$ & RMSE & $C_{b}$(95\%) & Bias & std$(\hat{\theta}_h)$ & RMSE & $C_{b}$(95\%) & Bias & std$(\hat{\theta}_h)$ & RMSE & $C_{b}$(95\%) \\
    \cmidrule(lr){1-1}\cmidrule(lr){2-6}\cmidrule(lr){7-10}\cmidrule(lr){11-14}\cmidrule(lr){15-18}
    125 & 0.088 & 0.857 & 0.861 & 0.930 & 0.920 & 0.379 & 0.446 & 0.586 & 0.480 & 0.275 & 0.407 & 0.491 & 0.708 & 0.060 & 0.362 & 0.367 & 0.923 \\
    250 & 0.067 & 0.496 & 0.500 & 0.947 & 0.940 & 0.312 & 0.301 & 0.434 & 0.394 & 0.222 & 0.277 & 0.355 & 0.692 & 0.111 & 0.265 & 0.287 & 0.914 \\
    500 & 0.053 & 0.236 & 0.242 & 0.955 & 0.952 & 0.250 & 0.196 & 0.317 & 0.321 & 0.161 & 0.183 & 0.243 & 0.731 & 0.123 & 0.193 & 0.229 & 0.887 \\
    1'000 & 0.053 & 0.144 & 0.153 & 0.942 & 0.941 & 0.218 & 0.135 & 0.257 & 0.226 & 0.126 & 0.127 & 0.179 & 0.708 & 0.139 & 0.133 & 0.193 & 0.792 \\
    8'000 & 0.012 & 0.044 & 0.045 & 0.955 & 0.954 & 0.220 & 0.049 & 0.226 & 0.001 & 0.051 & 0.042 & 0.066 & 0.759 & 0.144 & 0.053 & 0.153 & 0.183 \\
    \midrule
    \multicolumn{18}{c}{$h=2$, $\theta_0^{(h)}=0.1195$} \\
     & \multicolumn{5}{c}{DML} & \multicolumn{4}{c}{RA} & \multicolumn{4}{c}{DR} & \multicolumn{4}{c}{LP} \\
    T & Bias & std$(\hat{\theta}_h)$ & RMSE & $C_{b}$(95\%) & $C_{a}$(95\%) & Bias & std$(\hat{\theta}_h)$ & RMSE & $C_{b}$(95\%) & Bias & std$(\hat{\theta}_h)$ & RMSE & $C_{b}$(95\%) & Bias & std$(\hat{\theta}_h)$ & RMSE & $C_{b}$(95\%) \\
    \cmidrule(lr){1-1}\cmidrule(lr){2-6}\cmidrule(lr){7-10}\cmidrule(lr){11-14}\cmidrule(lr){15-18}
    125 & 0.107 & 0.815 & 0.822 & 0.925 & 0.916 & 0.353 & 0.455 & 0.576 & 0.478 & 0.256 & 0.424 & 0.495 & 0.748 & 0.031 & 0.410 & 0.411 & 0.931 \\
    250 & 0.065 & 0.517 & 0.521 & 0.932 & 0.932 & 0.278 & 0.313 & 0.419 & 0.425 & 0.197 & 0.296 & 0.356 & 0.731 & 0.078 & 0.306 & 0.316 & 0.927 \\
    500 & 0.060 & 0.259 & 0.266 & 0.946 & 0.942 & 0.232 & 0.212 & 0.314 & 0.346 & 0.151 & 0.201 & 0.252 & 0.737 & 0.104 & 0.221 & 0.244 & 0.911 \\
    1'000 & 0.060 & 0.167 & 0.178 & 0.918 & 0.917 & 0.210 & 0.153 & 0.260 & 0.220 & 0.125 & 0.146 & 0.193 & 0.717 & 0.128 & 0.160 & 0.204 & 0.854 \\
    8'000 & 0.011 & 0.050 & 0.051 & 0.951 & 0.951 & 0.214 & 0.052 & 0.220 & 0.001 & 0.049 & 0.048 & 0.069 & 0.793 & 0.128 & 0.055 & 0.140 & 0.352 \\
    \midrule
    \multicolumn{18}{c}{$h=3$, $\theta_0^{(h)}=0.0717$} \\
     & \multicolumn{5}{c}{DML} & \multicolumn{4}{c}{RA} & \multicolumn{4}{c}{DR} & \multicolumn{4}{c}{LP} \\
    T & Bias & std$(\hat{\theta}_h)$ & RMSE & $C_{b}$(95\%) & $C_{a}$(95\%) & Bias & std$(\hat{\theta}_h)$ & RMSE & $C_{b}$(95\%) & Bias & std$(\hat{\theta}_h)$ & RMSE & $C_{b}$(95\%) & Bias & std$(\hat{\theta}_h)$ & RMSE & $C_{b}$(95\%) \\
    \cmidrule(lr){1-1}\cmidrule(lr){2-6}\cmidrule(lr){7-10}\cmidrule(lr){11-14}\cmidrule(lr){15-18}
    125 & 0.105 & 0.899 & 0.905 & 0.934 & 0.926 & 0.305 & 0.482 & 0.571 & 0.483 & 0.225 & 0.452 & 0.505 & 0.767 & 0.012 & 0.448 & 0.448 & 0.937 \\
    250 & 0.069 & 0.527 & 0.532 & 0.920 & 0.916 & 0.248 & 0.339 & 0.420 & 0.467 & 0.175 & 0.322 & 0.366 & 0.766 & 0.058 & 0.340 & 0.345 & 0.934 \\
    500 & 0.060 & 0.287 & 0.294 & 0.935 & 0.931 & 0.208 & 0.226 & 0.307 & 0.362 & 0.140 & 0.220 & 0.260 & 0.774 & 0.087 & 0.249 & 0.264 & 0.915 \\
    1'000 & 0.061 & 0.185 & 0.195 & 0.921 & 0.919 & 0.193 & 0.164 & 0.253 & 0.242 & 0.118 & 0.160 & 0.199 & 0.765 & 0.112 & 0.178 & 0.210 & 0.902 \\
    8'000 & 0.011 & 0.057 & 0.058 & 0.940 & 0.940 & 0.194 & 0.058 & 0.203 & 0.003 & 0.045 & 0.055 & 0.071 & 0.831 & 0.117 & 0.061 & 0.131 & 0.479 \\
    \midrule
    \multicolumn{18}{c}{$h=4$, $\theta_0^{(h)}=0.0430$} \\
     & \multicolumn{5}{c}{DML} & \multicolumn{4}{c}{RA} & \multicolumn{4}{c}{DR} & \multicolumn{4}{c}{LP} \\
    T & Bias & std$(\hat{\theta}_h)$ & RMSE & $C_{b}$(95\%) & $C_{a}$(95\%) & Bias & std$(\hat{\theta}_h)$ & RMSE & $C_{b}$(95\%) & Bias & std$(\hat{\theta}_h)$ & RMSE & $C_{b}$(95\%) & Bias & std$(\hat{\theta}_h)$ & RMSE & $C_{b}$(95\%) \\
    \cmidrule(lr){1-1}\cmidrule(lr){2-6}\cmidrule(lr){7-10}\cmidrule(lr){11-14}\cmidrule(lr){15-18}
    125 & 0.138 & 0.949 & 0.959 & 0.929 & 0.908 & 0.270 & 0.500 & 0.569 & 0.538 & 0.205 & 0.475 & 0.517 & 0.773 & 0.021 & 0.496 & 0.497 & 0.935 \\
    250 & 0.064 & 0.671 & 0.674 & 0.910 & 0.907 & 0.217 & 0.369 & 0.428 & 0.462 & 0.159 & 0.352 & 0.387 & 0.764 & 0.049 & 0.377 & 0.380 & 0.922 \\
    500 & 0.051 & 0.322 & 0.327 & 0.940 & 0.926 & 0.177 & 0.246 & 0.304 & 0.390 & 0.119 & 0.239 & 0.267 & 0.806 & 0.062 & 0.270 & 0.277 & 0.929 \\
    1'000 & 0.056 & 0.201 & 0.209 & 0.938 & 0.936 & 0.170 & 0.174 & 0.244 & 0.280 & 0.106 & 0.170 & 0.201 & 0.803 & 0.092 & 0.192 & 0.213 & 0.921 \\
    8'000 & 0.011 & 0.063 & 0.064 & 0.946 & 0.946 & 0.174 & 0.061 & 0.184 & 0.007 & 0.041 & 0.061 & 0.074 & 0.874 & 0.100 & 0.067 & 0.121 & 0.662 \\
    \midrule
    \multicolumn{18}{c}{$h=5$, $\theta_0^{(h)}=0.0258$} \\
     & \multicolumn{5}{c}{DML} & \multicolumn{4}{c}{RA} & \multicolumn{4}{c}{DR} & \multicolumn{4}{c}{LP} \\
    T & Bias & std$(\hat{\theta}_h)$ & RMSE & $C_{b}$(95\%) & $C_{a}$(95\%) & Bias & std$(\hat{\theta}_h)$ & RMSE & $C_{b}$(95\%) & Bias & std$(\hat{\theta}_h)$ & RMSE & $C_{b}$(95\%) & Bias & std$(\hat{\theta}_h)$ & RMSE & $C_{b}$(95\%) \\
    \cmidrule(lr){1-1}\cmidrule(lr){2-6}\cmidrule(lr){7-10}\cmidrule(lr){11-14}\cmidrule(lr){15-18}
    125 & 0.147 & 1.015 & 1.025 & 0.919 & 0.900 & 0.238 & 0.539 & 0.589 & 0.534 & 0.190 & 0.515 & 0.549 & 0.782 & 0.032 & 0.551 & 0.552 & 0.921 \\
    250 & 0.069 & 0.647 & 0.651 & 0.913 & 0.910 & 0.185 & 0.381 & 0.423 & 0.485 & 0.135 & 0.367 & 0.392 & 0.755 & 0.036 & 0.394 & 0.395 & 0.918 \\
    500 & 0.045 & 0.340 & 0.343 & 0.936 & 0.934 & 0.150 & 0.256 & 0.297 & 0.454 & 0.103 & 0.251 & 0.271 & 0.805 & 0.047 & 0.283 & 0.287 & 0.943 \\
    1'000 & 0.049 & 0.213 & 0.219 & 0.935 & 0.933 & 0.147 & 0.183 & 0.234 & 0.337 & 0.092 & 0.179 & 0.202 & 0.845 & 0.071 & 0.202 & 0.214 & 0.933 \\
    8'000 & 0.007 & 0.070 & 0.070 & 0.948 & 0.947 & 0.150 & 0.067 & 0.165 & 0.017 & 0.034 & 0.067 & 0.075 & 0.899 & 0.082 & 0.078 & 0.113 & 0.789 \\
    \bottomrule\end{tabular}
    \begin{tablenotes}[flushleft]
    \item \textsc{Note}: The table depicts simulation results across $N=1'000$ draws obtained for the setting with
            one stochastic process. Except for the LP estimator, nuisance functions are estimated with random forests.
            For the DML estimator, we use 10-fold cross-fitting and set $k_T=T/10$. For sample size $T=125$, probabilities are winsorized at 1\%.
            The parameters of the data generating process are $n=12$,
            $\sigma_\epsilon=1.0$, $\gamma=0.6$, $p=2$,
            $q=1$, $\sigma_u=1.0$, $\alpha_A=0.3$,
            $\alpha_M=0.3$, $\rho_A=0.35$, $\rho_M=0.7$,
            $\beta_1=0.3$, $\beta_2=0.5$.
            $C_a$($\cdot$) and $C_b$($\cdot$) in the tables denote the coverage at the given confidence level using asymptotic and fixed-bandwidth critical values respectively.
    \end{tablenotes}
    \end{threeparttable}}
    \end{table}


Figure~\ref{fig:IRF} provides a visual summary of our results by depicting the true and estimated IRF. The regression adjustment estimator is biased, overestimating on average the true impact of the treatment $D_t$, in particular for longer horizons. When estimating the IRF with the doubly robust influence function (DR), the bias is reduced. Only when using an estimator that is Neyman orthogonal and uses cross-fitting (DML), the distribution of the estimated IRFs across simulation replications is centered around the true IRF. In the inset of the top panel of Figure~\ref{fig:IRF}, for the DML estimator, normalized biases as a function of $T$ are plotted for all forecast horizons $h$ and contrasted to the $\sqrt{T}$ scaling implied by Theorem~\ref{thm:ate}, showing that the bias of the DML estimator follows the $\sqrt{T}$ scaling implied by Theorem~\ref{thm:ate} quite well. The LP estimator finally exhibits a bias of similar magnitude as the RA estimator, as the linear estimator fails to capture the highly nonlinear and heterogeneous relation between $D_t$, $X_t$ and $Y_{t+h}$ in the DGP.


\begin{figure}[htbp]
\centering
\caption{Distribution of impulse response function estimates for a baseline nonlinear DGP with $n=12$, $\sigma_\epsilon=1.0$ and random forest nuisance function estimates}
  \includegraphics[trim={0 4cm 0 5cm},clip,width=\textwidth]{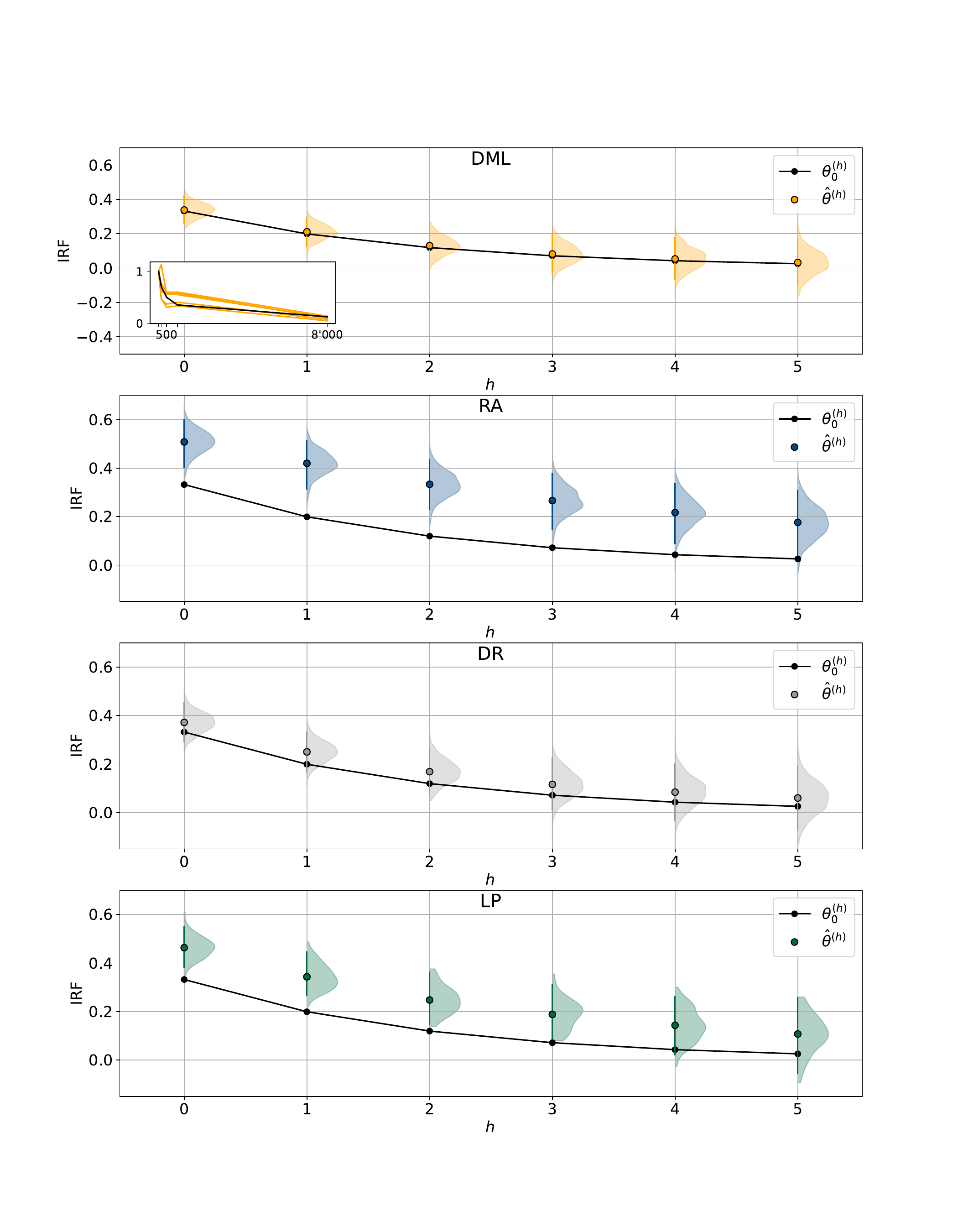}
  \label{fig:IRF}
  \begin{minipage}{\textwidth}
    \footnotesize\textsc{Note}: Comparison of the true $\theta_0^{(h)}$ with estimates $\hat\theta^{(h)}$ of the IRF obtained for the setting with one stochastic process of length $T=8'000$ from Table \ref{tab:result_baseline}.
    Except for the LP estimator, nuisance functions are estimated with random forests.
    For the DML estimator, we use 10-fold cross-fitting and set $k_T=T/10$.
    The parameters of the data generating process are $n=12$,
    $\sigma_\epsilon=1.0$, $\gamma=0.6$, $p=2$,
    $q=1$, $\sigma_u=1.0$, $\alpha_A=0.3$,
    $\alpha_M=0.3$, $\rho_A=0.35$, $\rho_M=0.7$,
    $\beta_1=0.3$, $\beta_2=0.5$.
    For individual $h$, we show kernel density estimates of the distribution of $\hat\theta^{(h)}$ across $N=1'000$ realizations. The dots indicate the average, and the vertical lines the (2.5\%, 97.5\%)-quantile range of the distribution. In the inset of the top panel, for the DML estimator, normalized biases as a function of $T$ are plotted for all forecast horizons $h$ and contrasted to the $\sqrt{T}$ scaling implied by Theorem~\ref{thm:ate} in black.
  \end{minipage}
\end{figure}


\section{Comparison to local projections}\label{sec:lp}
The simulation study in the previous section provided evidence that, in the presence of nonlinearities, local projections are asymptotically biased, while the proposed DML estimator consistently estimates the true impulse response function. Here, we contrast these two estimators and their underlying assumptions. For a comparison of local projections and VARs, see \cite{Plagborg2021}.

\subsection{An illustrative example}
Consider the following stochastic processes
\begin{equation}\label{eq:ex_dgp}
    \begin{aligned}
        Y_{t+h} &= f^{(h)}(D_t, X_t) + \epsilon_{t+h}\\
        D_t &= q(X_t, \eta_t),
    \end{aligned}
\end{equation}
where $\epsilon_{t+h}$ and $\eta_t$ are \emph{i.i.d.} noise terms with zero mean and finite variance. The functions $f^{(h)}(\cdot)$ and $q(\cdot)$ are measurable and possibly nonlinear. Let $X_t = (V_t, V_{t-1})'$ be a two-dimensional vector with $V_t = \phi V_{t-1} + u_t$, where $\vert\phi\vert < 1$ and $u_t$ is an \emph{i.i.d.}, mean zero random variable with finite variance. The quantity of interest is the impulse response function at horizon $h$, i.e. $\theta^{(h)}_0 = \E[\tau^{(h)}(X_t)]$ with $\tau^{(h)}(X_t) = \E[f^{(h)}(1, X_t)\vert X_t] - \E[f^{(h)}(0, X_t)\vert X_t]$. In the following, we illustrate how assumptions of local projection and DML estimators are satisfied in the example process \eqref{eq:ex_dgp}.

\subsubsection*{Local projection estimator}
Local projections estimate the impulse response function via the coefficient $\hat\beta^{(h)}$ in the linear regression model
\begin{equation*}
    Y_{t+h} = \hat\beta^{(h)} D_t + \hat\alpha' X_t + \hat{r}_{t+h}.
\end{equation*}
By the projection theorem, the population coefficient $\beta_0^{(h)}$ is given by
\begin{equation*}
    \beta_0^{(h)} = \frac{\E[f^{(h)}(D_t, X_t) (D_t - \lambda_{0}' X_t])}{\E[(D_t - \lambda_{0}' X_t)^2]},
\end{equation*}
where $\lambda_{0} = \arg\min\E[(D_t - \lambda' X_t)^2]$. Asymptotically, under certain regularity conditions on the time dependence and moments of the stochastic processes $S^{(h)}$, we have $\sqrt{T}(\hat\beta^{(h)} - \beta_0^{(h)}) \overset{d}{\rightarrow}\mathcal{N}(0, \Var(\beta_0^{(h)}))$. To ensure this result, $D_t$ and $X_t$ must be assumed to be stationary and ergodic for second moments, imposing restrictions on their time dependence, similar to -- although generally weaker then -- those in Assumption \ref{ass:oracle} for the DML estimator. For the process in \eqref{eq:ex_dgp}, this assumption is indeed fulfilled, as $D_t$ and $X_t$ are stationary and geometrically strong mixing.\footnote{Since $V_t$ is a linear process, it can be shown to be geometrically strong mixing with coefficients $\alpha(s) = O(|\phi|^{s})$ \citep[see Theorem 15.9 in][]{Davidson2021}.} Furthermore, $D_t$ and $X_t$ are assumed to be uncorrelated with the innovation $\epsilon_{t+h}$, which corresponds to Assumption \ref{ass:X_richness} for the DML estimator. Lastly, similar to Assumptions \ref{ass:oracle} and \ref{ass:sample_dependence}, the first four moments of $D_t$, $X_t$, and $\epsilon_{t+h}$ need to be finite. Note that since the linear regression estimator converges at the $\sqrt{T}$-rate, Assumption \ref{ass:risk_decay} is satisfied, even though this is not required for $\hat\beta^{(h)}$ to converge to $\beta_0^{(h)}$.

\subsubsection*{DML estimator}
Under Theorem \ref{thm:ate}, the DML estimator consistently recovers $\theta_0^{(h)}$ provided that Assumptions \ref{ass:constant_ate} and \ref{ass:oracle}-\ref{ass:sample_dependence} hold. Assumption \ref{ass:constant_ate} is met by construction of the example process \eqref{eq:ex_dgp}. Since $V_t$ is geometrically strong mixing and all involved functions are measurable, it follows that $S^{(h)}$ and $g(Z^{(h)}_t, h; \Gamma_0)$ are also geometrically strong mixing \citep[see Theorem 15.1 in][]{Davidson2021} and thus Assumption \ref{ass:oracle} is satisfied. Assumption \ref{ass:risk_decay} requires that the $L_2$-norm of the estimation error of the nuisance function estimators converges to zero at least at rate $T^{1/4}$. For the example process \eqref{eq:ex_dgp} this requirement is met for example by random forests, which converge at least at rate $T^{1/3}$ \citep{Davis2020}. Neural networks would also satisfy this condition, provided that the functions $f^{(h)}(\cdot)$ and $q(\cdot)$ are sufficiently smooth \citep[for more details, see][]{Ma2022}. Lasso can achieve an even faster rate of $T^{1/2}$ \citep{Wong2020}, provided that the nuisance functions can be well approximated by polynomials of the conditioning variables. Assumption \ref{ass:X_richness} requires the set of covariates to be sufficiently rich, so that past information on $Z_t^{(h)}$ cannot predict $Y_{t+h}$ and $D_t$. For the example process \eqref{eq:ex_dgp}, this assumption holds directly if the covariates $X_t$ include both $V_t$ and $V_{t-1}$. This would still hold if additional covariates or lagged values of $Z_t^{(h)}$ were included in the estimation. Assumption \ref{ass:sample_dependence} finally holds for any value of $\psi$ in the setting of \eqref{eq:ex_dgp}, since $S$ is geometrically strong mixing. Therefore, it is also sufficient that the residuals and estimation errors possess finite $4 + \nu$ moments, for some $\nu > 0$.

To conduct inference, Assumption \ref{ass:variance_estimator} additionally requires choosing a bandwidth $m_T$ depending on the convergence rate of the nuisance function estimators. For the example process \eqref{eq:ex_dgp}, random forests estimators would permit $m_T$ to be $o(T^{1/3})$ \citep[e.g.][]{NeweyWest1994}. Using Lasso estimators, $m_T$ can be $o(T^{1/2})$ \citep[e.g.][]{Lazarus2018}.

\subsection{Linear and nonlinear processes}
In case the function $f^{(h)}(D_t,X_t)$ is linear, then it can be shown that $\beta_0^{(h)} = \theta^{(h)}_0$ and thus linear projections recover the true impulse response function. However, if $f^{(h)}(D_t,X_t)$ is a nonlinear function, then the coefficients $\beta_0^{(h)}$ will no longer recover the true impulse response function, but an expected, weighted average of marginal effects $\tau^{(h)}(x)$. For an extensive discussion and derivation of these results, refer to \cite{kolesár2024dynamiccausaleffectsnonlinear}.

These effects are illustrated in Figure \ref{fig:lp_dml_comparison}, which compares the distribution of impulse response function estimates obtained with DML and local projection estimators for three different DGPs: a nonlinear process, a linear process with interactions, and a purely linear process. Detailed results are deferred to Tables \ref{tab:results_linear_constant} and \ref{tab:results_linear_interaction} in the Online Appendix. As already seen in Section \ref{sec:simulation} and again shown in Figure \ref{fig:non_linear_dgp}, estimation using local projections is biased on a nonlinear process, whereas the DML estimator recovers $\theta_0^{(h)}$. This also holds in a setting where the outcome variable is generated from a linear function, but there is interaction between $X_t$ and $D_t$ (see Figure \ref{fig:linear_dgp_interactions}). Local projections only recover $\theta_0^{(h)}$ for a purely linear process where its functional form is correctly specified (see Figure \ref{fig:linear_dgp_constant}). In this case, local projections have lower finite sample bias and variance than the DML estimator.

\begin{figure}[htbp]
    \caption{Comparison of the distribution of DML and LP impulse response function estimates for nonlinear and linear DGPs}
    \centering
    \begin{subfigure}[b]{0.333\textwidth}
        \includegraphics[width=\textwidth]{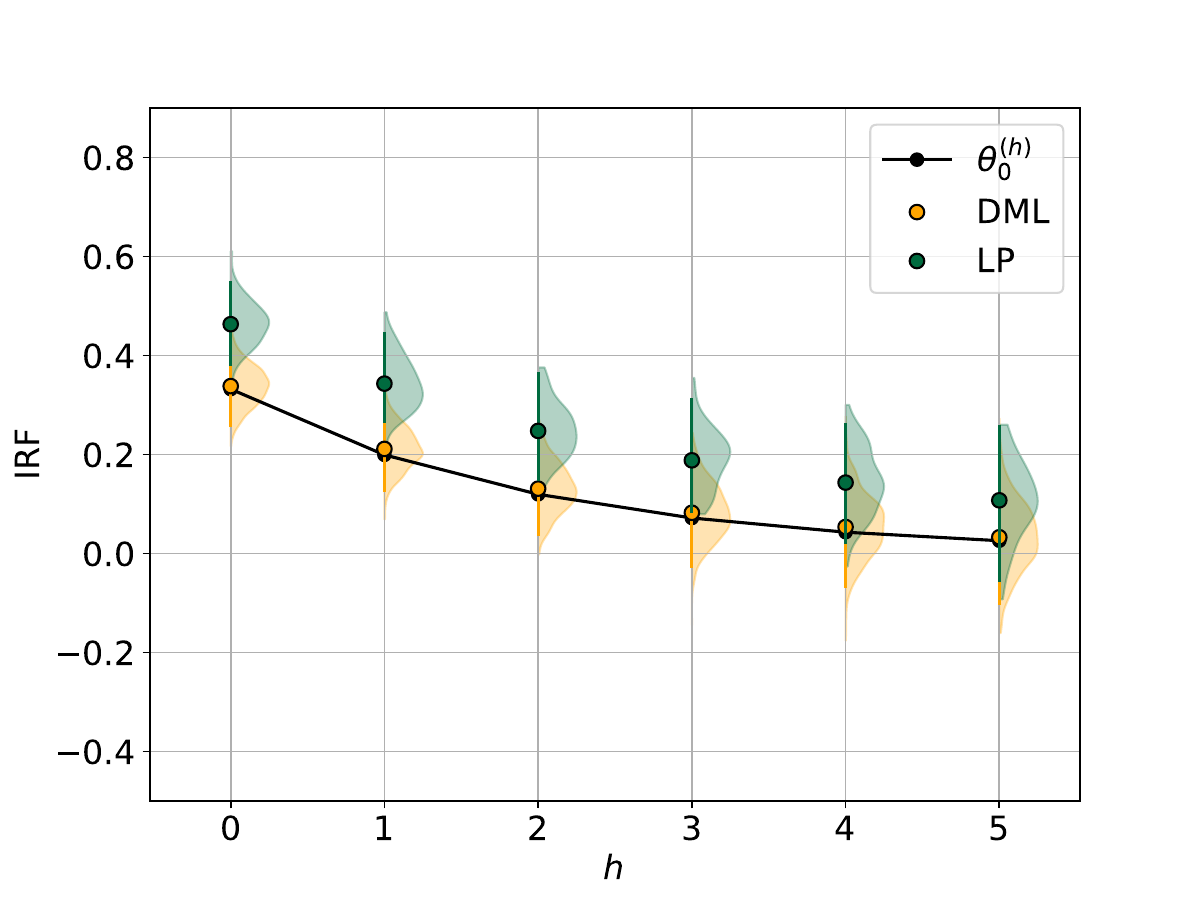}
        \caption{Nonlinear DGP}
        \label{fig:non_linear_dgp}
    \end{subfigure}\hfill
    \begin{subfigure}[b]{0.333\textwidth}
        \includegraphics[width=\textwidth]{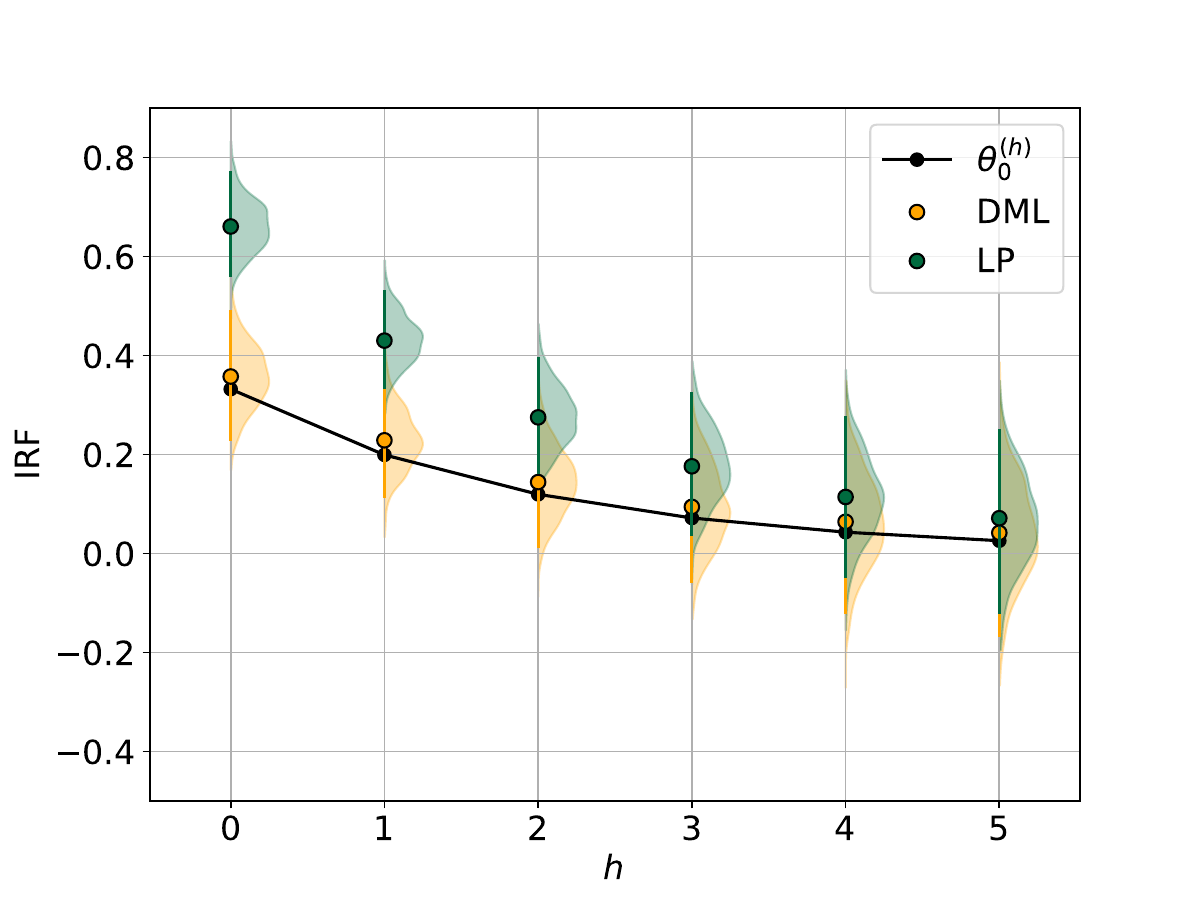}
        \caption{Linear DGP with interactions}
        \label{fig:linear_dgp_interactions}
    \end{subfigure}\hfill
    \begin{subfigure}[b]{0.333\textwidth}
        \includegraphics[width=\textwidth]{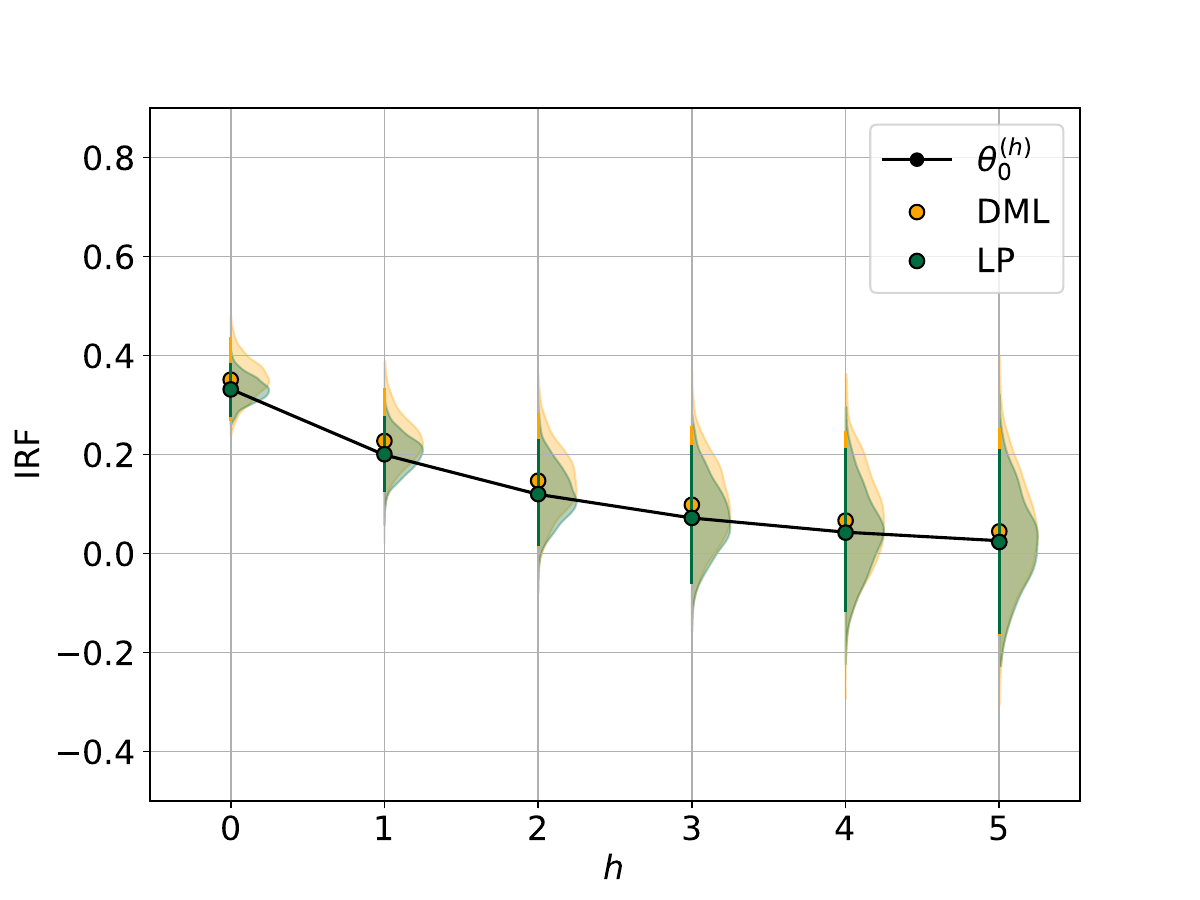}
        \caption{Linear DGP}
        \label{fig:linear_dgp_constant}
    \end{subfigure}
    \label{fig:lp_dml_comparison}
    \begin{minipage}{\textwidth}
        \footnotesize\textsc{Note}: The figure compares the true $\theta_0^{(h)}$ with estimates $\hat\theta^{(h)}$ of the IRF obtained for the setting with one stochastic process of length $T=8'000$ generated from three different DGPs: the nonlinear DGP studied in Section \ref{sec:simulation} (Figure \ref{fig:non_linear_dgp}), a linear DGP with interaction terms where $b(X_t) = 0.5 \sum_{i=1}^5X_{i,t}$ and $\tau(X_t) = \theta_0^{(0)} + \sum_{i=1}^3X_{i,t} - \sum_{i=4}^5X_{i,t}$ (Figure \ref{fig:linear_dgp_interactions}), and a linear DGP with $b(X_t) = 0.5 \sum_{i=1}^5X_{i,t}$ and $\tau(X_t) = \theta_0^{(0)}$ (Figure \ref{fig:linear_dgp_constant}).
        For the DML estimator, nuisance functions are estimated with random forests, we use 10-fold cross-fitting and set $k_T=T/10$.
        The parameters of the data generating processes are $n=12$,
        $\sigma_\epsilon=1.0$, $\gamma=0.6$, $p=2$,
        $q=1$, $\sigma_u=1.0$, $\alpha_A=0.3$,
        $\alpha_M=0.3$, $\rho_A=0.35$, $\rho_M=0.7$,
        $\beta_1=0.3$, $\beta_2=0.5$.
        For individual $h$, we show kernel density estimates of the distribution of $\hat\theta^{(h)}$ across $N=1'000$ realizations. The dots indicate the average, and the vertical lines the (2.5\%, 97.5\%)-quantile range of the distribution.
      \end{minipage}
\end{figure}

\section{Empirical Application}\label{sec:empirical}
As an illustration, we apply the proposed methodology to the empirical study conducted in \cite{Angrist2018}. Based on the same data set, we revisit the estimation of the effect of U.S. monetary policy decisions on macroeconomic aggregates using modern machine learning estimators. Our monthly observations cover the period from July 1989 to December 2008 and we estimate the effect of federal funds target rate changes on a set of macroeconomic outcome variables. As predictors, we consider the same futures-based expectation measure for the federal funds rate as in \cite{Angrist2018}, as well as the level of the target rate at the end of the prior month and its change, a scale factor that accounts for when within the month the Federal Reserve's Open Market Committee (FOMC) meeting was scheduled, dummies for months with a scheduled FOMC meeting, as well as measures for inflation and unemployment (including lagged values).\footnote{This corresponds to the model specification labelled OP$_{F2}$ in \cite{Angrist2018}.} Compared to the setup in \cite{Angrist2018}, we make the following modifications to accommodate for the change in estimation technique from linear models to flexible nonparametric machine learners. First, we exclude dummies for monthly fixed effects and special events like Y2K and the September 11, 2001 attacks, because IRF estimates are based on out-of-sample predictions in our approach. Second, we drop manually constructed interaction variables, as machine learning estimators are able to infer these effects from the data, if they are present. Third, we include up to four lags of inflation, unemployment rate and of the target variable. Our treatment variable $D_t$ can assume one of five discrete values $d\in\{-0.5\%, -0.25\%, \allowbreak 0.0\%, 0.25\%, 0.5\%\}$. The propensity score model for $e_0(d,x) = \Pr(D_t=d|X_t=x)$ is implemented as an ordinal classification \citep{OrdinalClassification2001}. To estimate the conditional mean nuisance functions $\mu_0(d, x, h)$, we include the treatments as dummy variables in $X_t$ and estimate/tune one joint model for all types of treatments (an S-learner in the terminology of \cite{Kuenzel2019}). As in our simulation experiments in Section \ref{sec:simulation}, we explore random forests and a gradient boosted trees algorithm as estimators for the nuisance functions. Details on tuning of the machine learners are provided in the Online Appendix~\ref{app:tuning_empirical}. We apply our cross-fitting approach with $K=10$. In line with \cite{Angrist2018}, we estimate impulse responses up to $H=24$ months and given the limited sample size set $k_T = 24$ to retain as much data as possible for estimation. As advocated in Section \ref{sec:cross-fitting}, we remove the $k_T$ observations from the data used to estimate the machine learners.

We present estimated IRFs for the federal funds rate and the unemployment rate in Figure~\ref{fig:IRF_FFED_UNRATE}. Additionally, estimates for the effects on the bond yield curve are provided in Figure~\ref{fig:IRF_YieldCurve} in the Online Appendix. Predictive performance of both types of machine learners explored are comparable, but random forests appear to produce slightly smoother impulse responses, which, similar to standard local projection estimation \citep{Barnichon2019}, is advantageous. We thus focus on results using random forests and for 25 basis point changes of the target rate.\footnote{Additional results are available from the authors on request.} Overall, we identify similar dynamics in the outcome variables as in \cite{Angrist2018}. However, for the federal funds rate, we find a larger absolute effect for target rate decreases of around 50 basis points that remains significant at the 5\% level until around one year after the target rate decrease. In comparison, the peak effect of a target rate increase is around one percentage point, occurring one year after the increase, though it is accompanied by greater uncertainty. Furthermore, we do not find that either expansionary or tightening monetary policy has a significant effect on the unemployment rate. Looking at the effects on the yield curve, in line with \cite{Angrist2018}, changes in the federal funds target rate have a higher initial impact on short-term yields than on long-term yields, as expected. Moreover, significant effects are observed only for shorter tenures. In contrast to \cite{Angrist2018}, however, our estimates do not suggest that term rates are less sensitive to policy accommodation than to tightening. Finally, the empirical application also provides evidence that even in settings with limited sample sizes commonly encountered in macroeconomic studies, sufficiently accurate estimates of the nonparametric nuisance functions can be obtained in order to produce IRF estimates comparable to ones obtained with conventional techniques, without having to construct all of the (interaction) variables and nonlinearities manually.
\begin{figure}[htbp]
    \caption{Estimated cumulative effects of target rate changes on the federal funds rate and the unemployment rate}
    \centering
    \begin{subfigure}[b]{\textwidth}
        \caption{Federal Funds Rate}
        \includegraphics[scale=0.30]{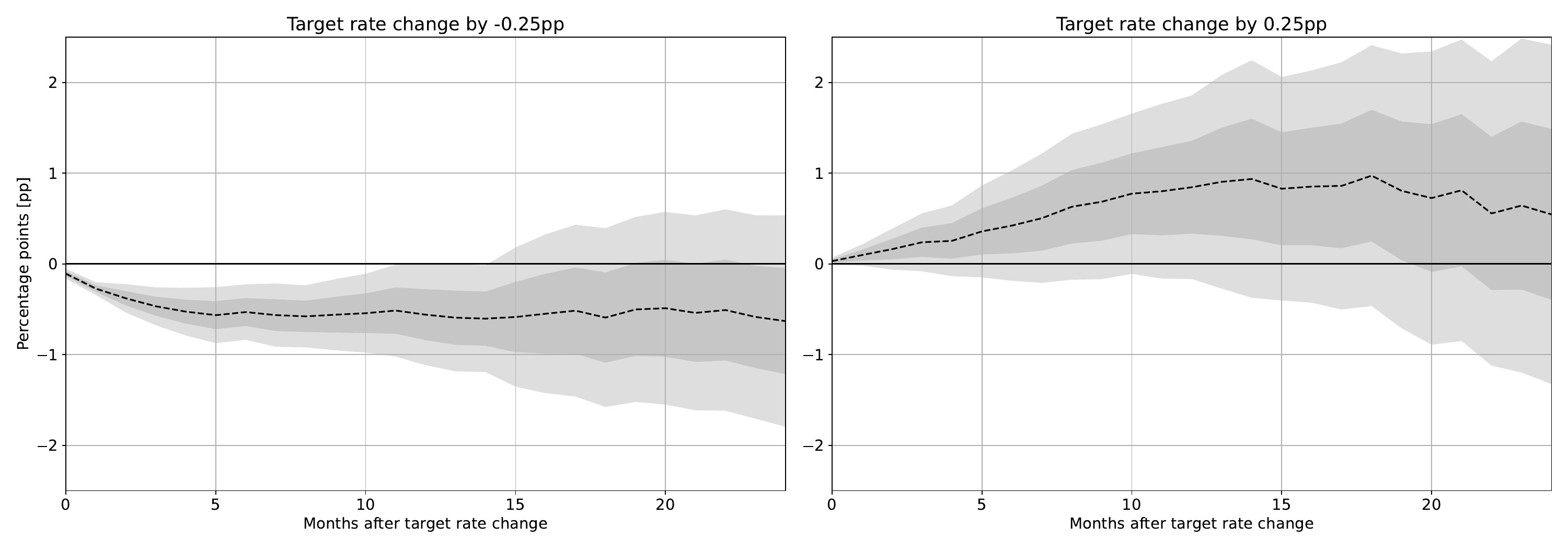}
        \label{fig:IRF_FFED}
    \end{subfigure}
    \begin{subfigure}[b]{\textwidth}
        \caption{Unemployment Rate}
        \includegraphics[scale=0.30]{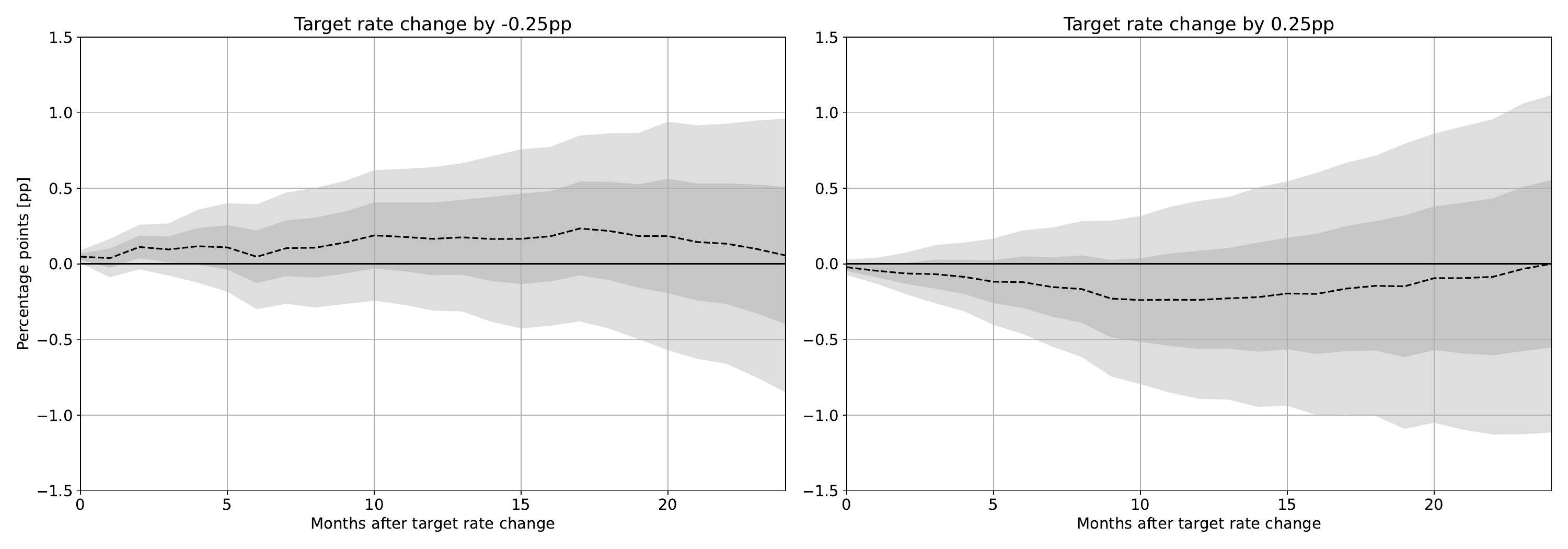}
        \label{fig:IRF_UNRATE}
    \end{subfigure}
    \label{fig:IRF_FFED_UNRATE}
    \begin{minipage}{\textwidth}
        \footnotesize\textsc{Note}: The figure shows the estimated cumulative effects of target rate changes on (a) the federal funds rate and (b) the unemployment rate for the time period July 1989 to December 2008. The left (right) column shows the effect of decreasing (increasing) the target rate by 25 basis points. The nuisance functions are estimated by random forests using 10-fold cross-fitting removing $k_T=24$ observations from the estimation sample at the boundary to the inference sample. The shaded areas represent 68\% and 95\% confidence intervals  with fixed-bandwidth critical values \citep{Kiefer2005}. The variances are estimated using bandwidths determined by the procedure of \cite{NeweyWest1994}.
      \end{minipage}
\end{figure}

\section{Conclusion}\label{sec:conclusion}
We have shown how to adopt recent ideas from the causal inference framework to flexibly estimate IRFs. This presents a novel estimator that can rely on fully nonparametric relations between treatment and outcome variables, opening up the possibility to use flexible machine learning approaches to estimate IRFs. Our theoretical results outline conditions for this estimator to be consistent and asymptotically normally distributed at the parametric rate. Simulations where a highly nonlinear time series is treated over time corroborate these results. Alternative estimators often used in practice estimate the IRF with a larger bias and fail to allow valid inference. Finally, we have illustrated the proposed methodology empirically by applying it to the estimation of the effects of macroeconomic shocks, allowing us to estimate IRFs of U.S. monetary policy decisions on macroeconomic aggregates using modern machine learning estimators. For future work, several semiparametric techniques available in the \emph{i.i.d.} setting could be extended to time series settings in order to develop our approach further. This includes approaches for continuous treatments \citep{Colangelo2023}, instrumental variables \citep{Chernozhukov2018}, the estimation of other moments \citep{Chernozhukov2022,Chernozhukov2022b}, or conditional treatment effects \citep{Zimmert2019,Fan2022,Semenova2020,KennedyTowards2023} to estimate generalized IRFs. In general, future research could extend the theoretical results developed in this paper to a broader class of estimands relying on linear scores \citep{Chernozhukov2018}.

\clearpage

\begin{appendix}
\renewcommand{\theequation}{A\arabic{equation}}
\setcounter{equation}{0}

\section*{Proofs}\label{app:proofs}
\begin{proof}[Proof of Theorem \ref{thm:identification}]
    We have that
    \begin{align*}
        \E\left[Y_{t+h}(d)\right] &= \E\left[\E\left[Y_{t+h}(d)\vert D_t=d, X_t\right]\right]= \E\left[\E\left[Y_{t+h}\vert D_t=d, X_t\right]\right] =\E\left[\mu_0(d,X_t,h)\right],
    \end{align*}
    where the second equality follows from Assumption \ref{ass:identification}.\ref{ass:identification_cia} and the last equality from Assumption \ref{ass:identification}.\ref{ass:identification_sutva}. It thus follows that $\E\left[Y_{t+h}(1)-Y_{t+h}(0)\right] = \E\left[\mu_0(1,X_t,h)-\mu_0(0,X_t,h)\right]$.
\end{proof}

\begin{proof}[Proof of Theorem \ref{thm:oracle}]
    Under Assumption \ref{ass:oracle} the stochastic process $\mathcal{G}^{(h)}$ satisfies the assumptions for the Central Limit Theorem for $\alpha$-mixing processes \citep{herrndorf1984}.
\end{proof}

\begin{proof}[Proof of Theorem \ref{thm:ate}]
    The proof follows a similar strategy to the one in \cite{Wager2022} and \cite{Chernozhukov2018} for the \emph{i.i.d.} case. For the sake of legibility, we ease the notation and drop the reference to the forecast horizon $h$. Let $\tilde\theta$ be the oracle IRF estimator as defined in Theorem \ref{thm:oracle}. Denote by $\mathcal{E}_T$ the event that $(\hat{\mu}_{S_{-i}}(d,x,h), \hat{e}_{S_{-i}}(x)) \in \Xi_{T}$ for all $i=1, ..., K$. We have that
    \begin{align*}
        \sqrt{T} (\hat\theta - \theta_0) &= \sqrt{T} (\hat\theta - \tilde\theta) + \sqrt{T}(\tilde\theta - \theta_0)\\
        &= \sqrt{T} \left( \sum_{i=1}^{K}\frac{|\mathcal{T}_i|}{T} \hat{\theta}_{S_i} - \tilde\theta\right) + \sqrt{T}(\tilde\theta - \theta_0)\\
        &=\sum_{i=1}^{K}\sqrt{T} \left(\frac{|\mathcal{T}_i|}{T} \hat{\theta}_{S_i} - \frac{|\mathcal{T}_i|}{T} \tilde\theta_{S_i}\right) + \sqrt{T}(\tilde\theta - \theta_0),
    \end{align*}
    where $\tilde\theta_{S_i} = \vert\mathcal{T}_i\vert^{-1} \sum_{t\in \mathcal{T}_i}g(Z_t; \Gamma_0)$. We have to show that the summation converges to zero in probability. Note that since $K$ is a finite integer, it suffices to show convergence for one summand. We begin by expanding a summand for some arbitrary $i$ as
    \begin{align*}
        \frac{|\mathcal{T}_i|}{T} \hat{\theta}_{S_i} - \frac{|\mathcal{T}_i|}{T}  \tilde\theta_{S_i} &= \frac{1}{|\mathcal{T}_i|} \sum_{t \in \mathcal{T}_i} g(Z_t; \hat\Gamma_{S_{-i}}) - g(Z_t;\Gamma_0) \\
        &= \frac{1}{|\mathcal{T}_i|} \sum_{t \in \mathcal{T}_i} \bigg(\hat{\mu}_{S_{-i}}(1,X_t) - \mu_0(1,X_t) + \frac{D_t}{\hat{e}_{S_{-i}}(X_t)}(Y_t - \hat{\mu}_{S_{-i}}(1,X_t)) \\
        &\qquad\qquad\qquad\qquad - \frac{D_t}{{e}_0(X_t)}(Y_t - {\mu}_0(1,X_t))\bigg)\\
        &\quad -\frac{1}{|\mathcal{T}_i|} \sum_{t \in \mathcal{T}_i} \bigg(\hat{\mu}_{S_{-i}}(0,X_t) - \mu_0(0,X_t) + \frac{1-D_t}{1-\hat{e}_{S_{-i}}(X_t)}(Y_t - \hat{\mu}_{S_{-i}}(0,X_t)) \\
        &\qquad\qquad\qquad\qquad - \frac{1-D_t}{1-{e}_0(X_t)}(Y_t - {\mu}_0(0,X_t))\bigg).
    \end{align*}
    We will prove convergence for the first summation, the second summation can be treated analogously. The first summation can be decomposed as
    \begin{align*}
        \frac{1}{|\mathcal{T}_i|} \sum_{t \in \mathcal{T}_i} & \left(\muError + \frac{D_t}{\hat{e}_{S_{-i}}(X_t)}(Y_t - \hat{\mu}_{S_{-i}}(1,X_t)) - \frac{D_t}{{e}_0(X_t)}(Y_t - {\mu}_0(1,X_t))\right)\\
        =&\underbrace{\frac{1}{|\mathcal{T}_i|} \sum_{t \in \mathcal{T}_i} \left( \muError \right) \left(1-\frac{D_t}{e_0(X_t)} \right)}_{= P_1}\\
        &\quad + \underbrace{\frac{1}{|\mathcal{T}_i|} \sum_{t \in \mathcal{T}_i} D_t \left(Y_t - \mu_0(1,X_t)\right) \left(\invpError \right)}_{= P_2}\\
        &\quad + \underbrace{\frac{1}{|\mathcal{T}_i|} \sum_{t \in \mathcal{T}_i} D_t \left(\muError\right) \left(\invpError\right)}_{= P_3}.
    \end{align*}
    We will show that $P_k = o_p(T^{-1/2})$ for $k\in\{1,2,3\}$.

    Term $P_1$: From the squared $L_2$-norm of $P_1$ we have that
    \begin{align*}
        \E\Bigg[ &\Bigg(\frac{1}{|\mathcal{T}_i|} \sum_{t \in \mathcal{T}_i} \left( \muError \right) \left(1-\frac{D_t}{e_0(X_t)} \right) \Bigg)^2 \Bigg]\\
        &= \frac{1}{|\mathcal{T}_i|^2}\E\left[\E\left[\left( \sum_{t \in \mathcal{T}_i} \left( \muError \right) \left(1-\frac{D_t}{e_0(X_t)} \right) \right)^2 \Bigg| S_{-i} \right] \right]\\
        &= \frac{1}{|\mathcal{T}_i|^2}\E\Bigg[ \sum_{t \in \mathcal{T}_i} \sum_{s \in \mathcal{T}_i} \E \Bigg[\left( \muError \right) \left(1-\frac{D_t}{e_0(X_t)} \right) \times\\
        &\phantom{= \frac{1}{|\mathcal{T}_i|^2}\E\Bigg[ \sum_{t \in \mathcal{T}_i} \sum_{s \in \mathcal{T}_i} \E \Bigg[} \left( \hat{\mu}_{S_{-i}}(1,X_s) - \mu_0(1,X_s) \right) \left(1-\frac{D_s}{e_0(X_s)} \right) \Bigg| S_{-i} \Bigg] \Bigg]\\
        &=\frac{1}{|\mathcal{T}_i|^2} \sum_{t \in \mathcal{T}_i} \E\left[\left( \muError \right)^2 \left(\frac{1}{e_0(X_t)} -1\right) \right]\\
        &\leq \frac{1}{\eta |\mathcal{T}_i|^2 } \sum_{t \in \mathcal{T}_i} \E\left[\left( \muError \right)^2 \right] = \frac{o_p(1)}{|\mathcal{T}_i|} = o_p(T^{-1}).
    \end{align*}
    The third equality follows from the fact that the sum has mean zero and the inequality at the end follows from Assumption \ref{ass:risk_decay}.\ref{ass:risk_decay_prob}. The step from the second to the third equality follows from Assumption \ref{ass:X_richness}, which gives for $t<s$ that
    \begin{align*}
        \E \Bigg[&\left( \muError \right) \left(1-\frac{D_t}{e_0(X_t)} \right) \left( \hat{\mu}_{S_{-i}}(1,X_s) - \mu_0(1,X_s) \right) \left(1-\frac{D_s}{e_0(X_s)} \right) \Bigg| S_{-i} \Bigg]\\
        &= \E \Bigg[\left( \muError \right) \left(1-\frac{D_t}{e_0(X_t)} \right) \left( \hat{\mu}_{S_{-i}}(1,X_s) - \mu_0(1,X_s) \right)\times\\
        &\phantom{= \E \Bigg[}\;\underbrace{\E\left[1-\frac{D_s}{e_0(X_s)} \Bigg| X_s, \{(X_u, Y_u, D_u) : u\in\mathcal{T}_i, u < s)\}, S_{-i} \right]}_{=0} \Bigg| S_{-i} \Bigg]= 0.
    \end{align*}
    The same argument can be made for $t>s$. Convergence in the last step finally follows from the fact that $K$ is a fixed and finite integer and thus $\lim_{T\to\infty}|\mathcal{T}_i|=\lim_{T\to\infty} T/K=\infty$ for all $i$ and by noting that conditional on $S_{-i}$, the nuisance function estimator is non-stochastic, and thus conditional on the event $\mathcal{E}_T$, we have that
    \begin{align*}
        \sup_{t\in\mathcal{T}_i} \E\left[\left( \muError \right)^2 \Bigg\vert S_{-i} \right] & \leq\sup_{t\in\mathcal{T}_i}  \sup_{\mu \in \Xi_{T}} \E\left[\left(\mu(1,X_t)-\mu_0(1,X_t)\right)^2  \Bigg\vert S_{-i}\right]\\
        & \leq \sup_{t\in\mathcal{T}_i} \sup_{\mu \in \Xi_{T}} \norm{\mu(1,X_t)-\mu_0(1,X_t)}_2^2 = o_p(1)
    \end{align*}
    since by Assumption \ref{ass:risk_decay} $\sup_{t\in\mathcal{T}_i} \sup_{\mu \in \Xi_{T}} \norm{\mu(D_t,X_t)-\mu_0(D_t,X_t)}_2^2 = (r_{\mu,T})^2 = o_p(1)$. By Lemma 6.1 in \cite{Chernozhukov2018} it follows that $ \sup_{t\in\mathcal{T}_i} \E[( \hat{\mu}_{S_{-i}}(1,X_t) \allowbreak - \mu_0(1,X_t) )^2] = o_p(1)$, and we thus conclude that $P_1$ is $o_p(T^{-1/2})$.

    Term $P_2$: Similarly, from the squared $L_2$-norm of $P_2$ we have that
    \begin{align*}
        \E\Bigg[&\Bigg(\frac{1}{|\mathcal{T}_i|} \sum_{t \in \mathcal{T}_i} D_t \left(Y_t - \mu_0(1,X_t)\right) \left(\invpError \right)   \Bigg)^2\Bigg]\\
        &= \frac{1}{|\mathcal{T}_i|^2} \E\Bigg[\E\Bigg[\Bigg( \sum_{t \in \mathcal{T}_i} D_t \left(Y_t - \mu_0(1,X_t)\right) \left(\invpError \right)   \Bigg)^2 \Bigg| S_{-i}\Bigg]\Bigg]\\
        &= \frac{1}{|\mathcal{T}_i|^2} \E\Bigg[ \sum_{t \in \mathcal{T}_i} \sum_{s \in \mathcal{T}_i} \E\Bigg[ D_t \left(Y_t - \mu_0(1,X_t)\right) \left(\invpError \right) \times\\
        &\phantom{= \frac{1}{|\mathcal{T}_i|^2} \E\Bigg[ \sum_{t \in \mathcal{T}_i} \sum_{s \in \mathcal{T}_i} \E\Bigg[ } D_s \left(Y_s - \mu_0(1,X_s)\right) \left(\frac{1}{\hat{e}_{S_{-i}}(X_s)} -\frac{1}{e_0(X_s)}\right) \Bigg| S_{-i}\Bigg]\Bigg]\\
        &= \frac{1}{|\mathcal{T}_i|^2} \sum_{t \in \mathcal{T}_i} \E\Bigg[  D_t \left(Y_t - \mu_0(1,X_t)\right)^2 \left(\invpError \right)^2 \Bigg]\\
        &\leq \frac{1}{|\mathcal{T}_i|^2} \frac{\epsilon^2_1}{\eta^2} \sum_{t \in \mathcal{T}_i} \E\left[\left( \hat{e}_{S_{-i}}(X_t) - e_0(X_t) \right)^2 \right] = \frac{o_p(1)}{|\mathcal{T}_i|} = o_p(T^{-1}).\\
    \end{align*}
    The third equality follows from the fact that the sum has mean zero, and the inequality follows from Assumptions \ref{ass:risk_decay}.\ref{ass:risk_decay_prob} and \ref{ass:risk_decay}.\ref{ass:risk_decay_variance}. The crucial step is again to establish that the variance of the sum equals the sum of the variances (from the second to the third equality). This follows from Assumption \ref{ass:X_richness}, which gives that whenever $t<s$, it holds that
    \begin{align*}
        \E&\Bigg[ D_t \left(Y_t - \mu_0(1,X_t)\right) \left(\invpError \right) \times \\
        &\phantom{\E\Bigg[ D_t }\;D_s \left(Y_s - \mu_0(1,X_s)\right) \left(\frac{1}{\hat{e}_{S_{-i}}(X_s)} -\frac{1}{e_0(X_s)}\right) \Bigg| S_{-i}\Bigg]\\
        &= \E\Bigg[ D_t \left(Y_t - \mu_0(1,X_t)\right) \left(\invpError \right) \times \\
        &\phantom{\E\Bigg[ D_t }\;\underbrace{\E\left[D_s(Y_s - \mu_0(1,X_s))\Bigg|X_s, \{(X_u, Y_u, D_u): u \in \mathcal{T}_i, u < s \}, S_{-i} \right]}_{=0}  \left(\frac{1}{\hat{e}_{S_{-i}}(X_s)} -\frac{1}{e_0(X_s)}\right)  \Bigg| S_{-i}\Bigg]=0,
    \end{align*}
    and the same argument can be made for $t>s$. Convergence in the last step finally follows from the fact that $K$ is a fixed and finite integer and thus $\lim_{T\to\infty}|\mathcal{T}_i|=\lim_{T\to\infty} T/K=\infty$ for all $i$ and by noting that conditional on $S_{-i}$, the nuisance function estimator is non-stochastic, and thus conditional on the event $\mathcal{E}_T$, we have that
    \begin{align*}
        \sup_{t\in\mathcal{T}_i} \E\left[\left( \hat{e}_{S_{-i}}(X_t) - e_0(X_t) \right)^2 \Bigg\vert S_{-i} \right] & \leq\sup_{t\in\mathcal{T}_i}  \sup_{e \in \Xi_{T}}  \E\left[\left(e(X_t) - e_0(X_t)\right)^2  \Bigg\vert S_{-i}\right]\\
        & \leq \sup_{t\in\mathcal{T}_i} \sup_{e \in \Xi_{T}} \norm{e(X_t) - e_0(X_t)}_2^2 = (r_{e,T})^2
    \end{align*}
    by definition of the rate $r_{e,T}$ in Assumption \ref{ass:risk_decay}. By Lemma 6.1 in \cite{Chernozhukov2018} and Assumption \ref{ass:risk_decay} it follows that $ \sup_{t\in\mathcal{T}_i} \E\left[\left( \hat{e}_{S_{-i}}(X_t) - e_0(X_t)  \right)^2 \right] = o_p(1)$, and we thus conclude that $P_2$ is $o_p(T^{-1/2})$.

    Term $P_3$: Finally, from the $L_1$-norm of $P_3$ we get that
    \begin{align*}
        \E\Bigg[\Bigg| & \frac{1}{|\mathcal{T}_i|} \sum_{t \in \mathcal{T}_i} D_t \left(\muError\right) \left(\invpError\right) \Bigg|\Bigg]\\
        &\leq \E\Bigg[ \frac{1}{|\mathcal{T}_i|} \sum_{t \in \mathcal{T}_i} D_t \left|\muError\right| \left|\invpError\right| \Bigg]\\
        &= \frac{1}{|\mathcal{T}_i|} \sum_{t \in \mathcal{T}_i} \E\Bigg[  D_t \left|\hat\mu_{S_{-i}}(D_t,X_t) - \mu_0(D_t,X_t)\right| \left|\invpError\right| \Bigg]\\
        &\leq \frac{1}{\eta}\frac{1}{|\mathcal{T}_i|} \sum_{t \in \mathcal{T}_i} \E\Bigg[ \left|\hat\mu_{S_{-i}}(D_t,X_t) - \mu_0(D_t,X_t)\right| \left|\hat{e}_{S_{-i}}(X_t) - e_0(X_t)\right| \Bigg] = \frac{o_p(1)}{T^{1/2}},
    \end{align*}
    where the last inequality follows from Assumption \ref{ass:risk_decay}.\ref{ass:risk_decay_prob}.  Convergence in the last equality finally follows from the fact that $K$ is a fixed and finite integer and thus $\lim_{T\to\infty}|\mathcal{T}_i|=\lim_{T\to\infty} T/K = \infty$ for all $i$ and by noting that conditional on $S_{-i}$, the nuisance function estimators are non-stochastic, and thus conditional on the event $\mathcal{E}_T$, we have that
    \begin{align*}
        &\sup_{t\in\mathcal{T}_i} \E\Bigg[ \left|\hat\mu_{S_{-i}}(D_t,X_t) - \mu_0(D_t,X_t)\right| \left|\hat{e}_{S_{-i}}(X_t) - e_0(X_t)\right| \Bigg\vert S_{-i} \Bigg]\\
        &\;\leq\sup_{t\in\mathcal{T}_i} \E\Bigg[ \left|\hat\mu_{S_{-i}}(D_t,X_t) - \mu_0(D_t,X_t)\right|^2\Bigg\vert S_{-i} \Bigg]^{1/2} \sup_{t\in\mathcal{T}_i} \E\Bigg[\left|\hat{e}_{S_{-i}}(X_t) - e_0(X_t)\right|^2 \Bigg\vert S_{-i} \Bigg]^{1/2}\\
        &\; \leq\sup_{t\in\mathcal{T}_i} \sup_{\mu \in \Xi_{T}}\E\Bigg[ \left|\mu(D_t,X_t)-\mu_0(D_t,X_t)\right|^2\Bigg\vert S_{-i} \Bigg]^{1/2} \sup_{t\in\mathcal{T}_i} \sup_{e \in \Xi_{T}}\E\Bigg[\left|e(X_t) - e_0(X_t)\right|^2 \Bigg\vert S_{-i} \Bigg]^{1/2}\\
        &\;\leq \sup_{t\in\mathcal{T}_i} \sup_{\mu \in \Xi_{T}} \norm{\mu(D_t,X_t)-\mu_0(D_t,X_t)}_2 \sup_{t\in\mathcal{T}_i} \sup_{e \in \Xi_{T}} \norm{e(X_t) - e_0(X_t)}_2 = r_{\mu,T} \cdot r_{e,T}
    \end{align*}
    by Cauchy-Schwarz and the definition of the rates $r_{\mu,T}$ and $r_{e,T}$ in Assumption \ref{ass:risk_decay}. By Lemma 6.1 in \cite{Chernozhukov2018} and Assumption \ref{ass:risk_decay}.\ref{ass:risk_decay_rate} it follows that $\sup_{t\in\mathcal{T}_i} \E\left[ \left|\muError\right| \left|\hat{e}_{S_{-i}}(X_t) - e_0(X_t)\right| \right] = o_p(T^{-1/2})$.

    We have shown that $P_1$, $P_2$ and $P_3$ are $o_p(T^{-1/2})$. It follows that $\frac{|\mathcal{T}_i|}{T} \hat{\theta}_{S_i} - \frac{|\mathcal{T}_i|}{T}  \tilde\theta_{S_i} = o_p(T^{-1/2})$ for all $i=1, ..., K$ and we can thus conclude that
    \begin{equation*}
        \sqrt{T} (\hat\theta - \theta) = \sum_{i=1}^{K}\underbrace{\sqrt{T} \left(\frac{|\mathcal{T}_i|}{T} \hat{\theta}_{S_i} - \frac{|\mathcal{T}_i|}{T} \tilde\theta_{S_i}\right)}_{=o_p(1)} + \underbrace{\sqrt{T}(\tilde\theta - \theta_0)}_{\overset{\mathrm{d}}{\to}\mathcal{N}(0, V_0)}.
    \end{equation*}
\end{proof}

\begin{lemma}\label{lemma:asym_independence}
Let $\Xi$ be a convex subset of some normed vector space, $g: \mathbb{R}^n \times \Xi \to \mathbb{R}$ be a measurable function, $\{Z_t : t \in \mathcal{T}\}$ an $\alpha$-mixing stochastic process with mixing coefficient $\alpha(m)$ and $Z_t$ a real-valued random vector. For $s\geq t$ denote by $\mathcal{F}_{t}^s = \sigma(Z_t,\dots, Z_s)$ the smallest $\sigma$-field such that $Z_t, \dots, Z_{s}$ are measurable. If $\norm{\sup_{\Gamma\in\Xi} g(Z_t, \Gamma)}_r < \infty$ for some $r\geq p \geq 1$, then for all $t\in\mathcal{T}$
$$\sup_{\Gamma\in\Xi}\E\left[ g(Z_t, \Gamma) | \mathcal{F}^{t-k}_{-\infty} \right] = \sup_{\Gamma\in\Xi}\E\left[ g(Z_t, \Gamma) \right] + O_p(\alpha(k)^{1/p-1/r}).$$
\end{lemma}

\begin{proof}[Proof of Lemma \ref{lemma:asym_independence}]
    We will prove the statement by bounding the $L_p$-norm. First notice that
    \begin{align*}
        \norm{\sup_{\Gamma\in\Xi}\E\left[ g(Z_t, \Gamma) | \mathcal{F}^{t-k}_{-\infty} \right] - \sup_{\Gamma\in\Xi}\E\left[ g(Z_t, \Gamma) \right]}_p \leq  \norm{\E\left[ \sup_{\Gamma\in\Xi} \left(g(Z_t, \Gamma) - \E\left[g(Z_t, \Gamma)\right]\right) | \mathcal{F}^{t-k}_{-\infty} \right] }_p
    \end{align*}
    By Theorem 15.2 in \cite{Davidson2021} we have
    \begin{equation*}
    \begin{aligned}
        &\norm{\E\left[ \sup_{\Gamma\in\Xi} \left(g(Z_t, \Gamma) - \E\left[g(Z_t, \Gamma)\right]\right) | \mathcal{F}^{t-k}_{-\infty} \right] }_p \\
        &\quad\leq 2(2^{1/2} + 1)\alpha(k)^{1/p-1/r} \norm{\sup_{\Gamma\in\Xi} \left(g(Z_t, \Gamma) - \E\left[g(Z_t, \Gamma)\right]\right)}_r\\
        &\quad\leq 2(2^{1/2} + 1)\alpha(k)^{1/p-1/r}\left( \norm{\sup_{\Gamma\in\Xi} g(Z_t, \Gamma) }_r + \left\vert\sup_{\Gamma\in\Xi} \E\left[g(Z_t, \Gamma)\right]\right\vert \right)=O(\alpha(k)^{1/p-1/r})
    \end{aligned}
    \end{equation*}
    where the second inequality follows from the Minkowski inequality.
\end{proof}

\begin{proof}[Proof of Theorem \ref{thm:ate_single}]
The proof builds on the proof of Theorem \ref{thm:ate} and we continue to omit the horizon $h$ for the sake of legibility. Following \cite{Davidson2021}, let the smallest $\sigma$-field on which the stochastic process $S_{-i} = \{Z_t : t \in \mathcal{T}_{-i} \}$ is measurable, be denoted as $\mathcal{F}_{-i} = \sigma\left(Z_t : t\in\mathcal{T} \wedge ( t < \inf(\mathcal{T}_i) - k_T \vee t > \sup(\mathcal{T}_i) + k_T) \right)$. Consider again the three summations $P_1$, $P_2$ and $P_3$ from the proof of Theorem \ref{thm:ate}.

Term $P_1$: From the squared $L_2$-norm of $P_1$ we to obtain
\begin{align}
    \begin{split}
    \E\Bigg[ &\Bigg(\frac{1}{|\mathcal{T}_i|} \sum_{t \in \mathcal{T}_i} \left( \muError \right) \left(1-\frac{D_t}{e_0(X_t)} \right) \Bigg)^2 \Bigg]\\
    &=\frac{1}{|\mathcal{T}_i|^2} \sum_{t \in \mathcal{T}_i}  \E \left[\left( \muError \right)^2 \left(1-\frac{D_t}{e_0(X_t)} \right)^2 \right]\\
    &\phantom{=}+\frac{1}{|\mathcal{T}_i|^2} \sum_{t \in \mathcal{T}_i} \sum_{s \in \mathcal{T}_i, s\neq t} \E \Bigg[\left( \muError \right) \left(1-\frac{D_t}{e_0(X_t)} \right) \times\\
    &\phantom{=+\frac{1}{|\mathcal{T}_i|^2} \sum_{t \in \mathcal{T}_i} \sum_{s \in \mathcal{T}_i, s\neq t} \E \Bigg[} \left( \hat{\mu}_{S_{-i}}(1,X_s) - \mu_0(1,X_s) \right) \left(1-\frac{D_s}{e_0(X_s)} \right) \Bigg].
    \label{eq:p1}
    \end{split}
\end{align}
First, note that by applying Hölder's inequality twice we have that for $r>p\geq1$ and $1/r = 1/r' + 1/r''$
\begin{align*}
    \sup_{t,s \in\mathcal{T}_i} \sup_{\mu\in\Xi_{T}} & \norm{ \left( \muErrorTilde \right) \left(1-\frac{D_t}{e_0(X_t)} \right) \left( \mu(1,X_s) - \mu_0(1,X_s) \right) \left(1-\frac{D_s}{e_0(X_s)} \right)}_r\\
    &\leq \sup_{t\in\mathcal{T}_i} \sup_{\mu\in\Xi_{T}} \norm{ \left( \muErrorTilde \right)^2 \left(1-\frac{D_t}{e_0(X_t)} \right)^2 }_r\\
    &\leq \sup_{t\in\mathcal{T}_i} \sup_{\mu\in\Xi_{T}} \norm{ \left( \muErrorTilde \right)^2 }_{r'}  \norm{\left(1-\frac{D_t}{e_0(X_t)} \right)^2 }_{r''}
\end{align*}
which is bounded by Assumption \ref{ass:sample_dependence}.\ref{ass:sample_dependence_bounded_y_error}. Next, for the first summand in \eqref{eq:p1}, note that conditional on $\mathcal{F}_{-i}$ the estimator is non-stochastic, and thus conditional on the event $\mathcal{E}_T$ we have that
\begin{align*}
    \sup_{t\in\mathcal{T}_i}  \E &\left[\left( \muError \right)^2 \left(1-\frac{D_t}{e_0(X_t)} \right)^2 \Bigg\vert \mathcal{F}_{-i} \right] \\
    &\leq \sup_{t\in\mathcal{T}_i} \sup_{\mu\in\Xi_{T}}  \E \left[\left( \mu(1,X_t) - \mu_0(1,X_t) \right)^2 \left(1-\frac{D_t}{e_0(X_t)} \right)^2 \Bigg\vert \mathcal{F}_{-i} \right]\\
    &\leq \sup_{t\in\mathcal{T}_i} \sup_{\mu\in\Xi_{T}}  \E \left[\left( \mu(1,X_t) - \mu_0(1,X_t) \right)^2 \left(1-\frac{D_t}{e_0(X_t)} \right)^2 \right] + O_p\left(\alpha(k_T)^\psi\right)\\
    &= o_p(1) + O_p\left(\alpha(k_T)^\psi\right)
\end{align*}
by Lemma \ref{lemma:asym_independence} in combination with Assumption \ref{ass:sample_dependence}, and the definition of the rate $r_{\mu,T}$ in Assumption \ref{ass:risk_decay}. By Lemma 6.1 in \cite{Chernozhukov2018} and Assumptions \ref{ass:risk_decay}.\ref{ass:risk_decay_rate} and \ref{ass:sample_dependence}.\ref{ass:sample_dependence_mixing} it follows that
\begin{align*}
    \frac{1}{|\mathcal{T}_i|^2} \sum_{t \in \mathcal{T}_i}  \E \left[\left( \muError \right)^2 \left(1-\frac{D_t}{e_0(X_t)} \right)^2 \right] = \frac{o_p(1)}{|\mathcal{T}_i|} = o_p(T^{-1}).
\end{align*}
Similarly, for the second summand of the $L_2$-norm of $P_1$, conditional on the event $\mathcal{E}_T$ we have that
\begin{align*}
    \sup_{t,s \in\mathcal{T}_i}  \E \Bigg[&\left( \muError \right) \left(1-\frac{D_t}{e_0(X_t)} \right) \times \\
    &\qquad\qquad \left( \hat{\mu}_{S_{-i}}(1,X_s) - \mu_0(1,X_s) \right) \left(1-\frac{D_s}{e_0(X_s)} \right) \Bigg\vert \mathcal{F}_{-i} \Bigg] \\
    &\leq \sup_{t,s \in\mathcal{T}_i} \sup_{e\in\Xi_{T}}  \E \bigg[\left( \mu(1,X_t) - \mu_0(1,X_t) \right) \left(1-\frac{D_t}{e_0(X_t)} \right) \times \\
    &\qquad\qquad\qquad\qquad \left( \mu(1,X_s) - \mu_0(1,X_s) \right) \left(1-\frac{D_s}{e_0(X_s)} \right) \Bigg\vert \mathcal{F}_{-i} \bigg]\\
    &\leq \sup_{t,s \in\mathcal{T}_i} \sup_{e\in\Xi_{T}}  \E \bigg[\left( \mu(1,X_t) - \mu_0(1,X_t) \right) \left(1-\frac{D_t}{e_0(X_t)} \right)  \times \\
    &\qquad\qquad\qquad\qquad\left( \mu(1,X_s) - \mu_0(1,X_s) \right) \left(1-\frac{D_s}{e_0(X_s)} \right) \bigg] + O_p\left(\alpha(k_T)^\psi\right)
\end{align*}
by Lemma \ref{lemma:asym_independence} in combination with Assumption \ref{ass:sample_dependence}, and the expectation in the last inequality is zero. By Lemma 6.1 in \cite{Chernozhukov2018} and Assumptions \ref{ass:risk_decay}.\ref{ass:risk_decay_rate} and \ref{ass:sample_dependence}.\ref{ass:sample_dependence_mixing} it follows that
\begin{align*}
    &\frac{1}{|\mathcal{T}_i|^2} \sum_{t \in \mathcal{T}_i} \sum_{s \in \mathcal{T}_i, s\neq t} \E \Bigg[\left( \muError \right) \left(1-\frac{D_t}{e_0(X_t)} \right) \times \\
    &\qquad\qquad\qquad\qquad\qquad \left( \hat{\mu}_{S_{-i}}(1,X_s) - \mu_0(1,X_s) \right) \left(1-\frac{D_s}{e_0(X_s)} \right) \Bigg] = O_p\left(\alpha(k_T)^\psi\right).
\end{align*}
By Assumption \ref{ass:sample_dependence}.\ref{ass:sample_dependence_mixing} it follows that $P_1$ is $o_p(T^{-1/2})$. Following a similar argument, it can be shown that $P_2$ is $o_p(T^{-1/2})$.

Convergence for $P_3$ can again be shown by bounding its $L_1$-norm as
    \begin{align*}
        \E\Bigg[\Bigg| & \frac{1}{|\mathcal{T}_i|} \sum_{t \in \mathcal{T}_i} D_t \left(\muError\right) \left(\invpError\right) \Bigg|\Bigg]\\
        &= \frac{1}{|\mathcal{T}_i|}\E\Bigg[\Bigg| \sum_{t \in \mathcal{T}_i} D_t \left(\hat\mu_{S_{-i}}(D_t,X_t) - \mu_0(D_t,X_t)\right) \left(\invpError\right) \Bigg|\Bigg]\\
        &\leq \frac{1}{\eta}\frac{1}{|\mathcal{T}_i|} \sum_{t \in \mathcal{T}_i} \E\Bigg[ \left|\hat\mu_{S_{-i}}(D_t,X_t) - \mu_0(D_t,X_t)\right| \left|\hat{e}_{S_{-i}}(X_t) - e_0(X_t)\right| \Bigg].
    \end{align*}
Noting that conditional on $\mathcal{F}_{-i}$, the nuisance function estimators are non-stochastic, and thus conditional on the event $\mathcal{E}_T$ we have that
    \begin{align*}
        \sup_{t\in\mathcal{T}_i} \E\Bigg[ &\left|\hat\mu_{S_{-i}}(D_t,X_t) - \mu_0(D_t,X_t)\right| \left|\hat{e}_{S_{-i}}(X_t) - e_0(X_t)\right| \Bigg\vert \mathcal{F}_{-i} \Bigg]\\
        &\leq \sup_{t\in\mathcal{T}_i} \sup_{\mu, e \in \Xi_T} \E\Bigg[ \left|\mu(D_t,X_t) - \mu_0(D_t,X_t)\right| \left|e(X_t) - e_0(X_t)\right| \Bigg\vert \mathcal{F}_{-i} \Bigg]\\
        &\leq \sup_{t\in\mathcal{T}_i} \sup_{\mu, e \in \Xi_T} \E\Bigg[ \left|\mu(D_t,X_t) - \mu_0(D_t,X_t)\right| \left|e(X_t) - e_0(X_t)\right| \Bigg] + O_p(\alpha(k_T)^\psi)\\
        & \leq \sup_{t\in\mathcal{T}_i} \sup_{\mu \in \Xi_{T}} \norm{\mu(D_t,X_t)-\mu_0(D_t,X_t)}_2 \sup_{t\in\mathcal{T}_i} \sup_{e \in \Xi_{T}} \norm{e(X_t) - e_0(X_t)}_2 + O_p(\alpha(k_T)^\psi)\\&
        = r_{\mu,T} \cdot r_{e,T} + O_p(\alpha(k_T)^\psi)
    \end{align*}
    by Lemma \ref{lemma:asym_independence} in combination with Assumption \ref{ass:sample_dependence}, Cauchy-Schwarz and the definition of the rates $r_{\mu,T}$ and $r_{e,T}$ in Assumption \ref{ass:risk_decay}. By Lemma 6.1 in \cite{Chernozhukov2018} and Assumptions \ref{ass:risk_decay}.\ref{ass:risk_decay_rate} and \ref{ass:sample_dependence}.\ref{ass:sample_dependence_mixing} it follows that $ \sup_{t\in\mathcal{T}_i} \E\left[ \left|\muError\right| \left|\hat{e}_{S_{-i}}(X_t) - e(X_t)\right| \right] = o_p(T^{-1/2})$. This concludes the proof as we can now apply the same arguments as in the proof of Theorem \ref{thm:ate}.
\end{proof}

\begin{proof}[Proof of Theorem \ref{thm:variance}]
Define $v_t(\theta, \Gamma) = g(Z_t; \Gamma) - \theta$, for some nuisance functions $\Gamma = (\mu, e)$, $g(Z_t; \Gamma)$ is the influence function \eqref{eq:dr_score} evaluated using the nuisance functions in $\Gamma$ and we will drop the forecast horizon $h$ everywhere for legibility. Let
\begin{equation*}
    V_T = \Var\left(\frac{1}{\sqrt{T}}\sum_{t\in\mathcal{T}}g(Z_t;\Gamma_0)\right) = \frac{1}{T}\left(\sum_{t=1}^T\E[v_t^2] + 2\sum_{s=1}^{T-1}\sum_{t=s+1}^{T}\E[v_tv_{t-s}]\right)
\end{equation*}
with $v_t = g(Z_t;\Gamma_0) - \theta_0$ and $\Gamma_0 = (\mu_0, e_0)$, and thus $\lim_{T\to\infty} V_T = V_0$. Next, define
\begin{equation*}
V_{S_i} = \frac{1}{|\mathcal{T}_i|}\left(\sum_{t\in\mathcal{T}_i}\E[v^2_t] + 2\sum_{s\in\mathcal{T}_i}\sum_{t\in\mathcal{T}_{i,s}}\E[v_tv_{t-s}]\right)
\end{equation*}
with $\mathcal{T}_{i,s} = \{t \in \mathcal{T}_i \, \vert \, t-s \geq \inf(\mathcal{T}_i) \}$ for $i=1, ..., K$. Then we have that $\lim_{T\to\infty}V_{S_i} = \lim_{T\to\infty} V_T = V_0$, so given $K$ being a finite integer, for $|\widehat{V} - V_T| \overset{p}{\to} 0$ it will be sufficient to show that $|\widehat{V}_{S_i} - V_{S_i}| = o_p(1)$, where
\begin{equation*}
    \widehat{V}_{S_i} = \frac{1}{|\mathcal{T}_i|}\left(\sum_{t\in\mathcal{T}_i}\hat{v}^2_t + 2\sum_{s=1}^{m_T}w(s,m_T)\sum_{t\in\mathcal{T}_{i,s}}\hat{v}
    _t\hat{v}_{t-s}\right)
\end{equation*}
with
\begin{equation*}
    \hat{v}_t = g(Z_t;\hat\Gamma_t) - \hat\theta \qquad \text{and} \qquad \hat\Gamma_t = \{\hat\Gamma_{S_{-i}} : i \in \{1, ..., K\}, t\in \mathcal{T}_i\}.
  \end{equation*}
Moreover, let
\begin{equation*}
    V^m_{S_i} = \frac{1}{|\mathcal{T}_i|}\left(\sum_{t\in\mathcal{T}_i}v^2_t + 2\sum_{s=1}^{m_T}w(s,m_T)\sum_{t\in\mathcal{T}_{i,s}}v_tv_{t-s}\right).
\end{equation*}

Applying the triangle inequality, we will follow \cite{Newey1987} and prove that the following three terms converge to zero in probability
\begin{equation}\label{eq:start}
        |\widehat{V}_{S_i} - V_{S_i}| \leq \underbrace{|\widehat{V}_{S_i} - V^m_{S_i}|}_{=P_1} + \underbrace{|V^m_{S_i} -\E[V^m_{S_i}]|}_{=P_2} + \underbrace{|\E[V^m_{S_i}] - V_{S_i}|}_{=P_3}.
\end{equation}
Since terms $P_2$ and $P_3$ do not contain any estimated quantities, they are $o_p(1)$ following the same arguments as in the proof of Theorem 2 in \cite{Newey1987} and \cite{Kool1988}, provided that Assumptions \ref{ass:variance_estimator}.\ref{ass:variance_estimator-i}-\ref{ass:variance_estimator}.\ref{ass:variance_estimator-iii} hold.

For the term $P_1$, define a function
\begin{align*}
    f(\theta, r) = &\frac{1}{|\mathcal{T}_i|}\sum_{t\in\mathcal{T}_i}v_t(\theta, \Gamma_0 + r(\hat\Gamma_t - \Gamma_0))^2 \\
    &+ \frac{2}{|\mathcal{T}_i|}\sum_{s=1}^{m_T}w(s,m_T)\sum_{t\in\mathcal{T}_{i,s}}v_{t}(\theta, \Gamma_0 + r(\hat\Gamma_t - \Gamma_0)) v_{t-s}(\theta, \Gamma_0 + r(\hat\Gamma_{t-s} - \Gamma_0)),
\end{align*}
so that $\widehat{V}_{S_i} = f(\hat\theta, 1)$ and $V^m_{S_i} = f(\theta_0, 0)$. By the multivariate mean-value theorem, for some $(\tilde\theta, \tilde{r})$ on the line segment from $(\theta_0, 0)$ to $(\hat\theta, 1)$, on the event $\mathcal{E}_T$ we have
\begin{equation*}
    \widehat{V}_{S_i} - V^m_{S_i} = f(\hat\theta, 1) - f(\theta_0, 0) = \frac{\partial f}{\partial\theta}(\tilde\theta, \tilde{r}) (\hat\theta-\theta_0) + \frac{\partial f}{\partial r}(\tilde\theta, \tilde{r})
\end{equation*}
and by the triangle inequality
\begin{equation}\label{eq:Pterms}
    |\widehat{V}_{S_i} - V^m_{S_i}| \leq \underbrace{\left\vert\frac{\partial f}{\partial\theta}(\tilde\theta, \tilde{r}) (\hat\theta-\theta_0)\right\vert}_{=P_{11}} + \underbrace{\left\vert\frac{\partial f}{\partial r}(\tilde\theta, \tilde{r})\right\vert}_{=P_{12}}.
\end{equation}
We will show that both terms on the right hand side of \eqref{eq:Pterms} are $o_p(1)$.\\

Term $P_{11}$: The partial derivative with respect to $\theta$ in $P_{11}$ is\
\begin{equation*}
\begin{aligned}
    \frac{\partial f}{\partial\theta}(\tilde\theta, \tilde{r}) &= -\frac{2}{|\mathcal{T}_i|}\sum_{t\in\mathcal{T}_i}v_t(\tilde\theta, \Gamma_0 + \tilde{r}(\hat\Gamma_t - \Gamma_0))\\
    &\qquad -\frac{2}{|\mathcal{T}_i|}\sum_{s=1}^{m_T}w(s,m_T)\sum_{t\in\mathcal{T}_{i,s}}\left[v_t(\tilde\theta, \Gamma_0 + \tilde{r}(\hat\Gamma_t - \Gamma_0)) + v_{t-s}(\tilde\theta, \Gamma_0 + \tilde{r}(\hat\Gamma_{t-s} - \Gamma_0))\right]
\end{aligned}
\end{equation*}
and thus
\begin{equation*}
    \left\vert\frac{\partial f}{\partial\theta}(\tilde\theta, \tilde{r})\right\vert \leq 2\sup_{t\in\mathcal{T}_i}|v_t(\tilde\theta, \Gamma_0 + \tilde{r}(\hat\Gamma_t - \Gamma_0))| + 4m_T\sup_{t\in\mathcal{T}_i}|v_t(\tilde\theta, \Gamma_0 + \tilde{r}(\hat\Gamma_t - \Gamma_0))|.
\end{equation*}
Then we get for $P_{11}$ that
\begin{equation}\label{eq:P11}
\begin{aligned}
    \left\vert\frac{\partial f}{\partial\theta}(\tilde\theta, \tilde{r})\right\vert|\hat\theta-\theta_0| \leq & \ 2\sup_{t\in\mathcal{T}_i}|v_t(\tilde\theta, \Gamma_0 + \tilde{r}(\hat\Gamma_t - \Gamma_0))|\cdot|\hat\theta-\theta_0| \\
    &+ 4\frac{m_T}{\sqrt{T}}\sup_{t\in\mathcal{T}_i}|v_t(\tilde\theta, \Gamma_0 + \tilde{r}(\hat\Gamma_t - \Gamma_0))|\cdot\sqrt{T}|\hat\theta-\theta_0|.
\end{aligned}
\end{equation}
By Theorem \ref{thm:ate_single}, on the event $\mathcal{E}_T$, $|\hat\theta-\theta_0| = o_p(1)$ and $\sqrt{T}(\hat\theta-\theta_0) = O_p(1)$. Moreover, by Assumption \ref{ass:variance_estimator}.\ref{ass:variance_estimator-iii}, $\lim_{T\to\infty} m_T / \sqrt{T} = 0$. Thus it remains to show that on the event $\mathcal{E}_T$, $\sup_{t\in\mathcal{T}_i}|v_t(\tilde\theta, \Gamma_0 + \tilde{r}(\hat\Gamma_t - \Gamma_0))|$ remains bounded in probability. By the triangular inequality we have that
\begin{equation*}
    \bigg\vert v_t(\tilde\theta, \Gamma_0 + \tilde{r}(\hat\Gamma_t - \Gamma_0)) \bigg\vert = \bigg\vert g(Z_t; \Gamma_0 + \tilde{r}(\hat\Gamma_t - \Gamma_0)) - \tilde\theta \bigg\vert \leq \bigg\vert g(Z_t; \Gamma_0 + \tilde{r}(\hat\Gamma_t - \Gamma_0))\bigg\vert + \vert\tilde\theta \vert.
\end{equation*}
For the first term after the inequality we have, for some $q>2$
\begin{align}\label{eq:h_1}
        &\norm{ g(Z_t; \Gamma_0 + \tilde{r}(\hat\Gamma_t - \Gamma_0))}_{q} \nonumber \\
        &= \bigg\lVert \mu_0(1, X_t) + \tilde{r}(\hat\mu_t(1, X_t) - \mu_0(1,X_t)) \nonumber\\
        &\qquad - \mu_0(0,X_t) - \tilde{r}(\hat\mu_t(0,X_t) - \mu_0(0,X_t)) \\
        &\qquad + \frac{D_t}{e_0(X_t)+\tilde{r}(\hat{e}_t(X_t) - e_0(X_t))} (Y_t-\mu_0(1,X_t)-\tilde{r}(\hat\mu_t(1,X_t) - \mu_0(1,X_t))) \nonumber\\
        &\qquad - \frac{1-D_t}{1-e_0(X_t)-\tilde{r}(\hat{e}_t(X_t) - e_0(X_t))}(Y_t-\mu_0(0,X_t)-\tilde{r}(\hat\mu_t(0,X_t) - \mu_0(0,X_t))) \bigg\rVert_{q} \nonumber\\
        &\leq \norm{\mu_0(1, X_t)}_{q} + \tilde{r}\norm{\hat\mu_t(1,X_t) - \mu_0(1, X_t)}_{q} + \norm{\mu_0(0, X_t)}_{q} + \tilde{r}\norm{\hat\mu_t(0,X_t) - \mu_0(0, X_t)}_{q}\nonumber \\
        &\qquad + \bigg\lVert \frac{D_t}{e_0(X_t)+\tilde{r}(\hat{e}_t(X_t) - e_0(X_t))} (Y_t-\mu_0(1,X_t)-\tilde{r}(\hat\mu_t(1,X_t) - \mu_0(1,X_t))) \bigg\rVert_{q}\nonumber\\
        &\qquad + \bigg\lVert \frac{1-D_t}{1-e_0(X_t)-\tilde{r}(\hat{e}_t(X_t) - e_0(X_t))}(Y_t-\mu_0(0,X_t)-\tilde{r}(\hat\mu_t(0,X_t) - \mu_0(0,X_t))) \bigg\rVert_{q}\nonumber
\end{align}
where
\begin{equation*}
    \hat\mu_t = \{\hat\mu_{S_{-i}} : i \in \{1, ..., K\}, t\in \mathcal{T}_i\} \qquad \text{and} \qquad \hat{e}_t = \{\hat{e}_{S_{-i}} : i \in \{1, ..., K\}, t\in \mathcal{T}_i\}.
\end{equation*}

From Assumptions \ref{ass:risk_decay} and \ref{ass:variance_estimator}.\ref{ass:variance_estimator-iv} we have
\begin{equation}\label{eq:cherno_results}
    \norm{\hat\mu_t(d,X_t) - \mu_0(d,X_t)}_q  \leq \frac{C}{\eta^{1/q}} \qquad\text{and}\qquad \norm{\mu_0(d,X_t)}_q \leq \frac{C}{\eta^{1/q}},
\end{equation}
for $q>2$, any $\hat\Gamma_t \in \Xi_T$ and $d\in\{0,1\}$, as shown in the proof of Theorem 5.1 in \cite{Chernozhukov2018}. Using these results, we get the following bound for $|\tilde\theta|$
\begin{equation*}
    |\tilde\theta| \leq |\tilde\theta - \theta_0| + |\theta_0| \leq |\hat\theta - \theta_0| + |\theta_0| \leq |\hat\theta - \theta_0| + 2\frac{C}{\eta^{1/q}},
\end{equation*}
where $|\theta_0| \leq 2C / \eta^{1/q}$ follows from Assumption \ref{ass:risk_decay}, as shown in the proof of Theorem 5.1 in \cite{Chernozhukov2018}. From Theorem \ref{thm:ate_single}, it follows that $|\hat\theta - \theta_0| = o_p(1)$, so $|\tilde\theta| = O_p(1)$.

We continue with the term in the second line of the last inequality in \eqref{eq:h_1}. For $q>2$, conditional on $\mathcal{E}_T$, we get
\begin{align}\label{eq:Dterm}
  \bigg\lVert &\frac{D_t}{e_0(X_t)+\tilde{r}(\hat{e}_t(X_t) - e_0(X_t))} (Y_t-\mu_0(1,X_t)-\tilde{r}(\hat\mu_t(1,X_t) - \mu_0(1,X_t))) \bigg\rVert_{q} \nonumber\\
  &\leq \frac{1}{\eta}\bigg\lVert D_t(Y_t-\mu_0(1,X_t)-\tilde{r}(\hat\mu_t(1,X_t) - \mu_0(1,X_t))) \bigg\rVert_{q}\nonumber\\
  &\leq \frac{1}{\eta}\bigg\lVert Y_t-\mu_0(1,X_t)-\tilde{r}(\hat\mu_t(1,X_t) - \mu_0(1,X_t)) \bigg\rVert_{q} \\
  &\leq \frac{1}{\eta}\norm{Y_t}_q+ \frac{1}{\eta}\norm{\mu_0(1,X_t)}_q + \frac{\tilde{r}}{\eta}\norm{\hat\mu_t(1,X_t) - \mu_0(1,X_t)}_{q} \leq \frac{1}{\eta}C + \frac{1}{\eta}\frac{C}{\eta^{1/q}} + \frac{\tilde{r}}{\eta}\frac{C}{\eta^{1/q}},\nonumber
\end{align}
where the first inequality follows from Assumption \ref{ass:risk_decay}. The second inequality follows from $D_t\in\{0,1\}$, and the third from the Minkowski inequality. The fourth follows from \eqref{eq:cherno_results}.

For the term in the third line of the last inequality in \eqref{eq:h_1}, we get for $q>2$ and conditional on $\mathcal{E}_T$ that
\begin{equation}\label{eq:1Dterm}
  \begin{aligned}
  \bigg\lVert &\frac{1-D_t}{1-e_0(X_t)-\tilde{r}(\hat{e}_t(X_t) - e_0(X_t))}(Y_t-\mu_0(0,X_t)-\tilde{r}(\hat\mu_t(0,X_t) - \mu_0(0,X_t))) \bigg\rVert_{q} \\
  &\leq \frac{1}{\eta}\bigg\lVert Y_t-\mu_0(0,X_t)-\tilde{r}(\hat\mu_t(0,X_t) - \mu_0(0,X_t)) \bigg\rVert_{q} \leq \frac{1}{\eta}C + \frac{1}{\eta}\frac{C}{\eta^{1/q}} + \frac{\tilde{r}}{\eta}\frac{C}{\eta^{1/q}},\\
  \end{aligned}
\end{equation}
using the same arguments as for the term in \eqref{eq:Dterm}. Since all of the terms in the first line of the last inequality in \eqref{eq:h_1} are $L_q$-bounded, they are also $O_p(1)$ by Markov's inequality. The same holds for \eqref{eq:Dterm} and \eqref{eq:1Dterm}, so we can conclude that all of the terms on the right-hand-side of the last inequality in \eqref{eq:h_1} are $O_p(1)$. As a consequence, $\sup_{t\in\mathcal{T}_i}|v_t(\tilde\theta, \Gamma_0 + \tilde{r}(\hat\Gamma_t - \Gamma_0))| = O_p(1)$ in \eqref{eq:P11} and thus $P_{11} = o_p(1)$.

Term $P_{12}$: From the partial derivative with respect to $r$ in $P_{12}$ we find
\begin{equation*}
\begin{aligned}
    \left|\frac{\partial f}{\partial r}(\tilde\theta, \tilde{r})\right| &= \bigg\vert \frac{2}{|\mathcal{T}_i|}\sum_{t\in\mathcal{T}_i}v_t(\tilde\theta, \Gamma_0 + \tilde{r}(\hat\Gamma_t - \Gamma_0)) \frac{\partial v_t}{\partial r}(\tilde\theta, \Gamma_0 + \tilde{r}(\hat\Gamma_t - \Gamma_0))\\
    &\qquad + \frac{2}{|\mathcal{T}_i|} \sum_{s=1}^{m_T}w(s,m_T) \sum_{t\in\mathcal{T}_{i,s}} v_t(\tilde\theta, \Gamma_0 + \tilde{r}(\hat\Gamma_t - \Gamma_0)) \frac{\partial v_{t-s}}{\partial r}(\tilde\theta, \Gamma_0 + \tilde{r}(\hat\Gamma_{t-s} - \Gamma_0))\\
    &\qquad + \frac{2}{|\mathcal{T}_i|} \sum_{s=1}^{m_T}w(s,m_T) \sum_{t\in\mathcal{T}_{i,s}} v_{t-s}(\tilde\theta, \Gamma_0 + \tilde{r}(\hat\Gamma_{t-s} - \Gamma_0)) \frac{\partial v_{t}}{\partial r}(\tilde\theta, \Gamma_0 + \tilde{r}(\hat\Gamma_t - \Gamma_0)) \bigg\vert \\
    & \leq 2\sup_{t\in\mathcal{T}_i}|v_t(\tilde\theta, \Gamma_0 + \tilde{r}(\hat\Gamma_t - \Gamma_0))| \sup_{t\in\mathcal{T}_i}\left\vert \frac{\partial v_t}{\partial r}(\tilde\theta, \Gamma_0 + \tilde{r}(\hat\Gamma_t - \Gamma_0)) \right\vert \\
    &\qquad + 4\frac{m_T}{T^{b_m}}\sup_{t\in\mathcal{T}_i}|v_t(\tilde\theta, \Gamma_0 + \tilde{r}(\hat\Gamma_t - \Gamma_0))| \sup_{t\in\mathcal{T}_i}\left\vert \frac{\partial v_t}{\partial r}(\tilde\theta, \Gamma_0 + \tilde{r}(\hat\Gamma_t - \Gamma_0))\right\vert T^{b_m}.
\end{aligned}
\end{equation*}

From proving $P_{11} = o_p(1)$ we know that $\sup_{t\in\mathcal{T}_i}|v_t(\tilde\theta, \Gamma_0 + \tilde{r}(\hat\Gamma_t - \Gamma_0))| = O_p(1)$. Moreover, $m_T / T^{b_m} = o(1)$ by Assumption \ref{ass:variance_estimator}.\ref{ass:variance_estimator-iii}. Thus it remains to show that $\sup_{t\in\mathcal{T}_i}|\frac{\partial v_t}{\partial r}(\tilde\theta, \Gamma_0 + \tilde{r}(\hat\Gamma_t - \Gamma_0))| = o_p(T^{-b_m})$. Since by Assumption \ref{ass:variance_estimator} $b_m < b_r$, we will show that the quantity is actually $o_p(T^{-b_r})$. By the triangle inequality, we have
\begin{equation}\label{eq:P12,i}
    \sup_{t\in\mathcal{T}_i}\left\vert\frac{\partial v_t}{\partial r}(\tilde\theta, \Gamma_0 + \tilde{r}(\hat\Gamma_t - \Gamma_0))\right\vert \leq P_{12,A} + P_{12,B} + P_{12,C} + P_{12,D},
\end{equation}
where
\begin{equation*}
\begin{aligned}
    P_{12,A} &= \sup_{t\in\mathcal{T}_i}\left\vert \hat\mu_t(1,X_t)  - \mu_0(1,X_t) - \left(\hat\mu_t(0,X_t)  - \mu_0(0,X_t)\right)\right\vert \\
    P_{12,B} &= \sup_{t\in\mathcal{T}_i}\bigg\vert \frac{D_t}{\left(e_0(X_t) + \tilde{r}(\hat{e}_t(X_t) - e_o(X_t))\right)^2} \times\\
    &\qquad\quad \left(Y_t - \mu_0(1,X_t) - \tilde{r}(\hat\mu_t(1,X_t) - \mu_0(1,X_t))\right) \left(\hat{e}_t(X_t) - e_0(X_t)\right)\bigg\vert \\
    P_{12,C} &= \sup_{t\in\mathcal{T}_i}\left\vert \frac{1-D_t}{1-e_0(X_t) - \tilde{r}(\hat{e}_t(X_t)-e_0(X_t))} \left(\hat\mu_t(0,X_t) - \mu_0(0,X_t)\right)\right\vert \\
    P_{12,D} &= \sup_{t\in\mathcal{T}_i}\bigg\vert \frac{1-D_t}{\left( (1-\tilde{r})(1-e_0(X_t)) + \tilde{r}(1-\hat{e}_t(X_t))\right)^2} \times\\
    &\qquad\quad \left(Y_t - \mu_0(0,X_t) + \tilde{r}(\hat\mu_t(0,X_t) - \mu_0(0,X_t))\right)\left(\hat{e}_t(X_t) - e_0(X_t)\right)\bigg\vert.
\end{aligned}
\end{equation*}

For $P_{12,A}$, we first establish that on the event $\mathcal{E}_T$ for $d\in\{0,1\}$ we have
\begin{equation}\label{eq:P12A}
\begin{aligned}
\sup_{t\in\mathcal{T}_i}&\E\left[ \left\vert\hat\mu_t(d,X_t) - \mu_0(d,X_t)\right\vert^2 \bigg\vert S_{-i}\right] \\
    &\leq \sup_{t\in\mathcal{T}_i}\sup_{\mu\in \Xi_T} \E\left[ \left(\mu(d,X_t) - \mu_0(d,X_t)\right)^2 \bigg\vert S_{-i}\right]\\
    &\leq \sup_{t\in\mathcal{T}_i}\sup_{\mu\in \Xi_T} \E\left[ \left(\mu(d,X_t) - \mu_0(d,X_t)\right)^2\right] + O_p(\alpha(k_T)^\psi)\\
    &\leq r_{\mu,T}^2 + O_p(\alpha(k_T)^\psi) = o_p(T^{-2 b_r}).
\end{aligned}
\end{equation}
The second inequality follows from Lemma \ref{lemma:asym_independence}, the third by the definition of the rate $r_{\mu,T}$ in Assumption \ref{ass:risk_decay}. The final equality follows by Lemma 6.1 in \cite{Chernozhukov2018} and Assumptions \ref{ass:risk_decay}.\ref{ass:risk_decay_rate},  \ref{ass:sample_dependence}.\ref{ass:sample_dependence_mixing} and \ref{ass:variance_estimator}.\ref{ass:variance_estimator-v}. We thus conclude that on the event $\mathcal{E}_T$,
\begin{equation}\label{eq:mu-convergence}
    \sup_{t\in\mathcal{T}_i} \left\vert\hat\mu_t(d,X_t) - \mu_0(d,X_t)\right\vert = o_p(T^{-b_r})
\end{equation}
and as a consequence
\begin{equation*}
    P_{12,A} \leq \sup_{t\in\mathcal{T}_i} \left\vert\hat\mu_t(1,X_t) - \mu_0(1,X_t)\right\vert + \sup_{t\in\mathcal{T}_i} \left\vert\hat\mu_t(0,X_t) - \mu_0(0,X_t)\right\vert = o_p(T^{-b_r}).
\end{equation*}

For $P_{12,B}$, we first establish that on the event $\mathcal{E}_T$ we have
\begin{equation*}
    \begin{aligned}
        \sup_{t\in\mathcal{T}_i}\E\left[ \left\vert\hat{e}_t(X_t) - e_0(X_t)\right\vert^2 \big\vert S_{-i}\right] &\leq \sup_{t\in\mathcal{T}_i}\sup_{e \in \Xi_T} \E\left[ (\hat{e}_t(X_t) - e_0(X_t))^2\big\vert S_{-i}\right]\\
        &\leq \sup_{t\in\mathcal{T}_i}\sup_{e \in \Xi_T} \E\left[ (\hat{e}_t(X_t) - e_0(X_t))^2\right] + O_p(\alpha(k_T)^\psi)\\
        &\leq r_{e,T}^2 + O_p(\alpha(k_T)^\psi) = o_p(T^{-2 b_r})\\
    \end{aligned}
\end{equation*}
by the same arguments as for \eqref{eq:P12A} and thus on the event $\mathcal{E}_T$,
\begin{equation}\label{eq:e-convergence}
    \sup_{t\in\mathcal{T}_i} \left\vert\hat{e}_t(X_t) - e_0(X_t)\right\vert = o_p(T^{-b_r}).
\end{equation}

This allows us to write
\begin{equation}\label{eq:P12B}
\begin{aligned}
    P_{12,B} &\leq \sup_{t\in\mathcal{T}_i}\frac{1}{\eta^2}\left\vert Y_t - (1-\tilde{r})\mu_0(1,X_t)-\tilde{r}\hat\mu_t(1,X_t)\right\vert \left\vert\hat{e}_t(X_t) - e_0(X_t)\right\vert \\
    &\leq \sup_{t\in\mathcal{T}_i}\frac{1-\tilde{r}}{\eta^2}\left\vert Y_t - \mu_0(1,X_t)\right\vert \left\vert\hat{e}_t(X_t) - e_0(X_t)\right\vert \\
    &\quad + \sup_{t\in\mathcal{T}_i}\frac{\tilde{r}}{\eta^2}\left\vert Y_t - \hat\mu_t(1,X_t)\right\vert \left\vert\hat{e}_t(X_t) - e_0(X_t)\right\vert \\
    &\leq \sup_{t\in\mathcal{T}_i}\frac{2}{\eta^2}\left\vert Y_t - \mu_0(1,X_t)\right\vert \left\vert\hat{e}_t(X_t) - e_0(X_t)\right\vert \\
    &\quad + \sup_{t\in\mathcal{T}_i}\frac{1}{\eta^2}\left\vert \hat\mu_t(1,X_t) - \mu_0(1,X_t)\right\vert \left\vert\hat{e}_t(X_t) - e_0(X_t)\right\vert, \\
\end{aligned}
\end{equation}
where the first inequality follows from Assumption \ref{ass:risk_decay} and $D_t\in\{0,1\}$. From Assumption \ref{ass:variance_estimator}.\ref{ass:variance_estimator-iv} and \eqref{eq:cherno_results}, we have for $q>2$, $d\in\{0,1\}$ and any $\mu \in \Xi_T$ that
\begin{equation}\label{eq:Y-convergence}
    \sup_{t\in\mathcal{T}_i}\norm{Y_t - \mu(d,X_t)}_q \leq \sup_{t\in\mathcal{T}_i} \norm{Y_t}_q + \sup_{t\in\mathcal{T}_i}\norm{\mu(d,X_t)}_q \leq C\left(1+\frac{1}{\eta^{1/q}}\right).
\end{equation}
Combined with \eqref{eq:e-convergence} it follows that for $d\in\{0,1\}$ and on the event $\mathcal{E}_T$ we have
\begin{align*}
    &\sup_{t\in\mathcal{T}_i}\norm{|Y_t - \mu(d,X_t)||\hat{e}_t(X_t) - e_0(X_t)|}_1 \\
    &\quad \leq \sup_{t\in\mathcal{T}_i}\norm{Y_t - \mu(d,X_t)}_2 \norm{\hat{e}_t(X_t) - e_0(X_t)}_2 \\
    &\quad \leq \sup_{t\in\mathcal{T}_i}C\left(1+\frac{1}{\eta^{1/q}}\right) \norm{\hat{e}_t(X_t) - e_0(X_t)}_2 = o_p(T^{-b_r})
\end{align*}
for any $\mu \in \Xi_T$, and as a consequence, combined with \eqref{eq:mu-convergence}, we have $P_{12,B} = o_p(T^{-b_r})$.

For $P_{12,C}$, we have by Assumption \ref{ass:risk_decay} and $D_t\in\{0,1\}$ that $P_{12,C} \leq \sup_{t\in\mathcal{T}_i}\frac{1}{\eta}| \hat\mu_t(0,X_t) - \mu_0(0,X_t) |$
and thus $P_{12,C} = o_p(T^{-b_r})$ by \eqref{eq:mu-convergence}.

For $P_{12,D}$ finally, we have
\begin{equation*}
\begin{aligned}
    P_{12,D} &\leq \sup_{t\in\mathcal{T}_i}\frac{1}{\eta^2}\left\vert (1-\tilde{r})(Y_t-\mu_0(0,X_t)) + \tilde{r}(Y_t - \hat\mu_t(0,X_t)) \right\vert \left\vert\hat{e}_t(X_t) - e_0(X_t)\right\vert \\
    &\leq \sup_{t\in\mathcal{T}_i}\frac{2}{\eta^2} \left\vert Y_t-\mu_0(0,X_t)\right\vert  \left\vert\hat{e}_t(X_t) - e_0(X_t)\right\vert \\
    &\quad + \sup_{t\in\mathcal{T}_i}\frac{1}{\eta^2} \left\vert \hat\mu_t(0,X_t)-\mu_0(0,X_t)\right\vert \left\vert\hat{e}_t(X_t) - e_0(X_t)\right\vert,
\end{aligned}
\end{equation*}
where the first inequality again follows from Assumption \ref{ass:risk_decay} and $D_t\in\{0,1\}$, and the second from the same arguments as in \eqref{eq:P12B}. Using \eqref{eq:mu-convergence}, \eqref{eq:e-convergence} and \eqref{eq:Y-convergence} lets us conclude that $P_{12,D} = o_p(T^{-b_r})$. This shows that all terms on the right hand side of \eqref{eq:P12,i} are $o_p(T^{-b_r})$ and thus $\sup_{t\in\mathcal{T}_i}|\frac{\partial v_t}{\partial r}(\tilde\theta, \Gamma_0 + \tilde{r}(\hat\Gamma_t - \Gamma_0))| = o_p(T^{-b_r})$. As a consequence, $P_{11} = o_p(1)$ and $P_{12} = o_p(1)$ in \eqref{eq:Pterms}. We conclude that $|\widehat{V}_{S_i} - V^m_{S_i}| = o_p(1)$ and thus $|\widehat{V}_{S_i} - V_{S_i}| = o_p(1)$ so that $|\widehat{V} - V_T| \overset{p}{\to} 0$.
\end{proof}

\end{appendix}

\clearpage
\bibliographystyle{apacite}
\bibliography{double_ml_ts}

\renewcommand{\thetable}{\thesection\arabic{table}}
\renewcommand{\thefigure}{\thesection\arabic{figure}}
\renewcommand{\theequation}{\thesection\arabic{equation}}
\clearpage
\begin{appendix}
    \begin{center}
        \Large
        Supplementary material for ``Semiparametric inference for impulse response functions using double/debiased machine learning''
    \end{center}
    \section{Hyperparameter tuning}
    \subsection{Tuning in the simulation study}\label{app:tuning}
    Random forests: We fix the number of trees to 500. For each simulation replication, the maximal depth of each tree ($d$), the minimal number of observations in the leafs of the trees ($\ell$) and the maximal fraction of features considered for each node split ($m_{try}$) are determined by cross-validation using the $K$ sub-processes as folds. We perform a simple grid search over $\{\{d,\ell, m_{try}\}~\lvert~d\in\{5, 10,20,50\} \wedge \ell \in \{1,5,10\}  \wedge m_{try} \in \{0.3, 1.0\}\}$ and select the hyperparameter-combination yielding the best predictive cross-validation performance.
    
    Gradient boosted trees: We perform a two-stage tuning procedure. In the first stage, we fix the learning rate of the tree booster to 0.1, set the number of boosting rounds to some very high number (10'000) and abort the estimation process if the predictive validation error has not decreased since 50 rounds of boosting. The maximum tree depth for the base learners ($d$), the minimum sum of instance weight needed in a child ($w$), the subsampling ratio for observations used to construct each tree ($s^o$), and the subsampling ratio for features when constructing each tree ($s^f$) are determined by cross-validation using the $K$ sub-processes as folds. We perform a simple grid search over $\{\{d,w,s^o, s^f\}~\lvert~d\in\{1, ...,10\} \wedge w \in \{1, ..., 10\}  \wedge s^o\in\{0.25,0.5,0.75, 1\} \wedge s^f\in\{0.25,0.5,0.75, 1\}\}$ and select the hyperparameter-combination yielding the best predictive cross-validation performance. In the second stage, using the optimal tree-hyperparameters from stage 1, we select the learning rate most frequently yielding the best predictive cross-validation performance. We perform a search over the candidate set $\{0.001, 0.002, 0.005, 0.01, 0.02, 0.05, 0.1, 0.25, 0.5\}$. Finally, the optimal number of boosting rounds is determined as the average early-stopped boosting round for the selected learning rate. For computational reasons, we repeat this procedure on $R=50$ simulated examples and fix the hyperparameters for all simulation replications to the hyperparameter combination most frequently yielding the best cross-validation performance.\\\\
    
    \subsection{Tuning in the empirical application}\label{app:tuning_empirical}
    Random forests: We fix the number of trees to 500. The maximal depth of each tree ($d$), the minimal number of observations in the leafs of the trees ($\ell$) and the size of the random subsets of features to consider when splitting a node ($\bar{m}$) are determined by cross-validation using the $K$ sub-processes as folds. We perform a simple grid search over the same candidate sets as for the simulation study (cf. Appendix \ref{app:tuning}) and select the hyperparameter-combination yielding the best predictive cross-validation performance.
    
    Gradient boosted trees: We perform a two-stage tuning procedure. In the first stage, we fix the learning rate of the tree booster to 0.1, the number of boosting rounds to 500 and find optimal tree parameters by cross-validation. For this, we perform a simple grid search over $\{\{d,w,s^o,s^f\}~\lvert~d\in\{1, ...,10\} \wedge w \in \{1, ..., 10\} \wedge s^o\in\{0.25,0.5,0.75, 1\} \wedge s^f\in\{0.25,0.5,0.75, 1\}\}$, for the same parameters as in Appendix \ref{app:tuning}, and select the hyperparameter-combination yielding the best predictive cross-validation performance. In the second stage, using the optimal tree-hyperparameters from stage 1, we finally select the learning rate and number of boosting rounds yielding the best predictive cross-validation performance. We perform a search over all combinations of $\{0.001, 0.002, 0.005, 0.01, 0.02, 0.05, 0.1, 0.25, 0.5\}$ for the learning rate and $\{10, 110, 210, 310, 410, 510\}$ for the number of boosting rounds.
    
    All estimators are tuned using 10-fold blocked cross-validation on the full sample, where at the boundary of the folds 24 observations are dropped to eliminate dependence between the folds.
    
    \section{Simulation with an empirically calibrated data generating process}\label{app:empirical_dgp}
    We calibrate the data generating process defined in Section \ref{sec:simulation} of the main text using monthly U.S. data from 1982 to 2012 obtained from the empirical study in \cite{Angrist2018}.
    In more detail, the outcome process follows
    \begin{equation*}
            Y_t = c + b(X_t) + \left(D_t -0.5\right)\tau(X_t) + \gamma Y_{t-1} + \epsilon_t.
    \end{equation*}
    We calibrate the parameters $c$ and $\gamma$, the innovation process $\epsilon_t$ and the process governing a set of confounder variables $X_t$. The vector $X_t$ contains the changes in the federal funds rate and in the ten-year Treasury yield, percentage changes in the S\&P 1000 index, the M1 money stock, civilian employment, and industrial production, as well as percentage-point changes in the unemployment rate. The outcome variable $Y_t$ is the monthly percentage changes in core personal consumption expenditures.
    The specification of $b(X_t)$, $\tau(X_t)$ and $e(X_t)$ are kept identical to those used in the simulation study in Section \ref{sec:simulation} of the main text.
    
    The innovation process is assumed to follow a GARCH($p,q$) specification. Following \cite{Lazarus2018,Adamek2024}, the confounder process $X_t$ is modelled using a dynamic factor model as
    \begin{equation*}
            \begin{aligned}
                    X_t &= \Lambda F_t + U_t\\
                    F_t &= \sum_{j=1}^{p_F} \Phi_j F_{t-j} + V_t, \qquad V_t \sim \mathcal{N}(0,I)\\
                    U_{i,t} &= \sum_{j=1}^{p_U} \phi_j U_{i,t-j} + \eta_t, \qquad \eta_{i,t} \sim \mathcal{N}(0,\sigma^2_{i,\eta}) \quad \text{for } i=1,\dots, \dim(X_t).
            \end{aligned}
    \end{equation*}
    Model orders are selected using the Bayesian Information Criterion (BIC). Specifically, we determine the GARCH($p,q$) orders for $\epsilon_t$, the number of latent factors $F_t$, the lag length $p_F$ of the factor VAR, and the AR order $p_U$ of the idiosyncratic components by minimizing the BIC. For the sample at hand, this procedure selects an ARCH(1) process for $\epsilon_t$, one latent factor following a VAR(1) process, and univariate AR(1) processes for the components of $U_t$.
    Simulation samples are generated by drawing independently from the distributions of $\epsilon_t$, $V_t$ and $\eta_t$.
    
    \clearpage
    \section{Additional tables and figures}\label{app:additional_tabs}
    
    \begin{table}[htbp]
    \caption{Simulation results for $K$ independent nonlinear baseline DGPs ($n=12$ and $\sigma_\epsilon=1.0$) and random forest nuisance function estimates}
    \label{tab:result_baseline_two_proc}
    \resizebox{\textwidth}{!}{\begin{threeparttable}\begin{tabular}{rrrrrrrrrrrrrrrrrr}
            \toprule
            \multicolumn{18}{c}{$h=0$, $\theta_0^{(h)}=0.3321$} \\
             & \multicolumn{5}{c}{DML} & \multicolumn{4}{c}{RA} & \multicolumn{4}{c}{DR} & \multicolumn{4}{c}{LP} \\
            T & Bias & std$(\hat{\theta}_h)$ & RMSE & $C_{b}$(95\%) & $C_{a}$(95\%) & Bias & std$(\hat{\theta}_h)$ & RMSE & $C_{b}$(95\%) & Bias & std$(\hat{\theta}_h)$ & RMSE & $C_{b}$(95\%) & Bias & std$(\hat{\theta}_h)$ & RMSE & $C_{b}$(95\%) \\
            \cmidrule(lr){1-1}\cmidrule(lr){2-6}\cmidrule(lr){7-10}\cmidrule(lr){11-14}\cmidrule(lr){15-18}
            125 & -0.020 & 0.749 & 0.749 & 0.969 & 0.959 & 0.379 & 0.438 & 0.579 & 0.530 & 0.275 & 0.400 & 0.486 & 0.732 & 0.099 & 0.361 & 0.375 & 0.868 \\
            250 & 0.025 & 0.406 & 0.407 & 0.954 & 0.953 & 0.276 & 0.275 & 0.390 & 0.512 & 0.198 & 0.256 & 0.323 & 0.752 & 0.110 & 0.245 & 0.269 & 0.872 \\
            500 & 0.027 & 0.212 & 0.214 & 0.944 & 0.944 & 0.215 & 0.199 & 0.293 & 0.469 & 0.139 & 0.185 & 0.231 & 0.776 & 0.139 & 0.181 & 0.228 & 0.821 \\
            1'000 & 0.015 & 0.131 & 0.132 & 0.962 & 0.959 & 0.174 & 0.130 & 0.217 & 0.431 & 0.095 & 0.121 & 0.154 & 0.809 & 0.135 & 0.121 & 0.181 & 0.747 \\
            8'000 & 0.009 & 0.044 & 0.045 & 0.957 & 0.955 & 0.177 & 0.052 & 0.184 & 0.015 & 0.042 & 0.044 & 0.061 & 0.818 & 0.138 & 0.043 & 0.145 & 0.094 \\
            \midrule
            \multicolumn{18}{c}{$h=1$, $\theta_0^{(h)}=0.1992$} \\
             & \multicolumn{5}{c}{DML} & \multicolumn{4}{c}{RA} & \multicolumn{4}{c}{DR} & \multicolumn{4}{c}{LP} \\
            T & Bias & std$(\hat{\theta}_h)$ & RMSE & $C_{b}$(95\%) & $C_{a}$(95\%) & Bias & std$(\hat{\theta}_h)$ & RMSE & $C_{b}$(95\%) & Bias & std$(\hat{\theta}_h)$ & RMSE & $C_{b}$(95\%) & Bias & std$(\hat{\theta}_h)$ & RMSE & $C_{b}$(95\%) \\
            \cmidrule(lr){1-1}\cmidrule(lr){2-6}\cmidrule(lr){7-10}\cmidrule(lr){11-14}\cmidrule(lr){15-18}
            125 & 0.017 & 0.762 & 0.762 & 0.958 & 0.953 & 0.412 & 0.439 & 0.602 & 0.467 & 0.296 & 0.402 & 0.499 & 0.729 & 0.072 & 0.369 & 0.376 & 0.904 \\
            250 & 0.057 & 0.398 & 0.402 & 0.971 & 0.951 & 0.306 & 0.284 & 0.418 & 0.441 & 0.215 & 0.263 & 0.340 & 0.741 & 0.099 & 0.269 & 0.287 & 0.912 \\
            500 & 0.053 & 0.220 & 0.226 & 0.949 & 0.941 & 0.251 & 0.197 & 0.319 & 0.330 & 0.163 & 0.182 & 0.244 & 0.708 & 0.135 & 0.190 & 0.233 & 0.872 \\
            1'000 & 0.037 & 0.142 & 0.147 & 0.951 & 0.951 & 0.211 & 0.130 & 0.248 & 0.236 & 0.119 & 0.123 & 0.171 & 0.731 & 0.128 & 0.134 & 0.186 & 0.829 \\
            8'000 & 0.014 & 0.045 & 0.047 & 0.934 & 0.934 & 0.219 & 0.052 & 0.225 & 0.001 & 0.052 & 0.045 & 0.069 & 0.748 & 0.142 & 0.046 & 0.149 & 0.134 \\
            \midrule
            \multicolumn{18}{c}{$h=2$, $\theta_0^{(h)}=0.1195$} \\
             & \multicolumn{5}{c}{DML} & \multicolumn{4}{c}{RA} & \multicolumn{4}{c}{DR} & \multicolumn{4}{c}{LP} \\
            T & Bias & std$(\hat{\theta}_h)$ & RMSE & $C_{b}$(95\%) & $C_{a}$(95\%) & Bias & std$(\hat{\theta}_h)$ & RMSE & $C_{b}$(95\%) & Bias & std$(\hat{\theta}_h)$ & RMSE & $C_{b}$(95\%) & Bias & std$(\hat{\theta}_h)$ & RMSE & $C_{b}$(95\%) \\
            \cmidrule(lr){1-1}\cmidrule(lr){2-6}\cmidrule(lr){7-10}\cmidrule(lr){11-14}\cmidrule(lr){15-18}
            125 & 0.058 & 0.821 & 0.823 & 0.951 & 0.942 & 0.394 & 0.466 & 0.610 & 0.458 & 0.288 & 0.433 & 0.520 & 0.713 & 0.043 & 0.412 & 0.414 & 0.930 \\
            250 & 0.067 & 0.460 & 0.464 & 0.967 & 0.938 & 0.289 & 0.299 & 0.416 & 0.434 & 0.205 & 0.280 & 0.347 & 0.797 & 0.081 & 0.313 & 0.324 & 0.926 \\
            500 & 0.067 & 0.248 & 0.257 & 0.942 & 0.939 & 0.243 & 0.206 & 0.319 & 0.307 & 0.162 & 0.197 & 0.255 & 0.725 & 0.109 & 0.219 & 0.245 & 0.911 \\
            1'000 & 0.040 & 0.161 & 0.166 & 0.945 & 0.945 & 0.196 & 0.145 & 0.244 & 0.238 & 0.111 & 0.141 & 0.179 & 0.764 & 0.112 & 0.154 & 0.191 & 0.886 \\
            8'000 & 0.014 & 0.052 & 0.054 & 0.933 & 0.932 & 0.213 & 0.057 & 0.221 & 0.002 & 0.051 & 0.050 & 0.071 & 0.787 & 0.126 & 0.054 & 0.137 & 0.367 \\
            \midrule
            \multicolumn{18}{c}{$h=3$, $\theta_0^{(h)}=0.0717$} \\
             & \multicolumn{5}{c}{DML} & \multicolumn{4}{c}{RA} & \multicolumn{4}{c}{DR} & \multicolumn{4}{c}{LP} \\
            T & Bias & std$(\hat{\theta}_h)$ & RMSE & $C_{b}$(95\%) & $C_{a}$(95\%) & Bias & std$(\hat{\theta}_h)$ & RMSE & $C_{b}$(95\%) & Bias & std$(\hat{\theta}_h)$ & RMSE & $C_{b}$(95\%) & Bias & std$(\hat{\theta}_h)$ & RMSE & $C_{b}$(95\%) \\
            \cmidrule(lr){1-1}\cmidrule(lr){2-6}\cmidrule(lr){7-10}\cmidrule(lr){11-14}\cmidrule(lr){15-18}
            125 & 0.057 & 0.849 & 0.851 & 0.943 & 0.939 & 0.347 & 0.487 & 0.598 & 0.475 & 0.257 & 0.459 & 0.526 & 0.756 & 0.026 & 0.461 & 0.461 & 0.930 \\
            250 & 0.072 & 0.483 & 0.489 & 0.954 & 0.936 & 0.257 & 0.326 & 0.416 & 0.416 & 0.188 & 0.310 & 0.363 & 0.816 & 0.060 & 0.350 & 0.355 & 0.929 \\
            500 & 0.066 & 0.273 & 0.281 & 0.943 & 0.941 & 0.216 & 0.224 & 0.311 & 0.337 & 0.147 & 0.216 & 0.261 & 0.759 & 0.090 & 0.243 & 0.259 & 0.937 \\
            1'000 & 0.045 & 0.181 & 0.187 & 0.935 & 0.932 & 0.181 & 0.159 & 0.241 & 0.252 & 0.107 & 0.156 & 0.189 & 0.786 & 0.094 & 0.177 & 0.201 & 0.908 \\
            8'000 & 0.013 & 0.059 & 0.061 & 0.931 & 0.931 & 0.195 & 0.063 & 0.205 & 0.005 & 0.047 & 0.057 & 0.074 & 0.818 & 0.107 & 0.061 & 0.123 & 0.579 \\
            \midrule
            \multicolumn{18}{c}{$h=4$, $\theta_0^{(h)}=0.0430$} \\
             & \multicolumn{5}{c}{DML} & \multicolumn{4}{c}{RA} & \multicolumn{4}{c}{DR} & \multicolumn{4}{c}{LP} \\
            T & Bias & std$(\hat{\theta}_h)$ & RMSE & $C_{b}$(95\%) & $C_{a}$(95\%) & Bias & std$(\hat{\theta}_h)$ & RMSE & $C_{b}$(95\%) & Bias & std$(\hat{\theta}_h)$ & RMSE & $C_{b}$(95\%) & Bias & std$(\hat{\theta}_h)$ & RMSE & $C_{b}$(95\%) \\
            \cmidrule(lr){1-1}\cmidrule(lr){2-6}\cmidrule(lr){7-10}\cmidrule(lr){11-14}\cmidrule(lr){15-18}
            125 & 0.071 & 0.896 & 0.899 & 0.941 & 0.934 & 0.299 & 0.513 & 0.594 & 0.487 & 0.222 & 0.486 & 0.534 & 0.771 & 0.027 & 0.481 & 0.482 & 0.935 \\
            250 & 0.084 & 0.532 & 0.539 & 0.936 & 0.931 & 0.233 & 0.352 & 0.422 & 0.455 & 0.173 & 0.336 & 0.378 & 0.758 & 0.061 & 0.391 & 0.396 & 0.918 \\
            500 & 0.057 & 0.308 & 0.313 & 0.937 & 0.930 & 0.188 & 0.241 & 0.305 & 0.404 & 0.129 & 0.235 & 0.269 & 0.790 & 0.073 & 0.263 & 0.273 & 0.932 \\
            1'000 & 0.040 & 0.200 & 0.203 & 0.936 & 0.932 & 0.157 & 0.173 & 0.234 & 0.278 & 0.093 & 0.170 & 0.193 & 0.807 & 0.071 & 0.193 & 0.206 & 0.937 \\
            8'000 & 0.014 & 0.066 & 0.068 & 0.931 & 0.930 & 0.176 & 0.067 & 0.189 & 0.009 & 0.044 & 0.064 & 0.078 & 0.846 & 0.088 & 0.067 & 0.111 & 0.740 \\
            \midrule
            \multicolumn{18}{c}{$h=5$, $\theta_0^{(h)}=0.0258$} \\
             & \multicolumn{5}{c}{DML} & \multicolumn{4}{c}{RA} & \multicolumn{4}{c}{DR} & \multicolumn{4}{c}{LP} \\
            T & Bias & std$(\hat{\theta}_h)$ & RMSE & $C_{b}$(95\%) & $C_{a}$(95\%) & Bias & std$(\hat{\theta}_h)$ & RMSE & $C_{b}$(95\%) & Bias & std$(\hat{\theta}_h)$ & RMSE & $C_{b}$(95\%) & Bias & std$(\hat{\theta}_h)$ & RMSE & $C_{b}$(95\%) \\
            \cmidrule(lr){1-1}\cmidrule(lr){2-6}\cmidrule(lr){7-10}\cmidrule(lr){11-14}\cmidrule(lr){15-18}
            125 & 0.069 & 0.951 & 0.954 & 0.933 & 0.929 & 0.254 & 0.542 & 0.598 & 0.519 & 0.192 & 0.516 & 0.551 & 0.775 & 0.026 & 0.523 & 0.524 & 0.922 \\
            250 & 0.081 & 0.634 & 0.639 & 0.931 & 0.929 & 0.199 & 0.364 & 0.415 & 0.493 & 0.151 & 0.352 & 0.383 & 0.791 & 0.064 & 0.420 & 0.425 & 0.921 \\
            500 & 0.056 & 0.337 & 0.342 & 0.936 & 0.931 & 0.165 & 0.264 & 0.312 & 0.414 & 0.117 & 0.257 & 0.283 & 0.805 & 0.060 & 0.285 & 0.292 & 0.941 \\
            1'000 & 0.039 & 0.217 & 0.221 & 0.941 & 0.939 & 0.141 & 0.186 & 0.233 & 0.340 & 0.086 & 0.183 & 0.202 & 0.834 & 0.051 & 0.208 & 0.214 & 0.938 \\
            8'000 & 0.013 & 0.070 & 0.072 & 0.931 & 0.931 & 0.155 & 0.072 & 0.171 & 0.027 & 0.039 & 0.069 & 0.079 & 0.875 & 0.069 & 0.073 & 0.100 & 0.844 \\
            \bottomrule\end{tabular}
    \begin{tablenotes}[flushleft]
    \item \textsc{Note}: The table depicts simulation results across $N=1'000$ draws obtained for the scenario with
            $K$ independent stochastic processes, with $K=10$. Except for the LP estimator, nuisance functions are estimated with random forest.
            For the DML estimator we set $k_T=T/10$ to obtain estimation samples of the same size as in the setting with one stochastic process. For sample size $T=125$, probabilities are winsorized at 1\%.
            The parameters of the data generating process are $n=12$,
            $\sigma_\epsilon=1.0$, $\gamma=0.6$, $p=2$,
            $q=1$, $\sigma_u=1.0$, $\alpha_A=0.3$,
            $\alpha_M=0.3$, $\rho_A=0.35$, $\rho_M=0.7$,
            $\beta_1=0.3$, $\beta_2=0.5$.
            $C_a$($\cdot$) and $C_b$($\cdot$) in the tables denote the coverage at the given confidence level using asymptotic and fixed-bandwidth critical values, respectively.
    \end{tablenotes}
    \end{threeparttable}}
    \end{table}

    \begin{table}[htbp]
    \caption{Simulation results for a nonlinear DGP ($n=20$, $\sigma_\epsilon=1.0$) and random forest nuisance function estimates}
    \label{tab:result_variations_confounders}
    \resizebox{\textwidth}{!}{\begin{threeparttable}\begin{tabular}{rrrrrrrrrrrrrrrrrr}
            \toprule
            \multicolumn{18}{c}{$h=0$, $\theta_0^{(h)}=0.3333$} \\
                    & \multicolumn{5}{c}{DML} & \multicolumn{4}{c}{RA} & \multicolumn{4}{c}{DR} & \multicolumn{4}{c}{LP} \\
            T & Bias & std$(\hat{\theta}_h)$ & RMSE & $C_{b}$(95\%) & $C_{a}$(95\%) & Bias & std$(\hat{\theta}_h)$ & RMSE & $C_{b}$(95\%) & Bias & std$(\hat{\theta}_h)$ & RMSE & $C_{b}$(95\%) & Bias & std$(\hat{\theta}_h)$ & RMSE & $C_{b}$(95\%) \\
            \cmidrule(lr){1-1}\cmidrule(lr){2-6}\cmidrule(lr){7-10}\cmidrule(lr){11-14}\cmidrule(lr){15-18}
            125 & 0.180 & 0.791 & 0.811 & 0.937 & 0.920 & 0.382 & 0.454 & 0.593 & 0.524 & 0.293 & 0.418 & 0.511 & 0.712 & 0.101 & 0.370 & 0.383 & 0.823 \\
            250 & 0.082 & 0.385 & 0.393 & 0.959 & 0.945 & 0.279 & 0.289 & 0.402 & 0.519 & 0.210 & 0.269 & 0.342 & 0.719 & 0.117 & 0.248 & 0.274 & 0.857 \\
            500 & 0.037 & 0.213 & 0.216 & 0.953 & 0.940 & 0.199 & 0.194 & 0.278 & 0.529 & 0.136 & 0.183 & 0.228 & 0.780 & 0.122 & 0.175 & 0.213 & 0.850 \\
            1'000 & 0.031 & 0.133 & 0.136 & 0.952 & 0.949 & 0.165 & 0.127 & 0.208 & 0.465 & 0.098 & 0.120 & 0.155 & 0.807 & 0.126 & 0.118 & 0.173 & 0.769 \\
            8'000 & 0.010 & 0.044 & 0.045 & 0.934 & 0.934 & 0.156 & 0.048 & 0.163 & 0.018 & 0.038 & 0.044 & 0.058 & 0.835 & 0.133 & 0.043 & 0.140 & 0.123 \\
            \midrule
            \multicolumn{18}{c}{$h=1$, $\theta_0^{(h)}=0.2000$} \\
                    & \multicolumn{5}{c}{DML} & \multicolumn{4}{c}{RA} & \multicolumn{4}{c}{DR} & \multicolumn{4}{c}{LP} \\
            T & Bias & std$(\hat{\theta}_h)$ & RMSE & $C_{b}$(95\%) & $C_{a}$(95\%) & Bias & std$(\hat{\theta}_h)$ & RMSE & $C_{b}$(95\%) & Bias & std$(\hat{\theta}_h)$ & RMSE & $C_{b}$(95\%) & Bias & std$(\hat{\theta}_h)$ & RMSE & $C_{b}$(95\%) \\
            \cmidrule(lr){1-1}\cmidrule(lr){2-6}\cmidrule(lr){7-10}\cmidrule(lr){11-14}\cmidrule(lr){15-18}
            125 & 0.207 & 0.827 & 0.852 & 0.935 & 0.920 & 0.418 & 0.449 & 0.613 & 0.488 & 0.317 & 0.411 & 0.519 & 0.698 & 0.058 & 0.381 & 0.386 & 0.884 \\
            250 & 0.103 & 0.393 & 0.406 & 0.965 & 0.952 & 0.301 & 0.283 & 0.413 & 0.447 & 0.220 & 0.261 & 0.341 & 0.724 & 0.088 & 0.268 & 0.282 & 0.900 \\
            500 & 0.061 & 0.224 & 0.232 & 0.948 & 0.945 & 0.235 & 0.194 & 0.304 & 0.381 & 0.158 & 0.183 & 0.242 & 0.712 & 0.114 & 0.186 & 0.218 & 0.898 \\
            1'000 & 0.053 & 0.144 & 0.154 & 0.931 & 0.929 & 0.203 & 0.134 & 0.243 & 0.283 & 0.121 & 0.127 & 0.175 & 0.734 & 0.126 & 0.132 & 0.182 & 0.827 \\
            8'000 & 0.020 & 0.045 & 0.049 & 0.925 & 0.924 & 0.202 & 0.047 & 0.207 & 0.001 & 0.052 & 0.043 & 0.068 & 0.738 & 0.139 & 0.047 & 0.147 & 0.153 \\
            \midrule
            \multicolumn{18}{c}{$h=2$, $\theta_0^{(h)}=0.1200$} \\
                    & \multicolumn{5}{c}{DML} & \multicolumn{4}{c}{RA} & \multicolumn{4}{c}{DR} & \multicolumn{4}{c}{LP} \\
            T & Bias & std$(\hat{\theta}_h)$ & RMSE & $C_{b}$(95\%) & $C_{a}$(95\%) & Bias & std$(\hat{\theta}_h)$ & RMSE & $C_{b}$(95\%) & Bias & std$(\hat{\theta}_h)$ & RMSE & $C_{b}$(95\%) & Bias & std$(\hat{\theta}_h)$ & RMSE & $C_{b}$(95\%) \\
            \cmidrule(lr){1-1}\cmidrule(lr){2-6}\cmidrule(lr){7-10}\cmidrule(lr){11-14}\cmidrule(lr){15-18}
            125 & 0.222 & 0.878 & 0.905 & 0.906 & 0.898 & 0.378 & 0.477 & 0.609 & 0.488 & 0.283 & 0.439 & 0.522 & 0.728 & 0.019 & 0.424 & 0.425 & 0.903 \\
            250 & 0.127 & 0.433 & 0.451 & 0.937 & 0.926 & 0.294 & 0.293 & 0.415 & 0.434 & 0.218 & 0.272 & 0.349 & 0.747 & 0.075 & 0.296 & 0.305 & 0.936 \\
            500 & 0.075 & 0.249 & 0.260 & 0.941 & 0.937 & 0.224 & 0.209 & 0.307 & 0.357 & 0.153 & 0.199 & 0.251 & 0.716 & 0.095 & 0.214 & 0.234 & 0.927 \\
            1'000 & 0.049 & 0.164 & 0.171 & 0.928 & 0.927 & 0.185 & 0.145 & 0.235 & 0.260 & 0.106 & 0.140 & 0.175 & 0.764 & 0.098 & 0.149 & 0.179 & 0.901 \\
            8'000 & 0.022 & 0.051 & 0.056 & 0.919 & 0.919 & 0.198 & 0.052 & 0.205 & 0.001 & 0.052 & 0.049 & 0.072 & 0.758 & 0.127 & 0.054 & 0.138 & 0.362 \\
            \midrule
            \multicolumn{18}{c}{$h=3$, $\theta_0^{(h)}=0.0720$} \\
                    & \multicolumn{5}{c}{DML} & \multicolumn{4}{c}{RA} & \multicolumn{4}{c}{DR} & \multicolumn{4}{c}{LP} \\
            T & Bias & std$(\hat{\theta}_h)$ & RMSE & $C_{b}$(95\%) & $C_{a}$(95\%) & Bias & std$(\hat{\theta}_h)$ & RMSE & $C_{b}$(95\%) & Bias & std$(\hat{\theta}_h)$ & RMSE & $C_{b}$(95\%) & Bias & std$(\hat{\theta}_h)$ & RMSE & $C_{b}$(95\%) \\
            \cmidrule(lr){1-1}\cmidrule(lr){2-6}\cmidrule(lr){7-10}\cmidrule(lr){11-14}\cmidrule(lr){15-18}
            125 & 0.228 & 0.987 & 1.013 & 0.924 & 0.901 & 0.353 & 0.486 & 0.601 & 0.534 & 0.271 & 0.450 & 0.526 & 0.747 & 0.022 & 0.454 & 0.454 & 0.914 \\
            250 & 0.137 & 0.487 & 0.506 & 0.932 & 0.924 & 0.266 & 0.328 & 0.422 & 0.467 & 0.201 & 0.311 & 0.370 & 0.746 & 0.053 & 0.340 & 0.344 & 0.925 \\
            500 & 0.083 & 0.289 & 0.301 & 0.941 & 0.936 & 0.216 & 0.228 & 0.314 & 0.376 & 0.152 & 0.218 & 0.266 & 0.735 & 0.085 & 0.244 & 0.259 & 0.925 \\
            1'000 & 0.058 & 0.176 & 0.186 & 0.940 & 0.939 & 0.180 & 0.152 & 0.236 & 0.287 & 0.109 & 0.149 & 0.185 & 0.784 & 0.089 & 0.165 & 0.187 & 0.926 \\
            8'000 & 0.024 & 0.058 & 0.062 & 0.917 & 0.917 & 0.185 & 0.056 & 0.194 & 0.001 & 0.051 & 0.055 & 0.075 & 0.807 & 0.111 & 0.061 & 0.127 & 0.563 \\
            \midrule
            \multicolumn{18}{c}{$h=4$, $\theta_0^{(h)}=0.0432$} \\
                    & \multicolumn{5}{c}{DML} & \multicolumn{4}{c}{RA} & \multicolumn{4}{c}{DR} & \multicolumn{4}{c}{LP} \\
            T & Bias & std$(\hat{\theta}_h)$ & RMSE & $C_{b}$(95\%) & $C_{a}$(95\%) & Bias & std$(\hat{\theta}_h)$ & RMSE & $C_{b}$(95\%) & Bias & std$(\hat{\theta}_h)$ & RMSE & $C_{b}$(95\%) & Bias & std$(\hat{\theta}_h)$ & RMSE & $C_{b}$(95\%) \\
            \cmidrule(lr){1-1}\cmidrule(lr){2-6}\cmidrule(lr){7-10}\cmidrule(lr){11-14}\cmidrule(lr){15-18}
            125 & 0.225 & 1.014 & 1.038 & 0.912 & 0.902 & 0.315 & 0.520 & 0.608 & 0.555 & 0.248 & 0.491 & 0.550 & 0.752 & 0.022 & 0.500 & 0.500 & 0.894 \\
            250 & 0.130 & 0.543 & 0.558 & 0.927 & 0.906 & 0.246 & 0.350 & 0.427 & 0.499 & 0.186 & 0.334 & 0.383 & 0.753 & 0.038 & 0.365 & 0.367 & 0.922 \\
            500 & 0.079 & 0.321 & 0.331 & 0.933 & 0.933 & 0.194 & 0.241 & 0.310 & 0.431 & 0.139 & 0.233 & 0.271 & 0.771 & 0.062 & 0.261 & 0.268 & 0.935 \\
            1'000 & 0.058 & 0.194 & 0.202 & 0.948 & 0.944 & 0.166 & 0.166 & 0.234 & 0.309 & 0.101 & 0.162 & 0.191 & 0.811 & 0.072 & 0.181 & 0.195 & 0.937 \\
            8'000 & 0.022 & 0.065 & 0.068 & 0.930 & 0.929 & 0.168 & 0.062 & 0.179 & 0.014 & 0.047 & 0.061 & 0.077 & 0.842 & 0.092 & 0.067 & 0.113 & 0.736 \\
            \midrule
            \multicolumn{18}{c}{$h=5$, $\theta_0^{(h)}=0.0259$} \\
                    & \multicolumn{5}{c}{DML} & \multicolumn{4}{c}{RA} & \multicolumn{4}{c}{DR} & \multicolumn{4}{c}{LP} \\
            T & Bias & std$(\hat{\theta}_h)$ & RMSE & $C_{b}$(95\%) & $C_{a}$(95\%) & Bias & std$(\hat{\theta}_h)$ & RMSE & $C_{b}$(95\%) & Bias & std$(\hat{\theta}_h)$ & RMSE & $C_{b}$(95\%) & Bias & std$(\hat{\theta}_h)$ & RMSE & $C_{b}$(95\%) \\
            \cmidrule(lr){1-1}\cmidrule(lr){2-6}\cmidrule(lr){7-10}\cmidrule(lr){11-14}\cmidrule(lr){15-18}
            125 & 0.208 & 1.055 & 1.075 & 0.914 & 0.903 & 0.250 & 0.526 & 0.582 & 0.583 & 0.196 & 0.496 & 0.533 & 0.801 & -0.002 & 0.515 & 0.515 & 0.911 \\
            250 & 0.136 & 0.593 & 0.608 & 0.919 & 0.902 & 0.217 & 0.361 & 0.421 & 0.548 & 0.170 & 0.343 & 0.383 & 0.778 & 0.027 & 0.381 & 0.382 & 0.942 \\
            500 & 0.070 & 0.339 & 0.346 & 0.925 & 0.925 & 0.167 & 0.253 & 0.303 & 0.482 & 0.119 & 0.245 & 0.272 & 0.802 & 0.042 & 0.277 & 0.280 & 0.937 \\
            1'000 & 0.055 & 0.213 & 0.220 & 0.934 & 0.932 & 0.151 & 0.178 & 0.234 & 0.356 & 0.094 & 0.175 & 0.198 & 0.838 & 0.055 & 0.199 & 0.206 & 0.935 \\
            8'000 & 0.021 & 0.070 & 0.073 & 0.934 & 0.933 & 0.152 & 0.067 & 0.166 & 0.025 & 0.043 & 0.067 & 0.079 & 0.859 & 0.075 & 0.071 & 0.103 & 0.822 \\
            \bottomrule\end{tabular}
            \begin{tablenotes}[flushleft]
            \item \textsc{Note}: The table depicts simulation results across $N=1'000$ draws obtained for the scenario with
                    one stochastic process. Except for the LP estimator, nuisance functions are estimated with random forest.
                    For the DML estimator we use 10-fold cross-fitting and set $k_T=T/10$. For sample size $T=125$, probabilities are winsorized at 1\%.
                    The parameters of the data generating process are $n=20$,
                    $\sigma_\epsilon=1.0$, $\gamma=0.6$, $p=2$,
                    $q=1$, $\sigma_u=1.0$, $\alpha_A=0.3$,
                    $\alpha_M=0.3$, $\rho_A=0.35$, $\rho_M=0.7$,
                    $\beta_1=0.3$, $\beta_2=0.5$.
                    $C_a$($\cdot$) and $C_b$($\cdot$) in the tables denote the coverage at the given confidence level using asymptotic and fixed-bandwidth critical values, respectively.
            \end{tablenotes}
            \end{threeparttable}}
    \end{table}

    \begin{table}[htbp]
      \caption{Simulation results for a nonlinear DGP ($n=12$, $\sigma_\epsilon=3.0$) and random forest nuisance function estimates}
      \label{tab:result_variations_noise}
      \resizebox{\textwidth}{!}{\begin{threeparttable}\begin{tabular}{rrrrrrrrrrrrrrrrrr}
            \toprule
            \multicolumn{18}{c}{$h=0$, $\theta_0^{(h)}=0.3321$} \\
             & \multicolumn{5}{c}{DML} & \multicolumn{4}{c}{RA} & \multicolumn{4}{c}{DR} & \multicolumn{4}{c}{LP} \\
            T & Bias & std$(\hat{\theta}_h)$ & RMSE & $C_{b}$(95\%) & $C_{a}$(95\%) & Bias & std$(\hat{\theta}_h)$ & RMSE & $C_{b}$(95\%) & Bias & std$(\hat{\theta}_h)$ & RMSE & $C_{b}$(95\%) & Bias & std$(\hat{\theta}_h)$ & RMSE & $C_{b}$(95\%) \\
            \cmidrule(lr){1-1}\cmidrule(lr){2-6}\cmidrule(lr){7-10}\cmidrule(lr){11-14}\cmidrule(lr){15-18}
            125 & 0.055 & 1.287 & 1.288 & 0.941 & 0.930 & 0.364 & 0.789 & 0.869 & 0.497 & 0.263 & 0.750 & 0.794 & 0.758 & 0.082 & 0.713 & 0.717 & 0.907 \\
            250 & 0.060 & 0.779 & 0.781 & 0.953 & 0.943 & 0.273 & 0.534 & 0.600 & 0.528 & 0.203 & 0.516 & 0.555 & 0.800 & 0.142 & 0.488 & 0.508 & 0.917 \\
            500 & 0.015 & 0.429 & 0.429 & 0.951 & 0.945 & 0.188 & 0.368 & 0.413 & 0.471 & 0.120 & 0.357 & 0.376 & 0.824 & 0.142 & 0.342 & 0.370 & 0.917 \\
            1'000 & 0.019 & 0.285 & 0.285 & 0.946 & 0.945 & 0.161 & 0.264 & 0.309 & 0.429 & 0.088 & 0.260 & 0.274 & 0.838 & 0.146 & 0.245 & 0.285 & 0.871 \\
            8'000 & 0.003 & 0.091 & 0.091 & 0.951 & 0.950 & 0.161 & 0.089 & 0.184 & 0.106 & 0.033 & 0.088 & 0.094 & 0.918 & 0.143 & 0.083 & 0.165 & 0.599 \\
            \midrule
            \multicolumn{18}{c}{$h=1$, $\theta_0^{(h)}=0.1992$} \\
             & \multicolumn{5}{c}{DML} & \multicolumn{4}{c}{RA} & \multicolumn{4}{c}{DR} & \multicolumn{4}{c}{LP} \\
            T & Bias & std$(\hat{\theta}_h)$ & RMSE & $C_{b}$(95\%) & $C_{a}$(95\%) & Bias & std$(\hat{\theta}_h)$ & RMSE & $C_{b}$(95\%) & Bias & std$(\hat{\theta}_h)$ & RMSE & $C_{b}$(95\%) & Bias & std$(\hat{\theta}_h)$ & RMSE & $C_{b}$(95\%) \\
            \cmidrule(lr){1-1}\cmidrule(lr){2-6}\cmidrule(lr){7-10}\cmidrule(lr){11-14}\cmidrule(lr){15-18}
            125 & 0.101 & 1.430 & 1.433 & 0.918 & 0.918 & 0.370 & 0.846 & 0.924 & 0.523 & 0.271 & 0.815 & 0.859 & 0.759 & 0.060 & 0.782 & 0.784 & 0.922 \\
            250 & 0.047 & 1.051 & 1.052 & 0.942 & 0.936 & 0.293 & 0.580 & 0.650 & 0.459 & 0.210 & 0.565 & 0.603 & 0.769 & 0.122 & 0.555 & 0.568 & 0.931 \\
            500 & 0.043 & 0.474 & 0.476 & 0.946 & 0.945 & 0.214 & 0.388 & 0.443 & 0.429 & 0.134 & 0.381 & 0.404 & 0.832 & 0.128 & 0.392 & 0.412 & 0.920 \\
            1'000 & 0.049 & 0.310 & 0.314 & 0.949 & 0.947 & 0.198 & 0.279 & 0.342 & 0.345 & 0.111 & 0.275 & 0.297 & 0.842 & 0.153 & 0.271 & 0.311 & 0.914 \\
            8'000 & 0.010 & 0.102 & 0.102 & 0.947 & 0.946 & 0.207 & 0.097 & 0.228 & 0.047 & 0.045 & 0.098 & 0.108 & 0.913 & 0.153 & 0.096 & 0.181 & 0.642 \\
            \midrule
            \multicolumn{18}{c}{$h=2$, $\theta_0^{(h)}=0.1195$} \\
             & \multicolumn{5}{c}{DML} & \multicolumn{4}{c}{RA} & \multicolumn{4}{c}{DR} & \multicolumn{4}{c}{LP} \\
            T & Bias & std$(\hat{\theta}_h)$ & RMSE & $C_{b}$(95\%) & $C_{a}$(95\%) & Bias & std$(\hat{\theta}_h)$ & RMSE & $C_{b}$(95\%) & Bias & std$(\hat{\theta}_h)$ & RMSE & $C_{b}$(95\%) & Bias & std$(\hat{\theta}_h)$ & RMSE & $C_{b}$(95\%) \\
            \cmidrule(lr){1-1}\cmidrule(lr){2-6}\cmidrule(lr){7-10}\cmidrule(lr){11-14}\cmidrule(lr){15-18}
            125 & 0.159 & 1.382 & 1.391 & 0.932 & 0.921 & 0.352 & 0.847 & 0.917 & 0.512 & 0.258 & 0.816 & 0.856 & 0.784 & 0.033 & 0.814 & 0.815 & 0.935 \\
            250 & 0.044 & 1.106 & 1.107 & 0.932 & 0.925 & 0.258 & 0.600 & 0.653 & 0.497 & 0.182 & 0.592 & 0.619 & 0.777 & 0.080 & 0.599 & 0.604 & 0.925 \\
            500 & 0.059 & 0.503 & 0.507 & 0.941 & 0.938 & 0.202 & 0.414 & 0.460 & 0.417 & 0.130 & 0.408 & 0.428 & 0.814 & 0.111 & 0.420 & 0.434 & 0.933 \\
            1'000 & 0.058 & 0.340 & 0.345 & 0.931 & 0.931 & 0.194 & 0.302 & 0.360 & 0.336 & 0.114 & 0.300 & 0.321 & 0.814 & 0.143 & 0.302 & 0.334 & 0.913 \\
            8'000 & 0.010 & 0.108 & 0.109 & 0.962 & 0.962 & 0.199 & 0.102 & 0.224 & 0.040 & 0.043 & 0.104 & 0.112 & 0.908 & 0.138 & 0.104 & 0.173 & 0.721 \\
            \midrule
            \multicolumn{18}{c}{$h=3$, $\theta_0^{(h)}=0.0717$} \\
             & \multicolumn{5}{c}{DML} & \multicolumn{4}{c}{RA} & \multicolumn{4}{c}{DR} & \multicolumn{4}{c}{LP} \\
            T & Bias & std$(\hat{\theta}_h)$ & RMSE & $C_{b}$(95\%) & $C_{a}$(95\%) & Bias & std$(\hat{\theta}_h)$ & RMSE & $C_{b}$(95\%) & Bias & std$(\hat{\theta}_h)$ & RMSE & $C_{b}$(95\%) & Bias & std$(\hat{\theta}_h)$ & RMSE & $C_{b}$(95\%) \\
            \cmidrule(lr){1-1}\cmidrule(lr){2-6}\cmidrule(lr){7-10}\cmidrule(lr){11-14}\cmidrule(lr){15-18}
            125 & 0.116 & 1.482 & 1.486 & 0.936 & 0.923 & 0.284 & 0.852 & 0.898 & 0.556 & 0.201 & 0.821 & 0.845 & 0.795 & -0.010 & 0.825 & 0.825 & 0.936 \\
            250 & 0.074 & 0.860 & 0.864 & 0.930 & 0.928 & 0.233 & 0.602 & 0.646 & 0.481 & 0.160 & 0.586 & 0.608 & 0.795 & 0.058 & 0.600 & 0.603 & 0.934 \\
            500 & 0.062 & 0.536 & 0.540 & 0.931 & 0.923 & 0.191 & 0.432 & 0.472 & 0.409 & 0.127 & 0.429 & 0.447 & 0.800 & 0.102 & 0.448 & 0.459 & 0.929 \\
            1'000 & 0.060 & 0.353 & 0.358 & 0.924 & 0.922 & 0.183 & 0.310 & 0.360 & 0.331 & 0.112 & 0.311 & 0.331 & 0.815 & 0.128 & 0.319 & 0.344 & 0.911 \\
            8'000 & 0.010 & 0.114 & 0.114 & 0.950 & 0.950 & 0.180 & 0.107 & 0.210 & 0.046 & 0.040 & 0.110 & 0.117 & 0.905 & 0.119 & 0.110 & 0.162 & 0.808 \\
            \midrule
            \multicolumn{18}{c}{$h=4$, $\theta_0^{(h)}=0.0430$} \\
             & \multicolumn{5}{c}{DML} & \multicolumn{4}{c}{RA} & \multicolumn{4}{c}{DR} & \multicolumn{4}{c}{LP} \\
            T & Bias & std$(\hat{\theta}_h)$ & RMSE & $C_{b}$(95\%) & $C_{a}$(95\%) & Bias & std$(\hat{\theta}_h)$ & RMSE & $C_{b}$(95\%) & Bias & std$(\hat{\theta}_h)$ & RMSE & $C_{b}$(95\%) & Bias & std$(\hat{\theta}_h)$ & RMSE & $C_{b}$(95\%) \\
            \cmidrule(lr){1-1}\cmidrule(lr){2-6}\cmidrule(lr){7-10}\cmidrule(lr){11-14}\cmidrule(lr){15-18}
            125 & 0.145 & 1.489 & 1.496 & 0.935 & 0.915 & 0.252 & 0.832 & 0.869 & 0.549 & 0.187 & 0.811 & 0.832 & 0.831 & 0.012 & 0.833 & 0.833 & 0.937 \\
            250 & 0.041 & 1.358 & 1.358 & 0.916 & 0.916 & 0.208 & 0.639 & 0.672 & 0.485 & 0.150 & 0.623 & 0.641 & 0.787 & 0.058 & 0.631 & 0.634 & 0.935 \\
            500 & 0.039 & 0.576 & 0.578 & 0.938 & 0.926 & 0.151 & 0.448 & 0.473 & 0.421 & 0.097 & 0.443 & 0.453 & 0.821 & 0.063 & 0.456 & 0.461 & 0.929 \\
            1'000 & 0.055 & 0.363 & 0.367 & 0.938 & 0.936 & 0.160 & 0.316 & 0.354 & 0.345 & 0.098 & 0.315 & 0.330 & 0.835 & 0.104 & 0.324 & 0.341 & 0.920 \\
            8'000 & 0.010 & 0.117 & 0.118 & 0.949 & 0.949 & 0.159 & 0.110 & 0.194 & 0.056 & 0.037 & 0.113 & 0.119 & 0.911 & 0.100 & 0.110 & 0.148 & 0.859 \\
            \midrule
            \multicolumn{18}{c}{$h=5$, $\theta_0^{(h)}=0.0258$} \\
             & \multicolumn{5}{c}{DML} & \multicolumn{4}{c}{RA} & \multicolumn{4}{c}{DR} & \multicolumn{4}{c}{LP} \\
            T & Bias & std$(\hat{\theta}_h)$ & RMSE & $C_{b}$(95\%) & $C_{a}$(95\%) & Bias & std$(\hat{\theta}_h)$ & RMSE & $C_{b}$(95\%) & Bias & std$(\hat{\theta}_h)$ & RMSE & $C_{b}$(95\%) & Bias & std$(\hat{\theta}_h)$ & RMSE & $C_{b}$(95\%) \\
            \cmidrule(lr){1-1}\cmidrule(lr){2-6}\cmidrule(lr){7-10}\cmidrule(lr){11-14}\cmidrule(lr){15-18}
            125 & 0.167 & 1.619 & 1.627 & 0.915 & 0.898 & 0.253 & 0.897 & 0.932 & 0.558 & 0.199 & 0.874 & 0.897 & 0.813 & 0.053 & 0.913 & 0.914 & 0.922 \\
            250 & 0.050 & 1.062 & 1.063 & 0.920 & 0.912 & 0.176 & 0.650 & 0.674 & 0.469 & 0.125 & 0.638 & 0.650 & 0.775 & 0.046 & 0.643 & 0.645 & 0.931 \\
            500 & 0.025 & 0.587 & 0.588 & 0.950 & 0.941 & 0.122 & 0.446 & 0.463 & 0.450 & 0.078 & 0.444 & 0.451 & 0.829 & 0.045 & 0.458 & 0.461 & 0.938 \\
            1'000 & 0.042 & 0.363 & 0.366 & 0.947 & 0.941 & 0.129 & 0.313 & 0.339 & 0.357 & 0.076 & 0.313 & 0.322 & 0.855 & 0.074 & 0.323 & 0.331 & 0.950 \\
            8'000 & 0.003 & 0.121 & 0.121 & 0.942 & 0.941 & 0.134 & 0.112 & 0.175 & 0.082 & 0.026 & 0.115 & 0.118 & 0.919 & 0.077 & 0.113 & 0.136 & 0.912 \\
            \bottomrule\end{tabular}
      \begin{tablenotes}[flushleft]
      \item \textsc{Note}: The table depicts simulation results across $N=1'000$ draws obtained for the scenario with
              one stochastic process. Except for the LP estimator, nuisance functions are estimated with random forest.
              For the DML estimator we use 10-fold cross-fitting and set $k_T=T/10$. For sample size $T=125$, probabilities are winsorized at 1\%.
              The parameters of the data generating process are $n=12$,
              $\sigma_\epsilon=3.0$, $\gamma=0.6$, $p=2$,
              $q=1$, $\sigma_u=1.0$, $\alpha_A=0.3$,
              $\alpha_M=0.3$, $\rho_A=0.35$, $\rho_M=0.7$,
              $\beta_1=0.3$, $\beta_2=0.5$.
              $C_a$($\cdot$) and $C_b$($\cdot$) in the tables denote the coverage at the given confidence level using asymptotic and fixed-bandwidth critical values, respectively.
      \end{tablenotes}
      \end{threeparttable}}
      \end{table}

      \begin{table}[htbp]
            \caption{Simulation results for an empirically calibrated DGP (see Section \ref{app:empirical_dgp}) and random forest nuisance function estimates}
            \label{tab:empircal_dgp}
            \resizebox{\textwidth}{!}{\begin{threeparttable}\begin{tabular}{rrrrrrrrrrrrrrrrrr}
                    \toprule
            \multicolumn{18}{c}{$h=0$, $\theta_0^{(h)}=0.2021$} \\
             & \multicolumn{5}{c}{DML} & \multicolumn{4}{c}{RA} & \multicolumn{4}{c}{DR} & \multicolumn{4}{c}{LP} \\
            T & Bias & std$(\hat{\theta}_h)$ & RMSE & $C_{b}$(95\%) & $C_{a}$(95\%) & Bias & std$(\hat{\theta}_h)$ & RMSE & $C_{b}$(95\%) & Bias & std$(\hat{\theta}_h)$ & RMSE & $C_{b}$(95\%) & Bias & std$(\hat{\theta}_h)$ & RMSE & $C_{b}$(95\%) \\
            \cmidrule(lr){1-1}\cmidrule(lr){2-6}\cmidrule(lr){7-10}\cmidrule(lr){11-14}\cmidrule(lr){15-18}
            125 & 0.058 & 0.343 & 0.348 & 0.966 & 0.948 & 0.244 & 0.273 & 0.366 & 0.279 & 0.203 & 0.263 & 0.332 & 0.681 & 0.146 & 0.260 & 0.298 & 0.883 \\
            250 & 0.042 & 0.207 & 0.211 & 0.935 & 0.931 & 0.210 & 0.190 & 0.283 & 0.238 & 0.168 & 0.184 & 0.249 & 0.629 & 0.153 & 0.189 & 0.243 & 0.845 \\
            500 & 0.038 & 0.130 & 0.136 & 0.949 & 0.945 & 0.177 & 0.129 & 0.219 & 0.242 & 0.130 & 0.125 & 0.180 & 0.641 & 0.154 & 0.127 & 0.200 & 0.752 \\
            1'000 & 0.034 & 0.096 & 0.101 & 0.919 & 0.919 & 0.159 & 0.090 & 0.182 & 0.148 & 0.105 & 0.088 & 0.137 & 0.625 & 0.152 & 0.091 & 0.177 & 0.580 \\
            8'000 & -0.002 & 0.031 & 0.032 & 0.949 & 0.949 & 0.153 & 0.033 & 0.156 & 0.000 & 0.035 & 0.031 & 0.047 & 0.760 & 0.153 & 0.033 & 0.157 & 0.001 \\
            \midrule
            \multicolumn{18}{c}{$h=1$, $\theta_0^{(h)}=0.0521$} \\
             & \multicolumn{5}{c}{DML} & \multicolumn{4}{c}{RA} & \multicolumn{4}{c}{DR} & \multicolumn{4}{c}{LP} \\
            T & Bias & std$(\hat{\theta}_h)$ & RMSE & $C_{b}$(95\%) & $C_{a}$(95\%) & Bias & std$(\hat{\theta}_h)$ & RMSE & $C_{b}$(95\%) & Bias & std$(\hat{\theta}_h)$ & RMSE & $C_{b}$(95\%) & Bias & std$(\hat{\theta}_h)$ & RMSE & $C_{b}$(95\%) \\
            \cmidrule(lr){1-1}\cmidrule(lr){2-6}\cmidrule(lr){7-10}\cmidrule(lr){11-14}\cmidrule(lr){15-18}
            125 & 0.040 & 0.347 & 0.349 & 0.965 & 0.939 & 0.108 & 0.271 & 0.291 & 0.330 & 0.088 & 0.267 & 0.282 & 0.791 & 0.029 & 0.274 & 0.276 & 0.941 \\
            250 & 0.025 & 0.245 & 0.247 & 0.962 & 0.944 & 0.100 & 0.192 & 0.216 & 0.258 & 0.079 & 0.191 & 0.206 & 0.789 & 0.041 & 0.199 & 0.203 & 0.936 \\
            500 & 0.027 & 0.151 & 0.153 & 0.933 & 0.929 & 0.084 & 0.137 & 0.161 & 0.265 & 0.062 & 0.137 & 0.151 & 0.792 & 0.043 & 0.141 & 0.147 & 0.918 \\
            1'000 & 0.019 & 0.100 & 0.102 & 0.947 & 0.946 & 0.073 & 0.095 & 0.119 & 0.194 & 0.047 & 0.095 & 0.106 & 0.808 & 0.042 & 0.097 & 0.106 & 0.927 \\
            8'000 & -0.000 & 0.035 & 0.035 & 0.954 & 0.953 & 0.072 & 0.034 & 0.079 & 0.019 & 0.015 & 0.034 & 0.037 & 0.894 & 0.040 & 0.034 & 0.053 & 0.780 \\
            \midrule
            \multicolumn{18}{c}{$h=2$, $\theta_0^{(h)}=0.0134$} \\
             & \multicolumn{5}{c}{DML} & \multicolumn{4}{c}{RA} & \multicolumn{4}{c}{DR} & \multicolumn{4}{c}{LP} \\
            T & Bias & std$(\hat{\theta}_h)$ & RMSE & $C_{b}$(95\%) & $C_{a}$(95\%) & Bias & std$(\hat{\theta}_h)$ & RMSE & $C_{b}$(95\%) & Bias & std$(\hat{\theta}_h)$ & RMSE & $C_{b}$(95\%) & Bias & std$(\hat{\theta}_h)$ & RMSE & $C_{b}$(95\%) \\
            \cmidrule(lr){1-1}\cmidrule(lr){2-6}\cmidrule(lr){7-10}\cmidrule(lr){11-14}\cmidrule(lr){15-18}
            125 & 0.037 & 0.362 & 0.364 & 0.969 & 0.947 & 0.067 & 0.263 & 0.271 & 0.309 & 0.054 & 0.261 & 0.266 & 0.860 & 0.007 & 0.272 & 0.272 & 0.946 \\
            250 & 0.008 & 0.228 & 0.228 & 0.955 & 0.944 & 0.045 & 0.191 & 0.196 & 0.288 & 0.034 & 0.190 & 0.193 & 0.814 & 0.003 & 0.196 & 0.196 & 0.941 \\
            500 & 0.010 & 0.150 & 0.150 & 0.945 & 0.943 & 0.038 & 0.136 & 0.141 & 0.245 & 0.026 & 0.137 & 0.139 & 0.830 & 0.005 & 0.138 & 0.138 & 0.951 \\
            1'000 & 0.006 & 0.108 & 0.108 & 0.934 & 0.934 & 0.034 & 0.099 & 0.105 & 0.184 & 0.020 & 0.101 & 0.103 & 0.828 & 0.007 & 0.101 & 0.101 & 0.938 \\
            8'000 & 0.000 & 0.036 & 0.036 & 0.954 & 0.953 & 0.036 & 0.034 & 0.050 & 0.050 & 0.008 & 0.035 & 0.036 & 0.916 & 0.010 & 0.035 & 0.037 & 0.936 \\
            \midrule
            \multicolumn{18}{c}{$h=3$, $\theta_0^{(h)}=0.0035$} \\
             & \multicolumn{5}{c}{DML} & \multicolumn{4}{c}{RA} & \multicolumn{4}{c}{DR} & \multicolumn{4}{c}{LP} \\
            T & Bias & std$(\hat{\theta}_h)$ & RMSE & $C_{b}$(95\%) & $C_{a}$(95\%) & Bias & std$(\hat{\theta}_h)$ & RMSE & $C_{b}$(95\%) & Bias & std$(\hat{\theta}_h)$ & RMSE & $C_{b}$(95\%) & Bias & std$(\hat{\theta}_h)$ & RMSE & $C_{b}$(95\%) \\
            \cmidrule(lr){1-1}\cmidrule(lr){2-6}\cmidrule(lr){7-10}\cmidrule(lr){11-14}\cmidrule(lr){15-18}
            125 & 0.032 & 0.362 & 0.363 & 0.950 & 0.936 & 0.042 & 0.282 & 0.285 & 0.331 & 0.034 & 0.277 & 0.279 & 0.811 & 0.001 & 0.285 & 0.285 & 0.932 \\
            250 & 0.009 & 0.238 & 0.239 & 0.935 & 0.929 & 0.027 & 0.196 & 0.198 & 0.284 & 0.020 & 0.197 & 0.198 & 0.790 & -0.006 & 0.199 & 0.199 & 0.940 \\
            500 & 0.005 & 0.151 & 0.151 & 0.947 & 0.945 & 0.023 & 0.134 & 0.136 & 0.231 & 0.015 & 0.135 & 0.136 & 0.840 & -0.003 & 0.135 & 0.135 & 0.957 \\
            1'000 & 0.003 & 0.106 & 0.106 & 0.939 & 0.937 & 0.022 & 0.096 & 0.098 & 0.233 & 0.013 & 0.098 & 0.099 & 0.852 & 0.000 & 0.098 & 0.098 & 0.947 \\
            8'000 & 0.000 & 0.037 & 0.037 & 0.945 & 0.945 & 0.022 & 0.034 & 0.040 & 0.054 & 0.005 & 0.036 & 0.036 & 0.925 & 0.002 & 0.035 & 0.035 & 0.944 \\
            \midrule
            \multicolumn{18}{c}{$h=4$, $\theta_0^{(h)}=0.0009$} \\
             & \multicolumn{5}{c}{DML} & \multicolumn{4}{c}{RA} & \multicolumn{4}{c}{DR} & \multicolumn{4}{c}{LP} \\
            T & Bias & std$(\hat{\theta}_h)$ & RMSE & $C_{b}$(95\%) & $C_{a}$(95\%) & Bias & std$(\hat{\theta}_h)$ & RMSE & $C_{b}$(95\%) & Bias & std$(\hat{\theta}_h)$ & RMSE & $C_{b}$(95\%) & Bias & std$(\hat{\theta}_h)$ & RMSE & $C_{b}$(95\%) \\
            \cmidrule(lr){1-1}\cmidrule(lr){2-6}\cmidrule(lr){7-10}\cmidrule(lr){11-14}\cmidrule(lr){15-18}
            125 & 0.023 & 0.376 & 0.376 & 0.945 & 0.933 & 0.034 & 0.290 & 0.292 & 0.315 & 0.026 & 0.286 & 0.287 & 0.777 & -0.001 & 0.296 & 0.296 & 0.925 \\
            250 & 0.006 & 0.228 & 0.228 & 0.944 & 0.938 & 0.025 & 0.190 & 0.192 & 0.316 & 0.020 & 0.191 & 0.192 & 0.810 & 0.003 & 0.195 & 0.195 & 0.940 \\
            500 & 0.006 & 0.155 & 0.155 & 0.940 & 0.936 & 0.017 & 0.141 & 0.142 & 0.252 & 0.011 & 0.142 & 0.142 & 0.815 & -0.005 & 0.145 & 0.145 & 0.938 \\
            1'000 & 0.006 & 0.103 & 0.103 & 0.956 & 0.955 & 0.020 & 0.093 & 0.095 & 0.185 & 0.012 & 0.095 & 0.096 & 0.871 & 0.001 & 0.095 & 0.095 & 0.951 \\
            8'000 & -0.000 & 0.036 & 0.036 & 0.953 & 0.953 & 0.016 & 0.033 & 0.037 & 0.065 & 0.003 & 0.035 & 0.035 & 0.936 & 0.000 & 0.035 & 0.035 & 0.956 \\
            \midrule
            \multicolumn{18}{c}{$h=5$, $\theta_0^{(h)}=0.0002$} \\
             & \multicolumn{5}{c}{DML} & \multicolumn{4}{c}{RA} & \multicolumn{4}{c}{DR} & \multicolumn{4}{c}{LP} \\
            T & Bias & std$(\hat{\theta}_h)$ & RMSE & $C_{b}$(95\%) & $C_{a}$(95\%) & Bias & std$(\hat{\theta}_h)$ & RMSE & $C_{b}$(95\%) & Bias & std$(\hat{\theta}_h)$ & RMSE & $C_{b}$(95\%) & Bias & std$(\hat{\theta}_h)$ & RMSE & $C_{b}$(95\%) \\
            \cmidrule(lr){1-1}\cmidrule(lr){2-6}\cmidrule(lr){7-10}\cmidrule(lr){11-14}\cmidrule(lr){15-18}
            125 & 0.024 & 0.370 & 0.371 & 0.944 & 0.930 & 0.022 & 0.283 & 0.284 & 0.305 & 0.017 & 0.281 & 0.282 & 0.909 & -0.001 & 0.285 & 0.285 & 0.935 \\
            250 & 0.012 & 0.237 & 0.238 & 0.936 & 0.934 & 0.016 & 0.194 & 0.194 & 0.304 & 0.013 & 0.193 & 0.193 & 0.796 & -0.004 & 0.197 & 0.197 & 0.946 \\
            500 & 0.006 & 0.152 & 0.152 & 0.960 & 0.954 & 0.017 & 0.135 & 0.136 & 0.226 & 0.012 & 0.136 & 0.136 & 0.851 & -0.003 & 0.137 & 0.137 & 0.957 \\
            1'000 & 0.006 & 0.104 & 0.105 & 0.953 & 0.950 & 0.017 & 0.096 & 0.097 & 0.171 & 0.011 & 0.097 & 0.097 & 0.869 & 0.000 & 0.096 & 0.096 & 0.953 \\
            8'000 & 0.000 & 0.036 & 0.036 & 0.950 & 0.950 & 0.013 & 0.033 & 0.036 & 0.068 & 0.003 & 0.035 & 0.035 & 0.931 & 0.000 & 0.034 & 0.034 & 0.960 \\
            \bottomrule\end{tabular}
            \begin{tablenotes}[flushleft]
            \item \textsc{Note}:
                        The table depicts simulation results across $N=1'000$ draws obtained for the scenario with
                        one stochastic process. Except for the LP estimator, nuisance functions are estimated with random forest.
                        For the DML estimator we use 10-fold cross-fitting and set $k_T=T/10$. For sample size $T=125$, probabilities are winsorized at 1\%.
                        The parameters of the data generating process are empirically calibrated using monthly U.S. data from 1982 to 2012 obtained from the empirical study in \cite{Angrist2018}.
                        $C_a$($\cdot$) and $C_b$($\cdot$) in the tables denote the coverage at the given confidence level using asymptotic and fixed-bandwidth critical values, respectively.
                        \end{tablenotes}
            \end{threeparttable}}
    \end{table}

    \begin{table}[htbp]
    \caption{Simulation results for a baseline nonlinear DGP ($n=12$, $\sigma_\epsilon=1.0$) and gradient boosting nuisance function estimates}
    \label{tab:results_boosting}
    \resizebox{\textwidth}{!}{\begin{threeparttable}\begin{tabular}{rrrrrrrrrrrrrrrrrr}
    \toprule
    \multicolumn{18}{c}{$h=0$, $\theta_0^{(h)}=0.3321$} \\
            & \multicolumn{5}{c}{DML} & \multicolumn{4}{c}{RA} & \multicolumn{4}{c}{DR} & \multicolumn{4}{c}{LP} \\
    T & Bias & std$(\hat{\theta}_h)$ & RMSE & $C_{b}$(95\%) & $C_{a}$(95\%) & Bias & std$(\hat{\theta}_h)$ & RMSE & $C_{b}$(95\%) & Bias & std$(\hat{\theta}_h)$ & RMSE & $C_{b}$(95\%) & Bias & std$(\hat{\theta}_h)$ & RMSE & $C_{b}$(95\%) \\
    \cmidrule(lr){1-1}\cmidrule(lr){2-6}\cmidrule(lr){7-10}\cmidrule(lr){11-14}\cmidrule(lr){15-18}
    125 & 0.166 & 1.285 & 1.295 & 0.937 & 0.937 & 0.071 & 0.525 & 0.530 & 0.800 & -0.019 & 0.514 & 0.515 & 0.900 & 0.090 & 0.357 & 0.368 & 0.871 \\
    250 & 0.090 & 0.631 & 0.638 & 0.922 & 0.922 & 0.111 & 0.391 & 0.406 & 0.732 & 0.081 & 0.313 & 0.324 & 0.883 & 0.128 & 0.251 & 0.282 & 0.847 \\
    500 & 0.043 & 0.269 & 0.273 & 0.961 & 0.955 & 0.280 & 0.242 & 0.370 & 0.398 & 0.158 & 0.199 & 0.254 & 0.805 & 0.135 & 0.173 & 0.220 & 0.829 \\
    1'000 & 0.028 & 0.158 & 0.160 & 0.942 & 0.939 & 0.073 & 0.189 & 0.203 & 0.662 & 0.038 & 0.141 & 0.146 & 0.919 & 0.138 & 0.126 & 0.187 & 0.730 \\
    8'000 & 0.008 & 0.040 & 0.041 & 0.960 & 0.960 & 0.115 & 0.042 & 0.122 & 0.089 & 0.014 & 0.040 & 0.043 & 0.952 & 0.131 & 0.046 & 0.139 & 0.141 \\
    \midrule
    \multicolumn{18}{c}{$h=1$, $\theta_0^{(h)}=0.1992$} \\
            & \multicolumn{5}{c}{DML} & \multicolumn{4}{c}{RA} & \multicolumn{4}{c}{DR} & \multicolumn{4}{c}{LP} \\
    T & Bias & std$(\hat{\theta}_h)$ & RMSE & $C_{b}$(95\%) & $C_{a}$(95\%) & Bias & std$(\hat{\theta}_h)$ & RMSE & $C_{b}$(95\%) & Bias & std$(\hat{\theta}_h)$ & RMSE & $C_{b}$(95\%) & Bias & std$(\hat{\theta}_h)$ & RMSE & $C_{b}$(95\%) \\
    \cmidrule(lr){1-1}\cmidrule(lr){2-6}\cmidrule(lr){7-10}\cmidrule(lr){11-14}\cmidrule(lr){15-18}
    125 & 0.166 & 1.322 & 1.332 & 0.952 & 0.938 & 0.090 & 0.547 & 0.554 & 0.751 & 0.009 & 0.514 & 0.514 & 0.898 & 0.060 & 0.362 & 0.367 & 0.923 \\
    250 & 0.105 & 0.641 & 0.650 & 0.955 & 0.954 & 0.135 & 0.398 & 0.421 & 0.701 & 0.090 & 0.315 & 0.327 & 0.877 & 0.111 & 0.265 & 0.287 & 0.914 \\
    500 & 0.067 & 0.266 & 0.275 & 0.959 & 0.957 & 0.320 & 0.236 & 0.397 & 0.286 & 0.175 & 0.196 & 0.263 & 0.777 & 0.123 & 0.193 & 0.229 & 0.887 \\
    1'000 & 0.044 & 0.154 & 0.160 & 0.953 & 0.952 & 0.097 & 0.187 & 0.211 & 0.553 & 0.052 & 0.137 & 0.146 & 0.916 & 0.139 & 0.133 & 0.193 & 0.792 \\
    8'000 & 0.016 & 0.042 & 0.045 & 0.937 & 0.936 & 0.134 & 0.044 & 0.141 & 0.013 & 0.021 & 0.042 & 0.047 & 0.928 & 0.144 & 0.053 & 0.153 & 0.183 \\
    \midrule
    \multicolumn{18}{c}{$h=2$, $\theta_0^{(h)}=0.1195$} \\
            & \multicolumn{5}{c}{DML} & \multicolumn{4}{c}{RA} & \multicolumn{4}{c}{DR} & \multicolumn{4}{c}{LP} \\
    T & Bias & std$(\hat{\theta}_h)$ & RMSE & $C_{b}$(95\%) & $C_{a}$(95\%) & Bias & std$(\hat{\theta}_h)$ & RMSE & $C_{b}$(95\%) & Bias & std$(\hat{\theta}_h)$ & RMSE & $C_{b}$(95\%) & Bias & std$(\hat{\theta}_h)$ & RMSE & $C_{b}$(95\%) \\
    \cmidrule(lr){1-1}\cmidrule(lr){2-6}\cmidrule(lr){7-10}\cmidrule(lr){11-14}\cmidrule(lr){15-18}
    125 & 0.153 & 1.439 & 1.447 & 0.959 & 0.935 & 0.104 & 0.561 & 0.570 & 0.732 & -0.006 & 0.541 & 0.541 & 0.913 & 0.031 & 0.410 & 0.411 & 0.931 \\
    250 & 0.107 & 0.723 & 0.731 & 0.944 & 0.937 & 0.149 & 0.426 & 0.451 & 0.662 & 0.082 & 0.329 & 0.339 & 0.887 & 0.078 & 0.306 & 0.316 & 0.927 \\
    500 & 0.072 & 0.300 & 0.308 & 0.952 & 0.952 & 0.310 & 0.266 & 0.408 & 0.346 & 0.171 & 0.219 & 0.278 & 0.800 & 0.104 & 0.221 & 0.244 & 0.911 \\
    1'000 & 0.051 & 0.176 & 0.183 & 0.930 & 0.926 & 0.104 & 0.210 & 0.234 & 0.478 & 0.058 & 0.156 & 0.167 & 0.898 & 0.128 & 0.160 & 0.204 & 0.854 \\
    8'000 & 0.015 & 0.049 & 0.051 & 0.943 & 0.943 & 0.131 & 0.051 & 0.141 & 0.014 & 0.019 & 0.048 & 0.052 & 0.935 & 0.128 & 0.055 & 0.140 & 0.352 \\
    \midrule
    \multicolumn{18}{c}{$h=3$, $\theta_0^{(h)}=0.0717$} \\
            & \multicolumn{5}{c}{DML} & \multicolumn{4}{c}{RA} & \multicolumn{4}{c}{DR} & \multicolumn{4}{c}{LP} \\
    T & Bias & std$(\hat{\theta}_h)$ & RMSE & $C_{b}$(95\%) & $C_{a}$(95\%) & Bias & std$(\hat{\theta}_h)$ & RMSE & $C_{b}$(95\%) & Bias & std$(\hat{\theta}_h)$ & RMSE & $C_{b}$(95\%) & Bias & std$(\hat{\theta}_h)$ & RMSE & $C_{b}$(95\%) \\
    \cmidrule(lr){1-1}\cmidrule(lr){2-6}\cmidrule(lr){7-10}\cmidrule(lr){11-14}\cmidrule(lr){15-18}
    125 & 0.132 & 1.513 & 1.519 & 0.958 & 0.944 & 0.124 & 0.622 & 0.634 & 0.694 & 0.007 & 0.591 & 0.591 & 0.894 & 0.012 & 0.448 & 0.448 & 0.937 \\
    250 & 0.097 & 0.797 & 0.803 & 0.933 & 0.930 & 0.143 & 0.476 & 0.497 & 0.611 & 0.074 & 0.364 & 0.371 & 0.901 & 0.058 & 0.340 & 0.345 & 0.934 \\
    500 & 0.078 & 0.327 & 0.337 & 0.953 & 0.946 & 0.289 & 0.287 & 0.407 & 0.391 & 0.163 & 0.238 & 0.289 & 0.840 & 0.087 & 0.249 & 0.264 & 0.915 \\
    1'000 & 0.054 & 0.196 & 0.204 & 0.949 & 0.946 & 0.113 & 0.228 & 0.255 & 0.442 & 0.058 & 0.168 & 0.178 & 0.915 & 0.112 & 0.178 & 0.210 & 0.902 \\
    8'000 & 0.014 & 0.055 & 0.057 & 0.946 & 0.946 & 0.128 & 0.058 & 0.141 & 0.021 & 0.017 & 0.055 & 0.057 & 0.937 & 0.117 & 0.061 & 0.131 & 0.479 \\
    \midrule
    \multicolumn{18}{c}{$h=4$, $\theta_0^{(h)}=0.0430$} \\
            & \multicolumn{5}{c}{DML} & \multicolumn{4}{c}{RA} & \multicolumn{4}{c}{DR} & \multicolumn{4}{c}{LP} \\
    T & Bias & std$(\hat{\theta}_h)$ & RMSE & $C_{b}$(95\%) & $C_{a}$(95\%) & Bias & std$(\hat{\theta}_h)$ & RMSE & $C_{b}$(95\%) & Bias & std$(\hat{\theta}_h)$ & RMSE & $C_{b}$(95\%) & Bias & std$(\hat{\theta}_h)$ & RMSE & $C_{b}$(95\%) \\
    \cmidrule(lr){1-1}\cmidrule(lr){2-6}\cmidrule(lr){7-10}\cmidrule(lr){11-14}\cmidrule(lr){15-18}
    125 & 0.150 & 1.681 & 1.688 & 0.960 & 0.937 & 0.125 & 0.678 & 0.689 & 0.695 & 0.041 & 0.696 & 0.697 & 0.902 & 0.021 & 0.496 & 0.497 & 0.935 \\
    250 & 0.069 & 0.867 & 0.870 & 0.927 & 0.924 & 0.142 & 0.510 & 0.529 & 0.591 & 0.077 & 0.405 & 0.412 & 0.873 & 0.049 & 0.377 & 0.380 & 0.922 \\
    500 & 0.068 & 0.363 & 0.369 & 0.952 & 0.940 & 0.245 & 0.299 & 0.387 & 0.447 & 0.148 & 0.261 & 0.300 & 0.852 & 0.062 & 0.270 & 0.277 & 0.929 \\
    1'000 & 0.051 & 0.219 & 0.225 & 0.942 & 0.942 & 0.110 & 0.247 & 0.271 & 0.408 & 0.054 & 0.181 & 0.189 & 0.934 & 0.092 & 0.192 & 0.213 & 0.921 \\
    8'000 & 0.014 & 0.061 & 0.063 & 0.948 & 0.948 & 0.126 & 0.063 & 0.141 & 0.026 & 0.017 & 0.060 & 0.063 & 0.939 & 0.100 & 0.067 & 0.121 & 0.662 \\
    \midrule
    \multicolumn{18}{c}{$h=5$, $\theta_0^{(h)}=0.0258$} \\
            & \multicolumn{5}{c}{DML} & \multicolumn{4}{c}{RA} & \multicolumn{4}{c}{DR} & \multicolumn{4}{c}{LP} \\
    T & Bias & std$(\hat{\theta}_h)$ & RMSE & $C_{b}$(95\%) & $C_{a}$(95\%) & Bias & std$(\hat{\theta}_h)$ & RMSE & $C_{b}$(95\%) & Bias & std$(\hat{\theta}_h)$ & RMSE & $C_{b}$(95\%) & Bias & std$(\hat{\theta}_h)$ & RMSE & $C_{b}$(95\%) \\
    \cmidrule(lr){1-1}\cmidrule(lr){2-6}\cmidrule(lr){7-10}\cmidrule(lr){11-14}\cmidrule(lr){15-18}
    125 & 0.099 & 1.794 & 1.796 & 0.949 & 0.927 & 0.148 & 0.705 & 0.720 & 0.669 & 0.082 & 0.729 & 0.734 & 0.898 & 0.032 & 0.551 & 0.552 & 0.921 \\
    250 & 0.077 & 0.937 & 0.940 & 0.931 & 0.926 & 0.119 & 0.532 & 0.545 & 0.584 & 0.068 & 0.421 & 0.427 & 0.884 & 0.036 & 0.394 & 0.395 & 0.918 \\
    500 & 0.063 & 0.384 & 0.389 & 0.951 & 0.943 & 0.216 & 0.308 & 0.376 & 0.491 & 0.131 & 0.271 & 0.301 & 0.862 & 0.047 & 0.283 & 0.287 & 0.943 \\
    1'000 & 0.039 & 0.233 & 0.236 & 0.946 & 0.938 & 0.097 & 0.256 & 0.274 & 0.445 & 0.047 & 0.193 & 0.198 & 0.927 & 0.071 & 0.202 & 0.214 & 0.933 \\
    8'000 & 0.010 & 0.068 & 0.068 & 0.951 & 0.951 & 0.116 & 0.071 & 0.136 & 0.040 & 0.013 & 0.067 & 0.068 & 0.941 & 0.082 & 0.078 & 0.113 & 0.789 \\
    \bottomrule\end{tabular}
    \begin{tablenotes}[flushleft]
    \item \textsc{Note}: The table depicts simulation results across $N=1'000$ draws obtained for the scenario with
            one stochastic process. Except for the LP estimator, nuisance functions are estimated with gradient boosting.
            For the DML estimator we use 10-fold cross-fitting and set $k_T=T/10$. For sample size $T=125$, probabilities are winsorized at 1\%.
            The parameters of the data generating process are $n=12$,
            $\sigma_\epsilon=1.0$, $\gamma=0.6$, $p=2$,
            $q=1$, $\sigma_u=1.0$, $\alpha_A=0.3$,
            $\alpha_M=0.3$, $\rho_A=0.35$, $\rho_M=0.7$,
            $\beta_1=0.3$, $\beta_2=0.5$.
            $C_a$($\cdot$) and $C_b$($\cdot$) in the tables denote the coverage at the given confidence level using asymptotic and fixed-bandwidth critical values, respectively.
    \end{tablenotes}
    \end{threeparttable}}
    \end{table}

    \begin{table}[htbp]
    \caption{Simulation results for a linear DGP ($n=12$, $\sigma_\epsilon=1.0$) and random forest nuisance function estimates}
    \label{tab:results_linear_constant}
    \resizebox{\textwidth}{!}{\begin{threeparttable}\begin{tabular}{rrrrrrrrrrrrrrrrrr}
    \toprule
    \multicolumn{18}{c}{$h=0$, $\theta_0^{(h)}=0.3321$} \\
            & \multicolumn{5}{c}{DML} & \multicolumn{4}{c}{RA} & \multicolumn{4}{c}{DR} & \multicolumn{4}{c}{LP} \\
    T & Bias & std$(\hat{\theta}_h)$ & RMSE & $C_{b}$(95\%) & $C_{a}$(95\%) & Bias & std$(\hat{\theta}_h)$ & RMSE & $C_{b}$(95\%) & Bias & std$(\hat{\theta}_h)$ & RMSE & $C_{b}$(95\%) & Bias & std$(\hat{\theta}_h)$ & RMSE & $C_{b}$(95\%) \\
    \cmidrule(lr){1-1}\cmidrule(lr){2-6}\cmidrule(lr){7-10}\cmidrule(lr){11-14}\cmidrule(lr){15-18}
    125 & 0.121 & 1.375 & 1.381 & 0.927 & 0.918 & 0.851 & 0.550 & 1.013 & 0.270 & 0.669 & 0.498 & 0.834 & 0.503 & -0.006 & 0.219 & 0.219 & 0.911 \\
    250 & 0.060 & 1.043 & 1.044 & 0.944 & 0.935 & 0.675 & 0.325 & 0.749 & 0.095 & 0.525 & 0.292 & 0.601 & 0.334 & 0.004 & 0.146 & 0.146 & 0.956 \\
    500 & 0.069 & 0.348 & 0.355 & 0.959 & 0.941 & 0.537 & 0.203 & 0.574 & 0.031 & 0.388 & 0.185 & 0.430 & 0.240 & -0.001 & 0.106 & 0.106 & 0.947 \\
    1'000 & 0.062 & 0.164 & 0.176 & 0.944 & 0.942 & 0.445 & 0.126 & 0.463 & 0.000 & 0.280 & 0.116 & 0.303 & 0.186 & -0.000 & 0.074 & 0.074 & 0.956 \\
    8'000 & 0.019 & 0.040 & 0.044 & 0.942 & 0.942 & 0.381 & 0.066 & 0.386 & 0.000 & 0.108 & 0.038 & 0.114 & 0.152 & -0.001 & 0.025 & 0.025 & 0.940 \\
    \midrule
    \multicolumn{18}{c}{$h=1$, $\theta_0^{(h)}=0.1992$} \\
            & \multicolumn{5}{c}{DML} & \multicolumn{4}{c}{RA} & \multicolumn{4}{c}{DR} & \multicolumn{4}{c}{LP} \\
    T & Bias & std$(\hat{\theta}_h)$ & RMSE & $C_{b}$(95\%) & $C_{a}$(95\%) & Bias & std$(\hat{\theta}_h)$ & RMSE & $C_{b}$(95\%) & Bias & std$(\hat{\theta}_h)$ & RMSE & $C_{b}$(95\%) & Bias & std$(\hat{\theta}_h)$ & RMSE & $C_{b}$(95\%) \\
    \cmidrule(lr){1-1}\cmidrule(lr){2-6}\cmidrule(lr){7-10}\cmidrule(lr){11-14}\cmidrule(lr){15-18}
    125 & 0.173 & 1.539 & 1.549 & 0.916 & 0.900 & 0.893 & 0.590 & 1.070 & 0.287 & 0.708 & 0.542 & 0.892 & 0.485 & -0.030 & 0.317 & 0.319 & 0.918 \\
    250 & 0.067 & 1.273 & 1.275 & 0.943 & 0.921 & 0.732 & 0.364 & 0.817 & 0.119 & 0.567 & 0.332 & 0.657 & 0.354 & -0.005 & 0.224 & 0.224 & 0.934 \\
    500 & 0.105 & 0.372 & 0.386 & 0.945 & 0.936 & 0.597 & 0.223 & 0.638 & 0.023 & 0.433 & 0.207 & 0.480 & 0.262 & -0.008 & 0.151 & 0.151 & 0.947 \\
    1'000 & 0.094 & 0.195 & 0.216 & 0.913 & 0.910 & 0.509 & 0.148 & 0.530 & 0.002 & 0.324 & 0.139 & 0.352 & 0.200 & 0.002 & 0.109 & 0.110 & 0.945 \\
    8'000 & 0.028 & 0.050 & 0.057 & 0.922 & 0.922 & 0.447 & 0.073 & 0.453 & 0.000 & 0.128 & 0.046 & 0.136 & 0.163 & 0.001 & 0.038 & 0.038 & 0.950 \\
    \midrule
    \multicolumn{18}{c}{$h=2$, $\theta_0^{(h)}=0.1195$} \\
            & \multicolumn{5}{c}{DML} & \multicolumn{4}{c}{RA} & \multicolumn{4}{c}{DR} & \multicolumn{4}{c}{LP} \\
    T & Bias & std$(\hat{\theta}_h)$ & RMSE & $C_{b}$(95\%) & $C_{a}$(95\%) & Bias & std$(\hat{\theta}_h)$ & RMSE & $C_{b}$(95\%) & Bias & std$(\hat{\theta}_h)$ & RMSE & $C_{b}$(95\%) & Bias & std$(\hat{\theta}_h)$ & RMSE & $C_{b}$(95\%) \\
    \cmidrule(lr){1-1}\cmidrule(lr){2-6}\cmidrule(lr){7-10}\cmidrule(lr){11-14}\cmidrule(lr){15-18}
    125 & 0.217 & 1.495 & 1.511 & 0.917 & 0.892 & 0.818 & 0.659 & 1.050 & 0.372 & 0.649 & 0.615 & 0.894 & 0.610 & -0.042 & 0.454 & 0.456 & 0.918 \\
    250 & 0.082 & 1.058 & 1.061 & 0.917 & 0.905 & 0.656 & 0.418 & 0.778 & 0.238 & 0.510 & 0.393 & 0.644 & 0.546 & -0.014 & 0.326 & 0.326 & 0.931 \\
    500 & 0.114 & 0.412 & 0.427 & 0.944 & 0.932 & 0.547 & 0.270 & 0.610 & 0.102 & 0.399 & 0.257 & 0.475 & 0.483 & -0.009 & 0.218 & 0.218 & 0.954 \\
    1'000 & 0.102 & 0.237 & 0.258 & 0.898 & 0.895 & 0.470 & 0.192 & 0.508 & 0.033 & 0.303 & 0.185 & 0.355 & 0.429 & 0.004 & 0.165 & 0.165 & 0.936 \\
    8'000 & 0.027 & 0.066 & 0.072 & 0.926 & 0.924 & 0.420 & 0.084 & 0.428 & 0.000 & 0.119 & 0.061 & 0.134 & 0.457 & 0.000 & 0.055 & 0.055 & 0.957 \\
    \midrule
    \multicolumn{18}{c}{$h=3$, $\theta_0^{(h)}=0.0717$} \\
            & \multicolumn{5}{c}{DML} & \multicolumn{4}{c}{RA} & \multicolumn{4}{c}{DR} & \multicolumn{4}{c}{LP} \\
    T & Bias & std$(\hat{\theta}_h)$ & RMSE & $C_{b}$(95\%) & $C_{a}$(95\%) & Bias & std$(\hat{\theta}_h)$ & RMSE & $C_{b}$(95\%) & Bias & std$(\hat{\theta}_h)$ & RMSE & $C_{b}$(95\%) & Bias & std$(\hat{\theta}_h)$ & RMSE & $C_{b}$(95\%) \\
    \cmidrule(lr){1-1}\cmidrule(lr){2-6}\cmidrule(lr){7-10}\cmidrule(lr){11-14}\cmidrule(lr){15-18}
    125 & 0.221 & 1.641 & 1.656 & 0.920 & 0.895 & 0.696 & 0.709 & 0.993 & 0.431 & 0.554 & 0.669 & 0.869 & 0.724 & -0.051 & 0.562 & 0.564 & 0.917 \\
    250 & 0.055 & 1.522 & 1.523 & 0.912 & 0.899 & 0.569 & 0.488 & 0.750 & 0.346 & 0.438 & 0.463 & 0.637 & 0.660 & -0.019 & 0.418 & 0.418 & 0.913 \\
    500 & 0.113 & 0.462 & 0.475 & 0.951 & 0.937 & 0.478 & 0.310 & 0.570 & 0.220 & 0.352 & 0.301 & 0.463 & 0.642 & -0.006 & 0.281 & 0.281 & 0.945 \\
    1'000 & 0.106 & 0.275 & 0.295 & 0.913 & 0.902 & 0.413 & 0.227 & 0.471 & 0.100 & 0.270 & 0.223 & 0.350 & 0.590 & 0.007 & 0.209 & 0.209 & 0.935 \\
    8'000 & 0.027 & 0.082 & 0.086 & 0.938 & 0.937 & 0.376 & 0.092 & 0.387 & 0.000 & 0.107 & 0.077 & 0.132 & 0.638 & 0.000 & 0.071 & 0.071 & 0.952 \\
    \midrule
    \multicolumn{18}{c}{$h=4$, $\theta_0^{(h)}=0.0430$} \\
            & \multicolumn{5}{c}{DML} & \multicolumn{4}{c}{RA} & \multicolumn{4}{c}{DR} & \multicolumn{4}{c}{LP} \\
    T & Bias & std$(\hat{\theta}_h)$ & RMSE & $C_{b}$(95\%) & $C_{a}$(95\%) & Bias & std$(\hat{\theta}_h)$ & RMSE & $C_{b}$(95\%) & Bias & std$(\hat{\theta}_h)$ & RMSE & $C_{b}$(95\%) & Bias & std$(\hat{\theta}_h)$ & RMSE & $C_{b}$(95\%) \\
    \cmidrule(lr){1-1}\cmidrule(lr){2-6}\cmidrule(lr){7-10}\cmidrule(lr){11-14}\cmidrule(lr){15-18}
    125 & 0.283 & 1.587 & 1.612 & 0.892 & 0.892 & 0.594 & 0.744 & 0.952 & 0.530 & 0.478 & 0.714 & 0.860 & 0.807 & -0.023 & 0.660 & 0.661 & 0.915 \\
    250 & 0.092 & 1.054 & 1.058 & 0.907 & 0.897 & 0.476 & 0.545 & 0.724 & 0.409 & 0.369 & 0.524 & 0.640 & 0.715 & -0.016 & 0.493 & 0.493 & 0.917 \\
    500 & 0.103 & 0.489 & 0.499 & 0.961 & 0.935 & 0.396 & 0.345 & 0.526 & 0.335 & 0.291 & 0.337 & 0.446 & 0.733 & -0.016 & 0.326 & 0.326 & 0.957 \\
    1'000 & 0.102 & 0.318 & 0.334 & 0.914 & 0.909 & 0.357 & 0.261 & 0.442 & 0.194 & 0.237 & 0.257 & 0.350 & 0.693 & 0.010 & 0.247 & 0.247 & 0.937 \\
    8'000 & 0.023 & 0.097 & 0.100 & 0.942 & 0.941 & 0.329 & 0.100 & 0.343 & 0.002 & 0.092 & 0.091 & 0.130 & 0.756 & -0.001 & 0.084 & 0.084 & 0.948 \\
    \midrule
    \multicolumn{18}{c}{$h=5$, $\theta_0^{(h)}=0.0258$} \\
            & \multicolumn{5}{c}{DML} & \multicolumn{4}{c}{RA} & \multicolumn{4}{c}{DR} & \multicolumn{4}{c}{LP} \\
    T & Bias & std$(\hat{\theta}_h)$ & RMSE & $C_{b}$(95\%) & $C_{a}$(95\%) & Bias & std$(\hat{\theta}_h)$ & RMSE & $C_{b}$(95\%) & Bias & std$(\hat{\theta}_h)$ & RMSE & $C_{b}$(95\%) & Bias & std$(\hat{\theta}_h)$ & RMSE & $C_{b}$(95\%) \\
    \cmidrule(lr){1-1}\cmidrule(lr){2-6}\cmidrule(lr){7-10}\cmidrule(lr){11-14}\cmidrule(lr){15-18}
    125 & 0.280 & 1.696 & 1.719 & 0.902 & 0.886 & 0.500 & 0.809 & 0.951 & 0.560 & 0.411 & 0.777 & 0.879 & 0.794 & 0.001 & 0.738 & 0.738 & 0.919 \\
    250 & 0.091 & 1.178 & 1.182 & 0.923 & 0.916 & 0.407 & 0.579 & 0.708 & 0.489 & 0.313 & 0.558 & 0.640 & 0.765 & -0.009 & 0.532 & 0.532 & 0.927 \\
    500 & 0.091 & 0.531 & 0.539 & 0.958 & 0.938 & 0.327 & 0.377 & 0.499 & 0.402 & 0.241 & 0.370 & 0.442 & 0.779 & -0.020 & 0.364 & 0.364 & 0.953 \\
    1'000 & 0.094 & 0.348 & 0.360 & 0.913 & 0.912 & 0.303 & 0.283 & 0.414 & 0.270 & 0.202 & 0.280 & 0.346 & 0.757 & 0.008 & 0.272 & 0.272 & 0.946 \\
    8'000 & 0.019 & 0.109 & 0.111 & 0.934 & 0.934 & 0.282 & 0.107 & 0.302 & 0.009 & 0.077 & 0.102 & 0.128 & 0.821 & -0.003 & 0.096 & 0.096 & 0.946 \\
    \bottomrule\end{tabular}
    \begin{tablenotes}[flushleft]
    \item \textsc{Note}: The table depicts simulation results across $N=1'000$ draws obtained for the scenario with
            one stochastic process. The outcome variable is generated from a linear DGP, i.e. $b(X_t) = 0.5 \sum_{i=1}^5X_{i,t}$ and $\tau(X_t) = \theta_0^{(0)}$. Except for the LP estimator, nuisance functions are estimated with random forest.
            For the DML estimator we use 10-fold cross-fitting and set $k_T=T/10$. For sample size $T=125$, probabilities are winsorized at 1\%.
            The parameters of the data generating process are $n=12$,
            $\sigma_\epsilon=1.0$, $\gamma=0.6$, $p=2$,
            $q=1$, $\sigma_u=1.0$, $\alpha_A=0.3$,
            $\alpha_M=0.3$, $\rho_A=0.35$, $\rho_M=0.7$,
            $\beta_1=0.3$, $\beta_2=0.5$.
            $C_a$($\cdot$) and $C_b$($\cdot$) in the tables denote the coverage at the given confidence level using asymptotic and fixed-bandwidth critical values, respectively.
    \end{tablenotes}
    \end{threeparttable}}
    \end{table}

    \begin{table}[htbp]
    \caption{Simulation results for a linear DGP with interactions ($n=12$, $\sigma_\epsilon=1.0$) and random forest nuisance function estimates}
    \label{tab:results_linear_interaction}
    \resizebox{\textwidth}{!}{\begin{threeparttable}\begin{tabular}{rrrrrrrrrrrrrrrrrr}
            \toprule
            \multicolumn{18}{c}{$h=0$, $\theta_0^{(h)}=0.3321$} \\
                    & \multicolumn{5}{c}{DML} & \multicolumn{4}{c}{RA} & \multicolumn{4}{c}{DR} & \multicolumn{4}{c}{LP} \\
            T & Bias & std$(\hat{\theta}_h)$ & RMSE & $C_{b}$(95\%) & $C_{a}$(95\%) & Bias & std$(\hat{\theta}_h)$ & RMSE & $C_{b}$(95\%) & Bias & std$(\hat{\theta}_h)$ & RMSE & $C_{b}$(95\%) & Bias & std$(\hat{\theta}_h)$ & RMSE & $C_{b}$(95\%) \\
            \cmidrule(lr){1-1}\cmidrule(lr){2-6}\cmidrule(lr){7-10}\cmidrule(lr){11-14}\cmidrule(lr){15-18}
            125 & 0.115 & 1.926 & 1.930 & 0.810 & 0.786 & 0.900 & 1.076 & 1.403 & 0.404 & 0.714 & 1.020 & 1.245 & 0.513 & 0.556 & 0.858 & 1.023 & 0.579 \\
            250 & 0.079 & 1.248 & 1.251 & 0.814 & 0.810 & 0.700 & 0.726 & 1.009 & 0.468 & 0.551 & 0.702 & 0.892 & 0.580 & 0.595 & 0.589 & 0.837 & 0.527 \\
            500 & 0.085 & 0.618 & 0.624 & 0.839 & 0.830 & 0.571 & 0.511 & 0.767 & 0.473 & 0.419 & 0.495 & 0.649 & 0.633 & 0.624 & 0.404 & 0.743 & 0.364 \\
            1'000 & 0.084 & 0.395 & 0.403 & 0.856 & 0.849 & 0.479 & 0.367 & 0.603 & 0.487 & 0.312 & 0.359 & 0.476 & 0.684 & 0.632 & 0.297 & 0.698 & 0.206 \\
            8'000 & 0.025 & 0.066 & 0.070 & 0.947 & 0.945 & 0.416 & 0.086 & 0.425 & 0.000 & 0.124 & 0.064 & 0.140 & 0.460 & 0.329 & 0.053 & 0.333 & 0.000 \\
            \midrule
            \multicolumn{18}{c}{$h=1$, $\theta_0^{(h)}=0.1992$} \\
                    & \multicolumn{5}{c}{DML} & \multicolumn{4}{c}{RA} & \multicolumn{4}{c}{DR} & \multicolumn{4}{c}{LP} \\
            T & Bias & std$(\hat{\theta}_h)$ & RMSE & $C_{b}$(95\%) & $C_{a}$(95\%) & Bias & std$(\hat{\theta}_h)$ & RMSE & $C_{b}$(95\%) & Bias & std$(\hat{\theta}_h)$ & RMSE & $C_{b}$(95\%) & Bias & std$(\hat{\theta}_h)$ & RMSE & $C_{b}$(95\%) \\
            \cmidrule(lr){1-1}\cmidrule(lr){2-6}\cmidrule(lr){7-10}\cmidrule(lr){11-14}\cmidrule(lr){15-18}
            125 & 0.173 & 1.778 & 1.787 & 0.863 & 0.843 & 0.821 & 0.825 & 1.164 & 0.350 & 0.627 & 0.759 & 0.985 & 0.543 & 0.308 & 0.592 & 0.667 & 0.771 \\
            250 & 0.113 & 1.003 & 1.009 & 0.881 & 0.878 & 0.662 & 0.562 & 0.868 & 0.330 & 0.498 & 0.525 & 0.723 & 0.519 & 0.375 & 0.425 & 0.567 & 0.705 \\
            500 & 0.112 & 0.502 & 0.514 & 0.859 & 0.859 & 0.542 & 0.390 & 0.668 & 0.300 & 0.381 & 0.368 & 0.530 & 0.539 & 0.413 & 0.303 & 0.512 & 0.567 \\
            1'000 & 0.103 & 0.300 & 0.318 & 0.868 & 0.863 & 0.464 & 0.276 & 0.540 & 0.267 & 0.289 & 0.264 & 0.392 & 0.571 & 0.433 & 0.222 & 0.487 & 0.341 \\
            8'000 & 0.029 & 0.062 & 0.069 & 0.919 & 0.919 & 0.449 & 0.082 & 0.456 & 0.000 & 0.127 & 0.058 & 0.139 & 0.355 & 0.231 & 0.051 & 0.237 & 0.006 \\
            \midrule
            \multicolumn{18}{c}{$h=2$, $\theta_0^{(h)}=0.1195$} \\
                    & \multicolumn{5}{c}{DML} & \multicolumn{4}{c}{RA} & \multicolumn{4}{c}{DR} & \multicolumn{4}{c}{LP} \\
            T & Bias & std$(\hat{\theta}_h)$ & RMSE & $C_{b}$(95\%) & $C_{a}$(95\%) & Bias & std$(\hat{\theta}_h)$ & RMSE & $C_{b}$(95\%) & Bias & std$(\hat{\theta}_h)$ & RMSE & $C_{b}$(95\%) & Bias & std$(\hat{\theta}_h)$ & RMSE & $C_{b}$(95\%) \\
            \cmidrule(lr){1-1}\cmidrule(lr){2-6}\cmidrule(lr){7-10}\cmidrule(lr){11-14}\cmidrule(lr){15-18}
            125 & 0.202 & 1.573 & 1.586 & 0.901 & 0.894 & 0.707 & 0.730 & 1.016 & 0.383 & 0.532 & 0.668 & 0.854 & 0.627 & 0.151 & 0.564 & 0.584 & 0.905 \\
            250 & 0.042 & 2.297 & 2.297 & 0.920 & 0.912 & 0.555 & 0.497 & 0.745 & 0.312 & 0.412 & 0.465 & 0.621 & 0.607 & 0.225 & 0.415 & 0.472 & 0.895 \\
            500 & 0.111 & 0.467 & 0.480 & 0.917 & 0.909 & 0.461 & 0.341 & 0.574 & 0.280 & 0.320 & 0.323 & 0.455 & 0.595 & 0.266 & 0.301 & 0.402 & 0.818 \\
            1'000 & 0.106 & 0.281 & 0.301 & 0.887 & 0.880 & 0.406 & 0.248 & 0.476 & 0.198 & 0.254 & 0.237 & 0.347 & 0.566 & 0.302 & 0.220 & 0.374 & 0.645 \\
            8'000 & 0.025 & 0.070 & 0.074 & 0.935 & 0.934 & 0.411 & 0.086 & 0.419 & 0.000 & 0.111 & 0.065 & 0.129 & 0.540 & 0.156 & 0.061 & 0.167 & 0.301 \\
            \midrule
            \multicolumn{18}{c}{$h=3$, $\theta_0^{(h)}=0.0717$} \\
                    & \multicolumn{5}{c}{DML} & \multicolumn{4}{c}{RA} & \multicolumn{4}{c}{DR} & \multicolumn{4}{c}{LP} \\
            T & Bias & std$(\hat{\theta}_h)$ & RMSE & $C_{b}$(95\%) & $C_{a}$(95\%) & Bias & std$(\hat{\theta}_h)$ & RMSE & $C_{b}$(95\%) & Bias & std$(\hat{\theta}_h)$ & RMSE & $C_{b}$(95\%) & Bias & std$(\hat{\theta}_h)$ & RMSE & $C_{b}$(95\%) \\
            \cmidrule(lr){1-1}\cmidrule(lr){2-6}\cmidrule(lr){7-10}\cmidrule(lr){11-14}\cmidrule(lr){15-18}
            125 & 0.208 & 1.625 & 1.638 & 0.915 & 0.901 & 0.594 & 0.720 & 0.933 & 0.454 & 0.453 & 0.669 & 0.808 & 0.709 & 0.072 & 0.612 & 0.617 & 0.929 \\
            250 & 0.105 & 0.913 & 0.919 & 0.913 & 0.912 & 0.462 & 0.498 & 0.679 & 0.394 & 0.339 & 0.470 & 0.579 & 0.670 & 0.128 & 0.459 & 0.476 & 0.911 \\
            500 & 0.102 & 0.470 & 0.481 & 0.940 & 0.927 & 0.388 & 0.334 & 0.512 & 0.296 & 0.269 & 0.322 & 0.420 & 0.686 & 0.171 & 0.331 & 0.372 & 0.903 \\
            1'000 & 0.102 & 0.285 & 0.303 & 0.911 & 0.908 & 0.350 & 0.243 & 0.426 & 0.191 & 0.222 & 0.236 & 0.324 & 0.658 & 0.214 & 0.237 & 0.319 & 0.827 \\
            8'000 & 0.023 & 0.082 & 0.085 & 0.934 & 0.933 & 0.363 & 0.092 & 0.375 & 0.001 & 0.097 & 0.078 & 0.124 & 0.706 & 0.104 & 0.074 & 0.128 & 0.709 \\
            \midrule
            \multicolumn{18}{c}{$h=4$, $\theta_0^{(h)}=0.0430$} \\
                    & \multicolumn{5}{c}{DML} & \multicolumn{4}{c}{RA} & \multicolumn{4}{c}{DR} & \multicolumn{4}{c}{LP} \\
            T & Bias & std$(\hat{\theta}_h)$ & RMSE & $C_{b}$(95\%) & $C_{a}$(95\%) & Bias & std$(\hat{\theta}_h)$ & RMSE & $C_{b}$(95\%) & Bias & std$(\hat{\theta}_h)$ & RMSE & $C_{b}$(95\%) & Bias & std$(\hat{\theta}_h)$ & RMSE & $C_{b}$(95\%) \\
            \cmidrule(lr){1-1}\cmidrule(lr){2-6}\cmidrule(lr){7-10}\cmidrule(lr){11-14}\cmidrule(lr){15-18}
            125 & 0.242 & 1.594 & 1.612 & 0.919 & 0.899 & 0.504 & 0.742 & 0.897 & 0.486 & 0.389 & 0.701 & 0.802 & 0.722 & 0.045 & 0.696 & 0.697 & 0.924 \\
            250 & 0.108 & 0.967 & 0.973 & 0.905 & 0.900 & 0.399 & 0.529 & 0.662 & 0.418 & 0.300 & 0.504 & 0.587 & 0.712 & 0.090 & 0.512 & 0.520 & 0.920 \\
            500 & 0.086 & 0.552 & 0.559 & 0.938 & 0.921 & 0.319 & 0.353 & 0.476 & 0.365 & 0.222 & 0.343 & 0.409 & 0.755 & 0.104 & 0.363 & 0.378 & 0.934 \\
            1'000 & 0.090 & 0.299 & 0.313 & 0.938 & 0.935 & 0.294 & 0.251 & 0.386 & 0.251 & 0.187 & 0.244 & 0.308 & 0.751 & 0.146 & 0.255 & 0.294 & 0.911 \\
            8'000 & 0.021 & 0.094 & 0.096 & 0.938 & 0.938 & 0.318 & 0.097 & 0.333 & 0.004 & 0.085 & 0.089 & 0.123 & 0.786 & 0.071 & 0.084 & 0.110 & 0.866 \\
            \midrule
            \multicolumn{18}{c}{$h=5$, $\theta_0^{(h)}=0.0258$} \\
                    & \multicolumn{5}{c}{DML} & \multicolumn{4}{c}{RA} & \multicolumn{4}{c}{DR} & \multicolumn{4}{c}{LP} \\
            T & Bias & std$(\hat{\theta}_h)$ & RMSE & $C_{b}$(95\%) & $C_{a}$(95\%) & Bias & std$(\hat{\theta}_h)$ & RMSE & $C_{b}$(95\%) & Bias & std$(\hat{\theta}_h)$ & RMSE & $C_{b}$(95\%) & Bias & std$(\hat{\theta}_h)$ & RMSE & $C_{b}$(95\%) \\
            \cmidrule(lr){1-1}\cmidrule(lr){2-6}\cmidrule(lr){7-10}\cmidrule(lr){11-14}\cmidrule(lr){15-18}
            125 & 0.242 & 1.669 & 1.686 & 0.899 & 0.899 & 0.428 & 0.786 & 0.894 & 0.516 & 0.342 & 0.745 & 0.820 & 0.757 & 0.043 & 0.770 & 0.771 & 0.912 \\
            250 & 0.106 & 0.977 & 0.983 & 0.930 & 0.920 & 0.328 & 0.552 & 0.642 & 0.482 & 0.244 & 0.528 & 0.582 & 0.749 & 0.047 & 0.539 & 0.541 & 0.934 \\
            500 & 0.084 & 0.533 & 0.540 & 0.943 & 0.923 & 0.269 & 0.370 & 0.457 & 0.403 & 0.190 & 0.362 & 0.409 & 0.775 & 0.067 & 0.386 & 0.392 & 0.941 \\
            1'000 & 0.084 & 0.317 & 0.328 & 0.942 & 0.937 & 0.251 & 0.265 & 0.365 & 0.303 & 0.162 & 0.260 & 0.306 & 0.797 & 0.102 & 0.274 & 0.292 & 0.939 \\
            8'000 & 0.016 & 0.104 & 0.105 & 0.942 & 0.942 & 0.273 & 0.103 & 0.292 & 0.009 & 0.071 & 0.099 & 0.121 & 0.829 & 0.045 & 0.095 & 0.106 & 0.924 \\
            \bottomrule\end{tabular}
    \begin{tablenotes}[flushleft]
    \item \textsc{Note}: The table depicts simulation results across $N=1'000$ draws obtained for the scenario with
            one stochastic process. The outcome variable is generated from a linear DGP, i.e. $b(X_t) = 0.5 \sum_{i=1}^5X_{i,t}$ and $\tau(X_t) = \theta_0^{(0)} + \sum_{i=1}^3X_{i,t} - \sum_{i=4}^5X_{i,t}$. Except for the LP estimator, nuisance functions are estimated with random forest.
            For the DML estimator we use 10-fold cross-fitting and set $k_T=T/10$. For sample size $T=125$, probabilities are winsorized at 1\%.
            The parameters of the data generating process are $n=12$,
            $\sigma_\epsilon=1.0$, $\gamma=0.6$, $p=2$,
            $q=1$, $\sigma_u=1.0$, $\alpha_A=0.3$,
            $\alpha_M=0.3$, $\rho_A=0.35$, $\rho_M=0.7$,
            $\beta_1=0.3$, $\beta_2=0.5$.
            $C_a$($\cdot$) and $C_b$($\cdot$) in the tables denote the coverage at the given confidence level using asymptotic and fixed-bandwidth critical values, respectively.
    \end{tablenotes}
    \end{threeparttable}}
    \end{table}
    
    \begin{figure}[htbp]
        \caption{Estimated cumulative effects of target rate changes on the bond yield curve}
        \centering
        \includegraphics[scale=0.30]{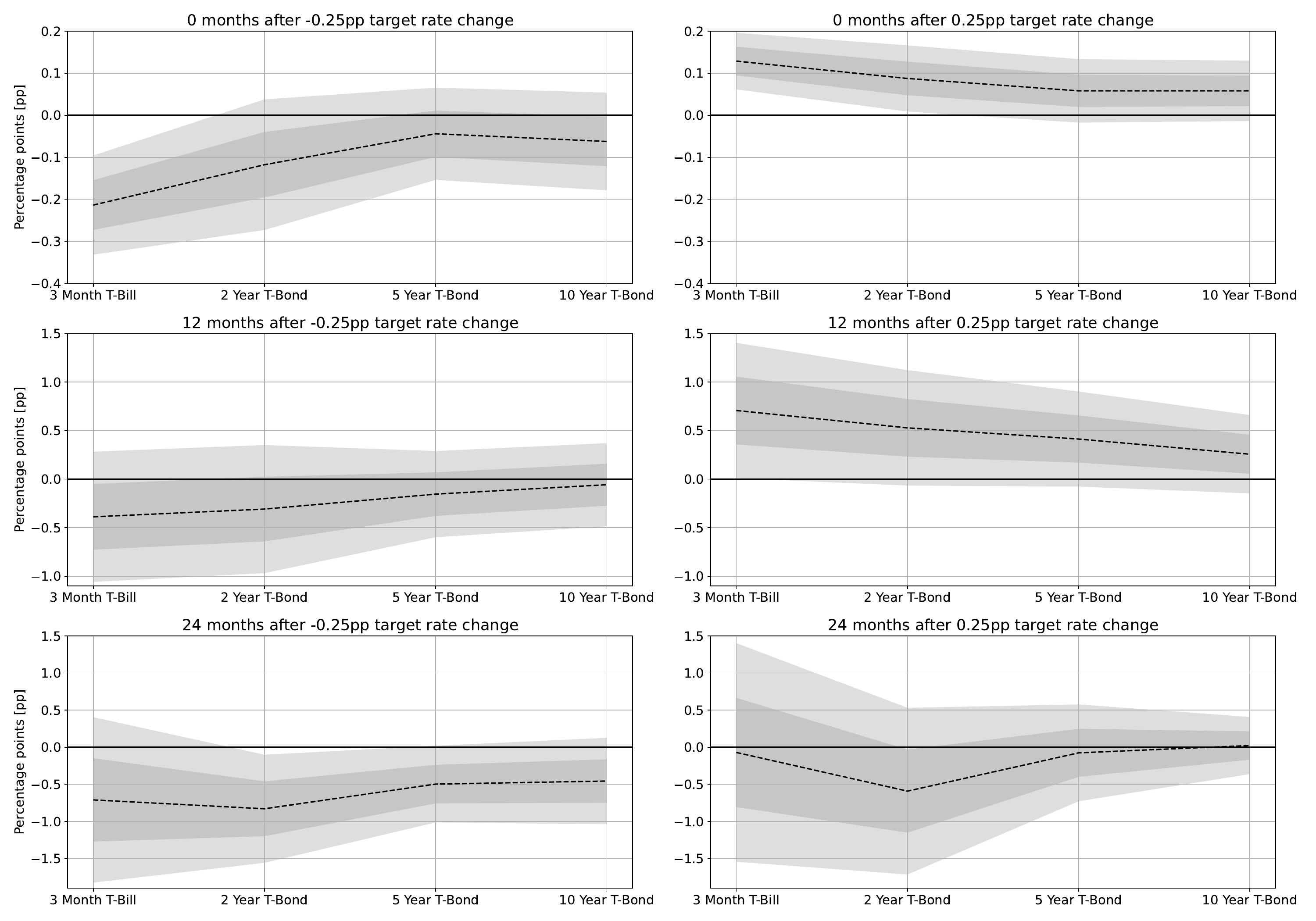}
        \label{fig:IRF_YieldCurve}
        \begin{minipage}{\textwidth}
            \footnotesize\textsc{Note}: The figure shows the estimated cumulative effects of target rate changes on the bond yield curve for the time period July 1989 to December 2008. The left (right) column shows the effect of decreasing (increasing) the target rate by 25 basis points. The estimated effects on the yield curve are depicted for 0 (top row), 12 (middle row), and 24 (bottom row) months after the target rate change. The nuisance functions are estimated by random forest using 10-fold cross-fitting removing $k_T=24$ observations from the estimation sample at the boundary to the inference sample. The shaded areas represent 68\% and 95\% confidence intervals  with fixed-bandwidth critical values \citep{Kiefer2005}. The variances are estimated using bandwidth determined by the procedure of \cite{NeweyWest1994}.
          \end{minipage}
    \end{figure}
    
    
\end{appendix}

\end{document}